\documentclass{siamonline1116}

\usepackage{braket,amsfonts}

\usepackage{array}

\usepackage[caption=false]{subfig}
\usepackage{pgfplots}
\pgfplotsset{compat=1.14}

\newsiamthm{claim}{Claim}
\newsiamremark{rem}{Remark}
\newsiamremark{expl}{Example}
\newsiamremark{hypothesis}{Hypothesis}
\crefname{hypothesis}{Hypothesis}{Hypotheses}

\usepackage{algorithmic}

\usepackage{graphicx,epstopdf}

\Crefname{ALC@unique}{Line}{Lines}

\usepackage{amsopn}

\numberwithin{theorem}{section}

\usepackage{tikz}

\usepackage{xspace}
\usepackage{bold-extra}
\usepackage[most]{tcolorbox}

\colorlet{texcscolor}{blue!50!black}
\colorlet{texemcolor}{red!70!black}
\colorlet{texpreamble}{red!70!black}
\colorlet{codebackground}{black!25!white!25}

\usepackage{amsmath,amssymb, amsfonts}
\usepackage{setspace}
\usepackage[mathscr]{euscript}
\usepackage{graphicx}
\usepackage{color}
\usepackage{cite}
\usepackage{geometry}
\usepackage{booktabs} 
\usepackage{enumerate}
\usepackage{algorithm}
\usepackage{float}
 \usepackage{relsize}

\def\X{{\bf X}}
\def\x{{\bf x}}

\newcommand{\D}{\mathscr D}
\def\E{{\mathbb{E}}}
\def\r{{\mathbb{R}}}

\def\V{{\operatorname{Var}}}

\def\I{{\mathcal I}}

\lstdefinestyle{siamlatex}{%
  style=tcblatex,
  texcsstyle=*\color{texcscolor},
  texcsstyle=[2]\color{texemcolor},
  keywordstyle=[2]\color{texemcolor},
  moretexcs={cref,Cref,maketitle,mathcal,text,headers,email,url},
}

\tcbset{%
  colframe=black!75!white!75,
  coltitle=white,
  colback=codebackground, 
  colbacklower=white, 
  fonttitle=\bfseries,
  arc=0pt,outer arc=0pt,
  top=1pt,bottom=1pt,left=1mm,right=1mm,middle=1mm,boxsep=1mm,
  leftrule=0.3mm,rightrule=0.3mm,toprule=0.3mm,bottomrule=0.3mm,
  listing options={style=siamlatex}
}

\newtcblisting[use counter=example]{example}[2][]{%
  title={Example~\thetcbcounter: #2},#1}

\newtcbinputlisting[use counter=example]{\examplefile}[3][]{%
  title={Example~\thetcbcounter: #2},listing file={#3},#1}

\DeclareTotalTCBox{\code}{ v O{} }
{ 
  fontupper=\ttfamily\color{black},
  nobeforeafter,
  tcbox raise base,
  colback=codebackground,colframe=white,
  top=0pt,bottom=0pt,left=0mm,right=0mm,
  leftrule=0pt,rightrule=0pt,toprule=0mm,bottomrule=0mm,
  boxsep=0.5mm,
  #2}{#1}

\patchcmd\newpage{\vfil}{}{}{}
\begin{tcbverbatimwrite}{tmp_\jobname_header.tex}
\title{Robustness of the Sobol' indices to marginal distribution uncertainty
  \thanks{\funding{This work was supported by the National Science Foundation under grants DMS-1522765 and DMS-1745654.}}}

\author{Joseph Hart
  \thanks{Department of Mathematics, North Carolina State University, Raleigh, NC (\email{jlhart3@ncsu.edu}).}
  \and
  Pierre Gremaud
   \thanks{Department of Mathematics, North Carolina State University, Raleigh, NC (\email{gremaud@ncsu.edu}).}
}

\headers{Robustness of Sobol' indices}
{Joseph Hart and Pierre Gremaud}
\end{tcbverbatimwrite}
\input{tmp_\jobname_header.tex}

\ifpdf
\hypersetup{ pdftitle={Guide to Using  SIAM'S \LaTeX\ Style} }
\fi

\begin{document}
\maketitle

\begin{tcbverbatimwrite}{tmp_\jobname_abstract.tex}
\begin{abstract}
Global sensitivity analysis (GSA) quantifies the influence of uncertain variables in a mathematical model. The Sobol' indices, a commonly used tool in GSA, seek to do this by attributing to each variable its relative contribution to the variance of the model output. In order to compute Sobol' indices, the user must specify a probability distribution for the uncertain variables. This distribution is typically unknown and must be chosen using limited data and/or knowledge. The usefulness of the Sobol' indices depends on their robustness to this distributional uncertainty. This article presents a novel method which uses ``optimal perturbations" of the marginal probability density functions to analyze the robustness of the Sobol' indices. The method is illustrated through synthetic examples and a model for contaminant transport.
\end{abstract}


\end{tcbverbatimwrite}
\input{tmp_\jobname_abstract.tex}

\section{Introduction}
Uncertainties are ubiquitous in mathematical models of complex processes. Understanding how each source of input uncertainty contributes to the model output uncertainty is critical to the development and use of such models. Global sensitivity analysis (GSA) addresses this challenge by considering the inputs of the mathematical model as random variables (or implicitly assuming they are uniformly distributed) and computing a measure of their relative importance in determining the model output. A common tool for this purpose is the Sobol' indices \cite{sobol93,sobol,Sobol_UQ_handbook}. In their classical framework, one assumes that the model inputs $\X = (X_1,X_2,\dots,X_p)$ are independent random variables and $f(\X)$ is some function mapping $\X$ to a scalar valued quantity of interest (QoI), a function of the model output. The Sobol' indices apportion how much of the variance of $f(\X)$ is contributed by each input.

Computing the Sobol' indices requires the user to specify a probability distribution for $\X$, draw samples from this distribution, and evaluate $f$ at these samples. However, the distribution of $\X$ is typically estimated from limited data and/or knowledge and hence is uncertain itself. This article studies the robustness of the Sobol' indices to perturbations in the marginal distributions of $\X$, i.e. the distributions of $X_i$, $i=1,2,\dots,p$. This study is motivated by problems where the only knowledge of $\X$ is a uncertain nominal estimate, for instance a least squares estimate with noisy data, and possibly loose upper and lower bounds. This is common in many applications. For such problems, the user often assumes that each $X_i$ follows a symmetric distribution centered at the nominal estimate with some variance, usually chosen with limited knowledge. An example is centering a uniform distribution at the nominal estimate. The approach proposed in this article may be applied by (i) giving each $X_i$ a distribution with support on a wide interval with greater probability near the nominal estimate, (ii) computing the Sobol' indices using these distributions, (iii) and assessing the Sobol' indices robustness with respect to changes in the distributions. 

In the authors previous work \cite{hart_robustness}, the robustness of Sobol' indices was studied by taking perturbations of the joint probability density function (PDF) of $\X$. This approach was useful for studying robustness with respect to statistical dependencies in $\X$. This article focuses on cases where the inputs are independent and the uncertainty is exclusively in their marginal distributions. We modify the framework introduced in \cite{hart_robustness} to develop a novel method to study the robustness of Sobol' indices with respect to changes in the marginal distributions. It is not a generalization, specialization, nor improvement of \cite{hart_robustness}, but rather a novel method which leverages ideas from \cite{hart_robustness}.


The robustness of Sobol' indices with respect to changes in the marginal distributions of $\X$ has been considered in the life cycle analysis literature \cite{lca_app} and ecology literature \cite{WARM}. ``Robust sensitivity indicators" are defined in \cite{deep_uncertainty} to account for uncertainty in the marginal distributions of $\X$. A generalization of the analysis of variance decomposition is considered in \cite{anova_mult_dist} to facilitate studying the dependence of Sobol' indices on the distribution of $\X$. Other work such as \cite{lca_correlations,hall,cousins} explore the robustness of Sobol' indices to distributional uncertainty and \cite{climate_app,dice,chick,beckman_mckay} explore more general questions of robustness for other computed quantities. In all of these cases, the user is required to either specify the parametric forms of distributions a priori, compute additional evaluations of $f$ beyond those needed to estimate Sobol' indices, or both. The approach proposed in this article requires neither.

Section~\ref{sec:review} reviews the Sobol' indices and sets the notation for the article. We present the mathematical foundations of our proposed marginal PDF perturbation method in Section~\ref{sec:perturbations} and an algorithmic summary is given in Section~\ref{sec:algo}. User guidelines are considered in Section~\ref{sec:compare}. Numerical results are presented in Section~\ref{sec:numerics}. Concluding remarks are made in Section~\ref{sec:conclusion}.

\section{Review of Sobol' indices}
\label{sec:review}
Let $X_i$ be a random variable supported on $\Omega_i \subset \r$, $i=1,2,\dots,p$, $\Omega = \Omega_1 \times \Omega_2 \times \Omega_p$, and $f:\Omega \to \r$. Assume that the inputs $\X=(X_1,X_2,\dots,X_p)$ are independent and $f(\X)$ is square integrable. For $u=\{i_1,i_2,\dots,i_k\}  \subset \{1,2,\dots,p\}$, let $\X_u=(X_{i_1},X_{i_2},\dots,X_{i_k})$ denote the corresponding subset of input variables. Then $f(\X)$ admits the following ANOVA (analysis of variance) decomposition,
\begin{eqnarray}
\label{anova}
f(\X) = f_0 + \sum\limits_{i=1}^p f_i(X_i) + \sum\limits_{i=1}^p \sum\limits_{j = i+1}^p f_{i,j}(X_i,X_j) + \cdots + f_{1,2,\dots,p}(\X)
\end{eqnarray}
where
\begin{align*}
&f_0=\E[f(\X)],\\
&f_i(X_i)=\E[f(\X)|X_i]-f_0, \nonumber \\
&f_{i,j}(X_i,X_j)=\E[f(\X)|X_i,X_j]-f_i(X_i)-f_j(X_j)-f_0, \nonumber \\
\vdots \nonumber \\
&f_u(\X_u) = \E[f(\X)|\X_u]-\sum_{v \subset u} f_v(\X_v). \nonumber
\end{align*}
The independence of the inputs $(X_1,X_2,\dots,X_p)$ implies that 
\begin{eqnarray}
\label{ortho}
\E[f_u(\X_u)f_v(\X_v)]=0 \qquad \text{when} \qquad u \ne v.
\end{eqnarray}
 Computing the variance of both sides of \eqref{anova} and applying \eqref{ortho} yields
\begin{eqnarray}
\V(f(\X)) = \sum\limits_{i=1}^p \V(f_i(X_i)) +  \sum\limits_{i=1}^p \sum\limits_{j = i+1}^p \V(f_{i,j}(X_i,X_j)) + \cdots + \V(f_{1,2,\dots,p}(\X)).
\end{eqnarray}
Hence the variance of $f(\X)$ may be decomposed into contributions from each subset of input variables. The Sobol' index for the variables $\X_u$ is defined as
\begin{eqnarray*}
S_u = \frac{\V(f_u(\X))}{\V(f(\X))},
\end{eqnarray*}
and may be interpreted as the relative contribution of $\X_u$ to the variance of $f(\X)$. In practice, one frequently computes the \textit{first order Sobol' indices} $\{S_k\}_{k=1}^p$, and the \textit{total Sobol' indices} $\{T_k\}_{k=1}^p$ defined by
\begin{eqnarray*}
T_k = \sum\limits_{k \in u} S_u, \qquad k=1,2,\dots,p.
\end{eqnarray*}
It easily observed that $0 \le S_k \le T_k \le 1$. The first order Sobol' index $S_k$ may be interpreted as the relative contribution of variable $X_k$ to the variance of $f(\X)$ by itself and the total Sobol' index $T_k$ may be interpreted as the relative contribution of $X_k$ to the variance of $f(\X)$ by itself and through interactions with other variables. The total Sobol' index may also be defined for a subset of variables $\X_u$ in an analogous fashion. This article will focus on the indices $\{S_k,T_k\}_{k=1}^p$; though the results are easily extended to general indices $S_u$ or $T_u$, $u \subset \{1,2,\dots,p\}$. 

There are a plurality of ways to estimate the first order and total Sobol' indices \cite{Sobol_UQ_handbook}. We focus on the Monte Carlo integration scheme presented in \cite{Sobol_UQ_handbook}. The indices $S_k$ and $T_k$ may be written probabilistically as
\begin{eqnarray*}
S_k = \frac{\V(\E[f(\X) \vert X_k ])}{\V(f(\X))} \qquad \text{and} \qquad T_k=\frac{\E[\V(f(\X)\vert \X_{\sim k} )]}{\V(f(\X))} .
\end{eqnarray*}
The properties of conditional expectation and Fubini's theorem yield
\begin{eqnarray*}
\V(\E[f(\X) \vert X_k ]) = \E[f(\X) (f(\X_k',\X_{\sim k})-f(\X'))],
\end{eqnarray*}
and
\begin{eqnarray*}
\E[\V(f(\X)\vert \X_{\sim k} )] = \frac{1}{2} \E[(f(\X)-f(X_k',\X_{\sim k}))^2],
\end{eqnarray*}
where $\X'$ is an independent copy of $\X$.

The indices $\{S_k,T_k\}_{k=1}^p$ may be estimated by Monte Carlo integration using $(p+2)N$ evaluations of $f$, where $N$ is the number of Monte Carlo samples. Specifically, the samples are generated by constructing 2 matrices of size $N \times p$, call them $A$ and $B$, where the rows correspond to independent samples of $\X$. Then $p$ additional matrices of size $N \times p$, denote them as $C_k$, $k=1,2,\dots,p$, are generated by setting the $i^{th}$ column of $C_k$ equal to the $i^{th}$ column of $A$, $i \ne k$,  and the $k^{th}$ column of $C_k$ equal to the $k^{th}$ column of $B$. The set of indices $\{S_k,T_k\}_{k=1}^p$ may be estimated using evaluations of $f$ at the rows of $A$, $B$, and $C_k$, $k=1,2,\dots,p$. This scheme is presented in Lines 1-6 of Algorithm~\ref{alg:robustness}.

The reader is directed to \cite{sobol93,sobol,iooss,hart_corr_var,Sobol_UQ_handbook,sobol2003,saltelli2010,kucherenko,fast_dependent_variables}, and references therein, for additional discussion about the Sobol' indices.

\section{Marginal PDF Perturbations}
\label{sec:perturbations}
We study the robustness of the Sobol' indices to distributional uncertainty by computing a derivative of the Sobol' index with respect to the PDFs of $X_1,X_2,\dots,X_p$. To this end we make the following assumptions:
\begin{enumerate}
\item $\Omega_i$ is a compact interval, $i=1,2,\dots,p$,
\item $X_i$ admits a PDF $\phi_i$, $i=1,2,\dots,p$,
\item $\phi_i$ is continuous on $\Omega_i$, $i=1,2,\dots,p$,
\item $\phi_i(x_i) > 0$ $\forall x_i \in \Omega_i$, $i=1,2,\dots,p$,
\item $f$ is bounded on $\Omega$.
\end{enumerate}
Without loss of generality, under the assumptions above, assume that $\Omega = [0,1]^p$. Distributions with unbounded support may be truncated to satisfy these assumptions, hence our method is applicable on a large class of problems.

In order to perturb the PDFs $\phi_1,\phi_2,\dots,\phi_p$ while preserving the non-negativity and unit mass properties of PDFs, we define the Banach spaces $V_i$, $i=1,2,\dots,p$, as all bounded functions on $\Omega_i$ with norm
\begin{eqnarray*}
\Big\vert \Big \vert \psi_i  \Big\vert \Big\vert_{V_i} = \Big\vert \Big\vert \frac{\psi_i}{\phi_i} \Big\vert \Big \vert_{L^\infty(\Omega_i)},
\end{eqnarray*}
where the $L^\infty(\Omega_i)$ norm is the traditional supremum norm. The product Banach space $V = V_1 \times V_2 \times \cdots \times V_p$ is equipped with norm
\begin{eqnarray*}
\Big\vert \Big \vert (\psi_1,\psi_2,\dots,\psi_p)  \Big\vert \Big\vert_V = \max\limits_{1 \le i \le p} \Big\vert \Big\vert \psi_i \Big\vert \Big \vert_{V_i}.
\end{eqnarray*}
Then for $\psi=(\psi_1,\psi_2,\dots,\psi_p) \in V$ with $\vert \vert \psi-\phi \vert \vert_V \le 1$, $\phi=(\phi_1,\phi_2,\dots,\phi_p)$, the functions
\begin{eqnarray*}
\frac{\phi_i+\psi_i}{1+ \int_{\Omega_i} \psi_i(x_i)dx_i}
\end{eqnarray*}
are PDFs supported on $\Omega_i$, $i=1,2,\dots,p$, i.e. are non-negative and integrate to 1 over $\Omega_i$.

For $\eta = (\eta_1,\eta_2,\dots,\eta_p) \in V$, define the first order and total Sobol' indices as the operators $S_k,T_k:V \to \r$, $k=1,2,\dots,p$,
\begin{eqnarray*}
S_k(\eta) = \frac{F_k(\eta)}{H_k(\eta)} \qquad \text{and} \qquad T_k(\eta) = \frac{G_k(\eta)}{H_k(\eta)},
\end{eqnarray*}
where $F_k,G_k,H_k:V \to \r$ are the operators
\begin{eqnarray*}
F_k(\eta) =  \frac{1}{\prod\limits_{i=1}^p \int_{\Omega_i} \eta_i(y)dy}  \int_{\Omega \times \Omega} f(\x)\left(f(x_k',\x_{\sim k})-f(\x')\right) \prod\limits_{i=1}^p \eta_i(x_i)  \eta_i(x_i')  d\x d\x',
\end{eqnarray*}
\begin{eqnarray*}
G_k(\eta) = \frac{1}{2} \frac{1}{\int_{\Omega_k} \eta_k(y)dy}  \int_{\Omega \times \Omega_k} \left(f(\x)-f(\x')\right)^2 \prod\limits_{i=1}^p \eta_i(x_i) \eta_k(x_k')  d\x d\x_k',
\end{eqnarray*}
 \begin{eqnarray*}
H_k(\eta) = \int_{\Omega} f(\x)^2 \prod\limits_{i=1}^p \eta_i(x_i) d\x - \frac{1}{\prod\limits_{i=1}^p \int_{\Omega_i} \eta_i(y) dy} \left( \int_{\Omega} f(\x) \prod\limits_{i=1}^p \eta_i(x_i) d\x \right)^2.
\end{eqnarray*}
For $\eta \in V$ with $\vert \vert \eta-\phi \vert \vert_V \le 1$, $S_k(\eta)$ and $T_k(\eta)$, $k=1,2,\dots,p$, are the first order and total Sobol' indices computed when $\X$ has the joint PDF 
\begin{eqnarray*}
\frac{\prod\limits_{i=1}^p \eta_i}{\prod\limits_{i=1}^p \int_{\Omega_i} \eta_i(y)dy}.
\end{eqnarray*}
Theorem \ref{thm:frechet} gives the Fr\'echet derivative of the Sobol' index with respect to the PDFs.
\begin{theorem}
\label{thm:frechet}
The operators $S_k$ and $T_k$, $k=1,2,\dots,p$, are Fr\'echet differentiable at the nominal PDF $\phi=(\phi_1,\phi_2,\dots,\phi_p)$ and the Fr\'echet derivatives $\D S_k(\phi): V \to \r$ and $\D T_k(\phi): V \to \r$, $k=1,2,\dots,p$, are the bounded linear operators given by
\begin{eqnarray*}
\D S_k(\phi) \psi = \frac{\D F_k(\phi) \psi}{H_k(\phi)} - S_k(\phi) \frac{\D H_k(\phi) \psi}{H_k(\phi)}, \qquad \D T_k(\phi) \psi = \frac{\D G_k(\phi) \psi}{H_k(\phi)} - T_k(\phi) \frac{\D H_k(\phi) \psi}{H_k(\phi)},
\end{eqnarray*}
where
\begin{align*}
&\D F_k(\phi)\psi = \int_{\Omega \times \Omega} f(\x)(f(x_k',\x_{\sim k})-f(\x'))  \left( \sum\limits_{i=1}^p \frac{\psi_i(x_i)}{\phi_i(x_i)}+\frac{\psi_i(x_i')}{\phi_i(x_i')} \right) \prod\limits_{i=1}^p \phi_i(x_i)\phi_i(x_i') d\x d\x' \\
&- \left( \int_\Omega \left( \sum\limits_{i=1}^p \frac{\psi_i(x_i)}{\phi_i(x_i)} \right) \prod\limits_{i=1}^p \phi_i(x_i) d\x  \right)  \int_{\Omega \times \Omega}  f(\x)(f(x_k',\x_{\sim k})-f(\x'))  \prod\limits_{i=1}^p \phi_i(x_i) \phi_i(x_i')  d\x d\x', \\
\end{align*}
\begin{align*}
\D G_k(\phi)\psi =& \frac{1}{2}  \int_{\Omega \times \Omega_k} (f(\x)-f(\x'))^2  \left( \frac{\psi_k(x_k')}{\phi_k(x_k')} + \sum\limits_{i=1}^p \frac{\psi_i(x_i)}{\phi_i(x_i)} \right) \prod\limits_{i=1}^p \phi_i(x_i) \phi_k(x_k') d\x d\x_k' \\
-& \frac{1}{2} \int_{\Omega_k} \psi_k(y)dy \int_{\Omega \times \Omega_k} (f(\x)-f(\x'))^2 \prod\limits_{i=1}^p \phi_i(x_i) \phi_k(x_k') d\x d\x_k', \\
\end{align*}
\begin{align*}
\D H_k(\phi)\psi = & \int_\Omega f(\x)^2 \left(\sum\limits_{i=1}^p \frac{\psi_i(x_i)}{\phi_i(x_i)} \right) \prod\limits_{i=1}^p \phi_i(x_i) d\x \\
& +\left( \int_\Omega \left( \sum\limits_{i=1}^p \frac{\psi_i(x_i)}{\phi_i(x_i)} \right) \prod\limits_{i=1}^p \phi_i(x_i) d\x  \right) \left( \int_{\Omega} f(\x) \prod\limits_{i=1}^p \phi_i(x_i) d\x \right)^2\\
& - 2  \left( \int_{\Omega} f(\x) \prod\limits_{i=1}^p \phi_i(x_i) d\x \right) \left( \int_{\Omega} f(\x) \left( \sum\limits_{i=1}^p \frac{\psi_i(x_i)}{\phi_i(x_i)} \right) \prod\limits_{i=1}^p \phi_i(x_i) d\x \right).
\end{align*}
\end{theorem}

\begin{proof}
One may easily observe that $H_k(\eta)>0$ in a neighborhood of $\phi$ (assuming $f(\x)$ is not constant). It is sufficient to compute the Fr\'echet derivatives of $F_k,G_k$, and $H_k$ as the Fr\'echet derivatives of $S_k$ and $T_k$ will follow from the quotient rule for differentiation. The Fr\'echet derivative of the operator
\begin{eqnarray*}
 (\phi_1,\phi_2,\dots,\phi_p) \mapsto \int_{\Omega} \prod\limits_{i=1}^p \phi_i(x_i) d\x,
\end{eqnarray*}
mapping from $V$ to $\r$, acting on $\psi = (\psi_1,\psi_2,\dots,\psi_p) \in V$, is easily shown to be
\begin{eqnarray*}
\int_{\Omega} \sum\limits_{j=1}^p \frac{\psi_j(x_j)}{\phi_j(x_j)} \prod\limits_{i=1}^p \phi_i(x_i) d\x.
\end{eqnarray*}
Similar results hold for the Fr\'echet derivatives of 
\begin{eqnarray*}
 \int_{\Omega} f(\x) \prod\limits_{i=1}^p \phi_i(x_i) d\x \qquad \text{and} \qquad  \int_{\Omega} f(\x)^2 \prod\limits_{i=1}^p \phi_i(x_i) d\x .
\end{eqnarray*}
The Fr\'echet derivative of $H_k$ follows from the sum/difference, product, and chain rules for differentiation. Similar arguments may be used to compute Fr\'echet derivatives of $F_k$ and $G_k$.
\end{proof}

Theorem~\ref{thm:frechet} can be generalized to Sobol' indices $S_u$ or $T_u$ for any $u \subset \{1,2,\dots,p\}$; we omit this generalization for simplicity. If the Sobol' indices are estimated with the Monte Carlo integration scheme presented in Section~\ref{sec:review}, then the action of the Fr\'echet derivatives on any $\psi \in V$ may be estimated using the existing evaluations of $f$. Hence the additional computation needed to estimate the Fr\'echet derivative is negligible.

We would like determine an optimal perturbation of the marginal PDFs which yields the greatest change in the Sobol' indices. For notational simplicity, let $\I$ denote a generic Sobol' index, which may be $S_k$ or $T_k$ for any $k=1,2,\dots,p$. The locally optimal perturbation direction is given by the $\psi \in V$, $\vert \vert \psi \vert \vert \le 1$, which maximizes the absolute value of the Fr\'echet derivative of $\I$. To estimate it, we discretize $V$ by partitioning each interval $\Omega_i=[0,1]$ into $M_i$ disjoint subintervals, $\{R_i^j\}_{j=1}^{M_i}$, i.e. $\Omega_i = \cup_{j=1}^{M_i} R_i^j$, $i=1,2,\dots,p$, and define basis functions $ \psi_i^j:\Omega \to \{0,1\}^p$ as
\begin{eqnarray*}
 \psi_i^j(\X)=\chi_{R_i^j}(X_i)e_i, \qquad j=1,2,\dots,M_i, i=1,2,\dots,p
 \end{eqnarray*}
where $\chi$ is the indicator function of a set and $e_i$ is the $i^{th}$ canonical unit vector in $\r^p$. Then
\begin{eqnarray*}
V_M=\text{span}\{ \psi_i^j \vert j=1,2,\dots,M_i,i=1,2,\dots,p \} \subset V,
\end{eqnarray*}
is a subspace of $V$ on which we will perform our robustness analysis. Discretizing $V$ with piecewise constant functions may be viewed as approximating all bounded functions by discretizing them on a coarse grid which is constrained by the existing samples and evaluations of $f$.

Then the locally optimal perturbation is solution of the optimization problem
\begin{eqnarray}
\label{opt_prob}
\max\limits_{\substack{\psi \in V_M\\ \vert \vert \psi \vert \vert_{V}=1}} \vert \D \I(\phi)\psi \vert.
\end{eqnarray}
Exploiting properties of the basis we chose for $V_M$, \eqref{opt_prob} may be solved in closed form yielding the optimal perturbation direction
\begin{eqnarray*}
\psi^\star = (\psi_1^\star,\psi_2^\star,\dots,\psi_p^\star) = \sum\limits_{i=1}^p \sum\limits_{j=1}^M a_i^j \psi_i^j
\end{eqnarray*}
where 
\begin{eqnarray*}
a_i^j = \left( \inf_{x_i \in R_j} \phi_i(x_i) \right) \text{sign}\left(\D \I (\phi)\psi_i^j \right) .
\end{eqnarray*}
Because of the norm on $V$, this formulation gives equal weight to each marginal distribution. In actuality, we would like to focus our effort on taking perturbations on the ``most important" marginals. To this end, we define the weights
\begin{eqnarray*}
\delta_i = \frac{\sum\limits_{j=1}^{M_i} \vert \D \I (\phi)\psi_i^j  \vert}{\sum\limits_{i=1}^p \sum\limits_{j=1}^{M_i} \vert \D \I (\phi)\psi_i^j  \vert }
\end{eqnarray*}
and define the perturbed PDF as
\begin{eqnarray}
\label{eqn:perturbed_pdf}
\prod\limits_{i=1}^p \frac{\phi_i(x_i) + \delta \delta_i \psi_i^\star}{1+\delta \delta_i \int_{\Omega_i} \psi_i^\star(x_i)dx_i},
\end{eqnarray}
where $\delta \in [-\overline{\delta},\overline{\delta}]$, $\overline{\delta}=\left( \max_{i=1,2,\dots,p} \delta_i \right)^{-1}$; the determination of $\delta$ will be discussed in Section~\ref{sec:algo}. The perturbed PDF depends upon the partitions $\{R_i^j\}_{j=1}^{M_i}$, $i=1,2,\dots,p$, and perturbation size $\delta$, which depends on user specified parameters, see Section~\ref{sec:algo}.

We will refer to the distribution of $\X$ as the nominal distribution and the distribution with PDF \eqref{eqn:perturbed_pdf} as the perturbed distribution. Likewise, we refer to the Sobol' indices computed using the nominal distribution as nominal Sobol' indices. We compute \textit{perturbed Sobol' indices}, where the inputs follow the perturbed distribution, by reweighing the evaluations of $f$ used to estimate the nominal Sobol' indices. Hence the computational effort to compute perturbed Sobol' indices is negligible. 

\section{Algorithm}
\label{sec:algo}
Evaluating the perturbed PDF requires estimation of the Sobol' index Fr\'echet derivative to compute $\psi^\star$, and the determination of $\delta$. As previously discussed, the Sobol' index Fr\'echet derivative is easily estimated using existing evaluations of $f$. The scalar $\delta$ may be determined by finding the largest and smallest possible values in $[-\overline{\delta},\overline{\delta}]$ for which the perturbed index estimator's sample standard deviation is sufficiently small. Specifically, a threshold should be defined by the user and a value of $\delta$ is accepted if the ratio of sample standard deviations for the perturbed indices and nominal indices is below the threshold. 

Algorithm~\ref{alg:robustness} provides an overview of the Sobol' index computation with marginal distribution robustness post processing. The algorithm requires four inputs:
\begin{enumerate}
\item[$\bullet$] $n$, the number of Monte Carlo samples
\item[$\bullet$] $\{R_i^j\}_{j=1}^{M_i}$, a partition of $\Omega_i$, $i=1,2,\dots,p$
\item[$\bullet$] $\tau$, the threshold for admissible increases in the Sobol' index estimator's sample standard deviation
\item[$\bullet$] $r$, an integer used to compute admissible value of $\delta$
\end{enumerate}

The determination of $n$ is problem dependent and must be done by the user to compute Sobol' indices.

The partitions $\{R_i^j\}_{j=1}^{M_i}$, $i=1,2,\dots,p$, may be a chosen in many ways. A general partitioning strategy is to specific a number of subintervals $M$ and partition each $\Omega_i$ into $M$ subintervals according to the quantiles of the marginal distribution of $X_i$. Alternatively, the user may use problem specific information to customize the partition, for instance, taking a finer partition in regions of greater uncertainty or adapting $M_i$ based on the uncertainty in each marginal.

The threshold $\tau$ bounds the increase in estimation error introduced by reweighing samples to compute perturbed Sobol' indices. If the original Sobol' index estimation has small estimator standard deviation, then $\tau=1.5$ is considered a reasonable threshold to permit nontrivial perturbations while preserving reasonable accuracy of the estimator.

To determine admissible values of $\delta$, we consider a function $\Delta:[-\overline{\delta},\overline{\delta}] \to \r$ which inputs a perturbation size $\delta$ and returns the maximum (over all $p$ indices) ratio of sample standard deviations of the perturbed and nominal Sobol' index estimators. We evaluate $\Delta$ by subsampling the data and computing a collection of estimators and their standard deviations. Finding admissible values of $\delta$ is equivalent to finding solutions of $\Delta(\delta) \le \tau$. Intuition (and direct calculations) indicates that $\Delta$ is approximately a quadratic function centered at $\delta=0$. Admissible values of $\delta$ are determined by evaluating $\Delta$ at $r$ equally spaced points in $[-\overline{\delta},\overline{\delta}]$ and accepting those for which $\Delta \le \tau$. Since it is discretizing an interval, $r$ does not need to be very large, $r=60$ is suggested as a reasonable choice to balance the computational cost (the for loop in Line~14 of Algorithm~\ref{alg:robustness}) and the necessary discretization to determine admissible values of $\delta$. The evaluations of $\Delta$ at $r$ points may be plotted to confirm that $r=60$ is sufficient.

Our choices for the inputs described above are based on numerical experimentation, intuition, and experience from the authors previous work in \cite{hart_robustness}.

 \begin{algorithm}
\caption{Computation of Sobol' indices with marginal distribution robustness post processing} \label{alg:robustness}
\begin{algorithmic}[1]
\STATE \textbf{Input: } $n$, $\{ \{R_i^j\}_{j=1}^{M_i} \}_{i=1}^p$, $\tau$, $r$
\STATE Draw $n$ samples of $\X$, store them in $A \in \r^{n \times p}$
\STATE Draw $n$ samples of $\X$, store them in $B \in \r^{n \times p}$
\STATE Construct $p$ matrices $C_k \in \r^{n \times p}$, $k=1,2,\dots,p$, where $C_k(:,i)=A(:,i)$, $i\ne k$, and $C_k(:,k)=B(:,k)$
\STATE Evaluate $f(A),f(B),f(C_k)$, $k=1,2,\dots,p$
\STATE Compute $S_k$ and $T_k$, $k=1,2,\dots,p$
\STATE Evaluate $\phi_k(A(:,k))$ and $\phi_k(B(:,k))$, $k=1,2,\dots,p$
\STATE Determine $\inf_{x_i \in R_i^j} \phi_i(x_i)$, $j=1,2,\dots,M_i$, $i=1,2,\dots,p$
\STATE Compute $\D S_k(\phi) \psi_i^j$, $k=1,2,\dots,p$, $j=1,2,\dots,M_i$, $i=1,2,\dots,p$
\STATE Compute $\D T_k(\phi)\psi_i^j$, $k=1,2,\dots,p$, $j=1,2,\dots,M_i$, $i=1,2,\dots,p$
\FOR{$k$ from 1 to $p$}
\STATE Determine $\psi^{(k,1)} \in V_M$, $\vert \vert \psi^{(k,1)}\vert \vert_V \le 1$, which maximizes $\vert \D S_k(\phi)\psi^{(k,1)}  \vert$
\STATE Determine $\psi^{(k,2)} \in V_M$, $\vert \vert \psi^{(k,2)} \vert \vert_V \le 1$, which maximizes $\vert \D T_k(\phi)\psi^{(k,2)}  \vert$
\FOR{$\ell$ from 0 to $r$}
\STATE Compute $\{\tilde{S}_i^{(k,\ell,1)},\tilde{T}_i^{(k,\ell,1)}\}_{i=1}^p$, $\Delta^{(k,\ell,1)}$ with perturbation $\psi^{(k,1)}$, $\delta=\overline{\delta} \left(-1+\frac{2\ell}{r} \right)$
\STATE Compute $\{\tilde{S}_i^{(k,\ell,2)},\tilde{T}_i^{(k,\ell,2)}\}_{i=1}^p$, $\Delta^{(k,\ell,2)}$ with perturbation $\psi^{(k,2)}$, $\delta=\overline{\delta} \left(-1+\frac{2\ell}{r} \right)$
\ENDFOR
\ENDFOR
\STATE \textbf{Output: } sets of perturbed first order and total Sobol' indices with $\Delta^{(k,\ell,I)} \le \tau$
\STATE Matlab notation is used where $A(:,k)$ denotes the $k^{th}$ column of matrix $A$ and $f(A)$ denotes the vector generated by evaluating $f$ at each row of the matrix $A$.
\end{algorithmic} 
\end{algorithm}

Algorithm~\ref{alg:robustness} returns the nominal Sobol' indices and a collection of perturbed Sobol' indices. The user may query these perturbed Sobol' indices and visualize measures of robustness alongside the indices themselves. In addition, samples may be drawn from the perturbed distributions. The ratios of the nominal PDF and perturbed PDFs are computed (though not explicitly written) in Algorithm~\ref{alg:robustness}. Using candidate samples from the nominal distribution and the ratio of PDFs provides an easy way to generate samples from the perturbed distributions with rejection sampling. Visualizing the perturbed marginal distributions provides additional insights into the regions of parameter space which cause a lack of robustness. The examples in Section~\ref{sec:numerics} illustrate efficient ways to visualize the output of Algorithm~\ref{alg:robustness} and make inferences with it.

\section{User guidelines}
\label{sec:compare}
The method presented in Section~\ref{sec:perturbations} only considers perturbations of the one dimensional marginals PDFs, we refer to this as the marginal perturbation approach. In the authors previous work \cite{hart_robustness}, perturbation were taken on the joint PDF, we refer to this as the joint perturbation approach. Other forms of perturbations may be considered as well. For instance, assuming a block dependency structure motivates a form of PDF perturbations on the marginal PDFs for each block. This approach would be a hybrid between the marginal perturbation approach of this article and joint perturbation approach of \cite{hart_robustness}. In this section, we compare and contrast important features of the marginal perturbation and joint perturbation approaches to provide some principles for a user to consider when deciding which method is appropriate for a given application. Since the two approaches seek to solve different problems, the user must be cognizant of their differences when determining which to use in an application.

If the inputs are assumed to be independent for lack of better knowledge, but dependencies are suspected, then the joint perturbation approach is more appropriate. If the inputs are believed to be independent, the joint perturbation approach may give pessimistic results by introducing dependencies which are not realistic, hence the marginal perturbation approach is superior. 

In the joint perturbation approach, the user must generate a partition of $\Omega$ which ensures a sufficient number of samples in each subset. If $p$ (the number of uncertain inputs) is large then the joint perturbation approach yields a partition which is coarse in each input domain $\Omega_i$. This problem is amplified if the nominal distribution of $\X$ is far from being uniform. In contrast, the partitioning is trivial in the marginal perturbation approach because it only requires partitioning each $\Omega_i$, $i=1,2,\dots,p$, separately. Hence it allows a richer class of perturbations in the marginal distributions. In addition, the user may insert a priori knowledge about each marginal distribution separately, for instance, in the partitions input to Algorithm~\ref{alg:robustness}.

In the marginal perturbation approach, each marginal is weighted by $\delta_i$, a measure of the Fr\'echet derivative's magnitude. This weighting strategy manifests itself differently in the joint perturbation approach, where the input variables are prioritized by the partitioning algorithm. The optimality of the weighting scheme is unclear in both cases; however, the simplicity of the weighting scheme for the marginal perturbation approach makes it more appealing.

The user must prioritize whether there is greater uncertainty in the marginal distributions or the dependency structure. The number of uncertain inputs and characteristics of their nominal distribution influence the performance of the two approaches and must also be considered. Subsection~\ref{sec:numerics_comparison} gives a numerical illustration comparing the two approaches.

\section{Numerical results}
\label{sec:numerics}
Three examples are given below. The first example compares the marginal perturbation approach of this article with the joint perturbation approach of \cite{hart_robustness}. The second is a synthetic example serving to illustrate our marginal perturbation method. The third is an application of robustness analysis to an advection diffusion contaminant transport problem. For simplicity and conciseness, we only report results for the total Sobol' indices.

\subsection{Comparison of marginal perturbation and joint perturbation approaches}
\label{sec:numerics_comparison}
Let
\begin{eqnarray}
\label{linear_fun}
f(\X) = \sum\limits_{i=1}^{10} a_i X_i,
\end{eqnarray}
where $a_i = 11-i$, $i=1,2,\dots,10$. We consider two cases, first when $X_i$ has a uniform distribution on the interval $[0,1]$, $i=1,2,\dots,10$, and second, when $X_i$ has a truncated normal distribution on the interval $[0,1]$, $i=1,2,\dots,10$. In the truncated normal case, the mean of $X_i$ is $0.5$ and its standard deviation is $0.73-0.04i$, $i=1,2,\dots,10$. 

The marginal perturbation and joint perturbation approaches are compared by applying them to analyze the robustness of the total Sobol' indices of \eqref{linear_fun}, computed using $5,000$ Monte Carlo samples, for each of the two cases. Figure~\ref{fig:linear_fun} displays the nominal, maximum, and minimum Total Sobol' indices, where the maximum and minimum are determined by searching all perturbed total Sobol' indices with $\Delta \le \tau$. The top of the green bar and bottom of the magenta bar denote the maximum and minimum total Sobol' index, respectively, as determined by the marginal perturbation approach; the blue region between them denotes the nominal total Sobol' index. The top of the orange bar and bottom of the yellow bar denote the maximum and minimum total Sobol' index, respectively, as determined by the joint perturbation approach; likewise the blue region between them denotes the nominal total Sobol' index. The left panel corresponds to the uniform marginals case and the right panel corresponds to the truncated normal marginals case.

In both cases, the marginal perturbation approach found larger perturbation of the total Sobol' indices than the joint perturbation approach, i.e. the marginal perturbation approach discovered a greater lack of robustness than the joint perturbation approach. This is because it was able to take a finer partition of $\Omega$ and hence find ``better" perturbations. 

The total Sobol' indices in truncated normal marginals case are more robust than in the uniform marginals case. This occurs because the largest perturbations of $T_i$ corresponds to putting greater probability near the boundaries, i.e. for $X_i$ near 0 and 1. The truncated normal marginals case puts less probability near the boundaries so perturbations of this nominal PDF will not change the total Sobol' indices as much as in the uniform marginals case. Also notice that the joint perturbation approach yields very small perturbation for the truncated marginals case. This is because there are low probability tails in several marginal distributions which hampers the partitioning algorithm. The marginal perturbation approach does not suffer this limitation. 

Using the marginal perturbation approach, the maximum and minimum perturbed total Sobol' indices are approximately symmetric about the nominal total Sobol' indices. In the joint perturbation approach, statistical dependencies are introduced which result in a decrease of the magnitude of the total Sobol' indices, see \cite{hart_corr_var,hart_robustness}. As a result, the joint perturbation approach yields a skewness where the largest perturbations correspond to decreasing the total Sobol' indices.

\begin{figure}[h]
\centering
\includegraphics[width=.49 \textwidth]{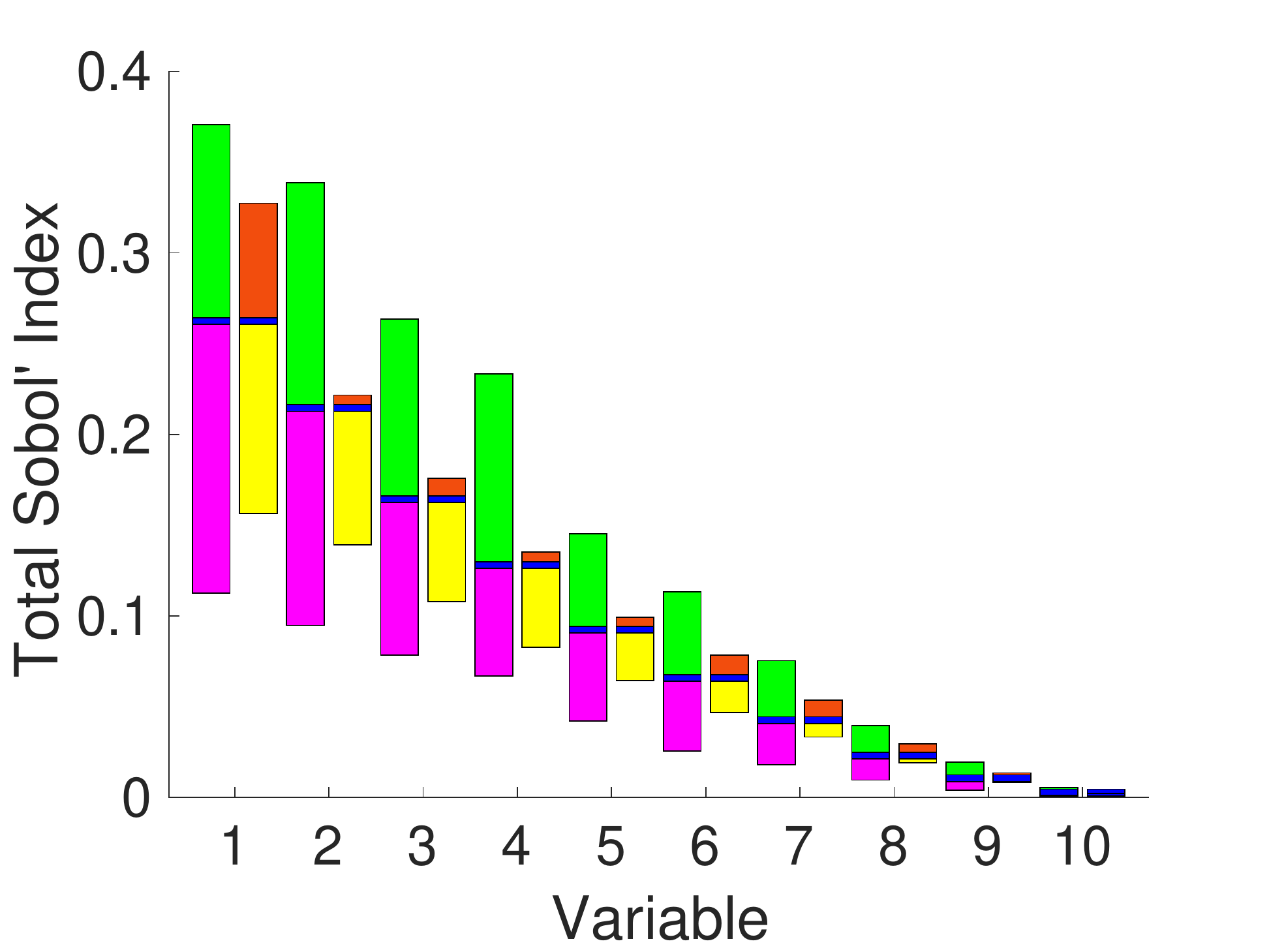}
\includegraphics[width=.49 \textwidth]{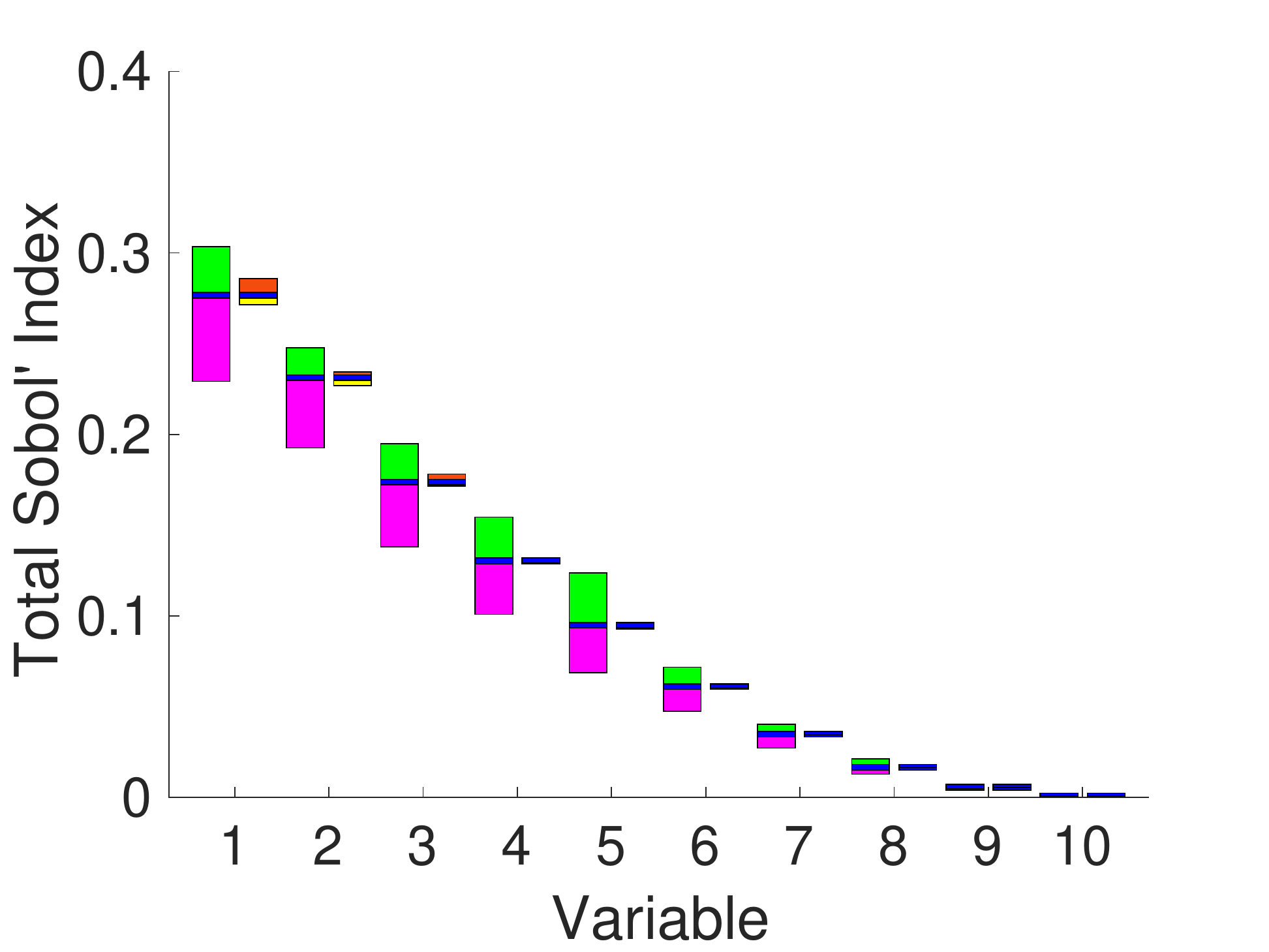}
\caption{Comparison of the marginal perturbation and joint perturbation approaches for \eqref{linear_fun}; the left and right panels correspond to the uniform and truncated normal marginals cases, respectively. The maximum and minimum total Sobol' indices, as determined by the marginal perturbation approach, are given by the top of the green bars and the bottom of the magenta bars, respectively. The maximum and minimum total Sobol' indices, as determined by the joint perturbation approach, are given by the top of the orange bars and the bottom of the yellow bars, respectively. The blue area corresponds to the nominal total Sobol' indices.}
\label{fig:linear_fun}
\end{figure}

\subsection{Synthetic example}

To emulate a problem where there is uncertainty in the support of $\X$, let $X_i$ be uniformly distributed on the interval $[A_i,B_i]$, where $A_i$ and $B_i$ are themselves random variables which are uniformly distributed on $[0,.1]$ and $[.9,1]$, respectively, $i=1,2,3$. Consider $f:[0,1]^3 \to \r$ defined by
\begin{eqnarray}
\label{synthetic_ex_function}
f(\X) = 2X_2e^{-2X_1}+X_3^2.
\end{eqnarray}

It is apparent from the structure of \eqref{synthetic_ex_function} that $f$ will vary more in the $X_2$ direction when $X_1$ is close to 0 and less when $X_1$ is close to 1. We study the sensitivity of $f$ to $X_2$, and its robustness with respect to changes in the marginal distributions. Figure~\ref{fig:synthetic_ex_function_total_indices} displays the total Sobol' indices of \eqref{synthetic_ex_function} in blue, along with the total Sobol' indices under perturbations of the marginal distributions which maximize $T_2$ in red and minimize $T_2$ in cyan.

\begin{figure}[h]
\centering
\includegraphics[width=.65 \textwidth]{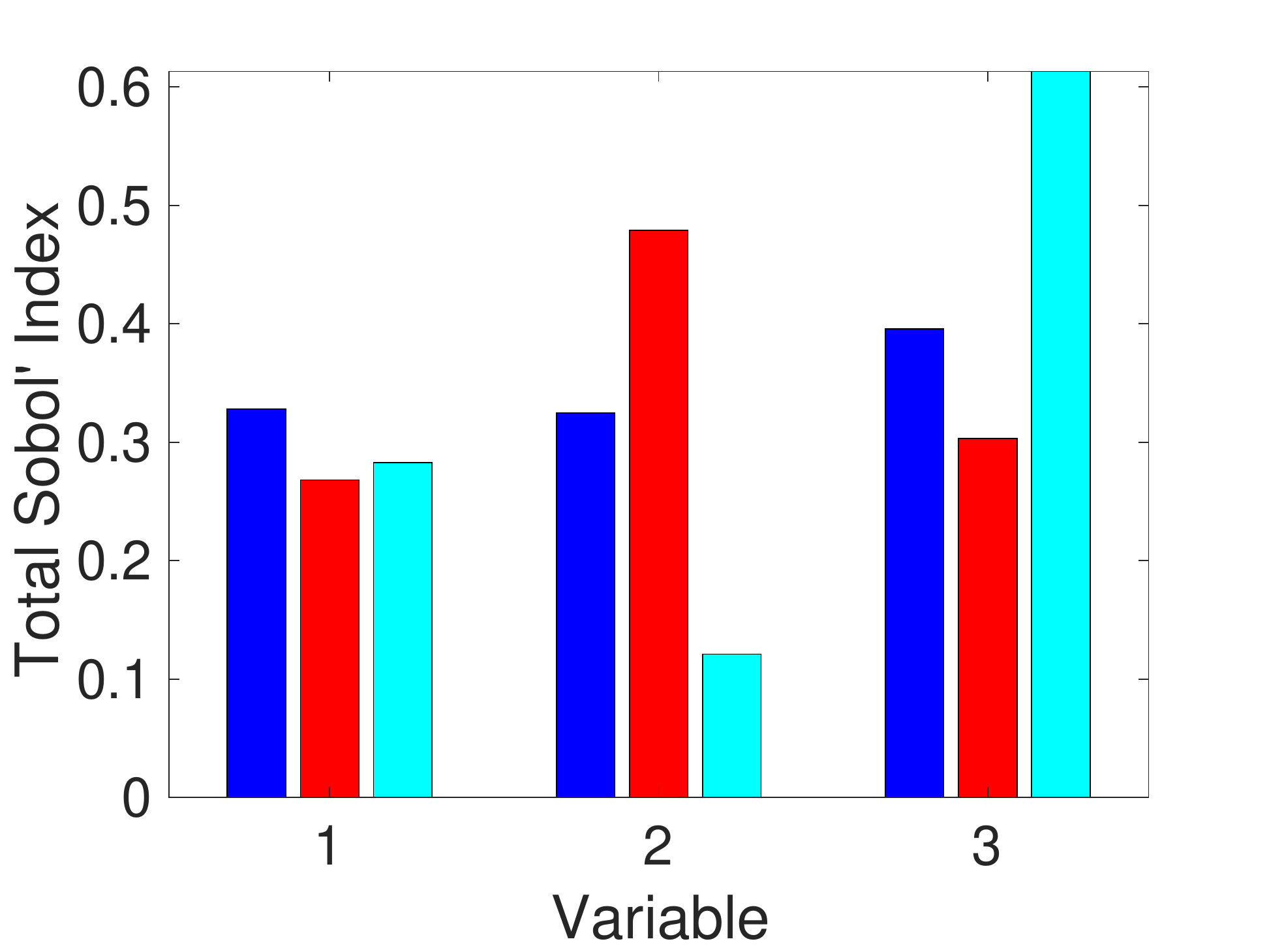}
\caption{Total Sobol' indices of \eqref{synthetic_ex_function}. Blue bars denote the indices computed using the nominal distribution, red bars denote the indices when the distribution is perturbed to maximize $T_2$, cyan bars denote the indices when the distribution is perturbed to minimize $T_2$.}
\label{fig:synthetic_ex_function_total_indices}
\end{figure}

Figure~\ref{fig:synthetic_ex_function_distributions} displays samples from the nominal distribution in blue (left), the perturbed distribution which maximizes $T_2$ in red (center), and the perturbed distribution which minimizes $T_2$ in cyan (right). In each column, the top row is the marginal distribution for $X_1$, the middle row is the marginal distribution for $X_2$, and the bottom row is the marginal distribution for $X_3$. Notice that, as expected, the distribution of $X_1$ is perturbed to give greater probability near 0 (1) to maximize (minimize) $T_2$. Further, to maximize (minimize) $T_2$, the distribution of $X_2$ is concentrated toward the boundary (interior) which has the effect of maximizing (minimizing) the variance of $X_2$. Conversely, the distribution of $X_3$ is concentrated toward the interior (boundary) to minimize (maximize) the variance of $X_3$, since decreasing (increasing) the variance of $X_3$ will increase (decrease) $T_2$.

\begin{figure}[h]
\begin{tabular}{ m{1 cm} m{3.2 cm} m{3.2 cm} m{3.2 cm} }
$X_1$ & \includegraphics[width=.2 \textwidth]{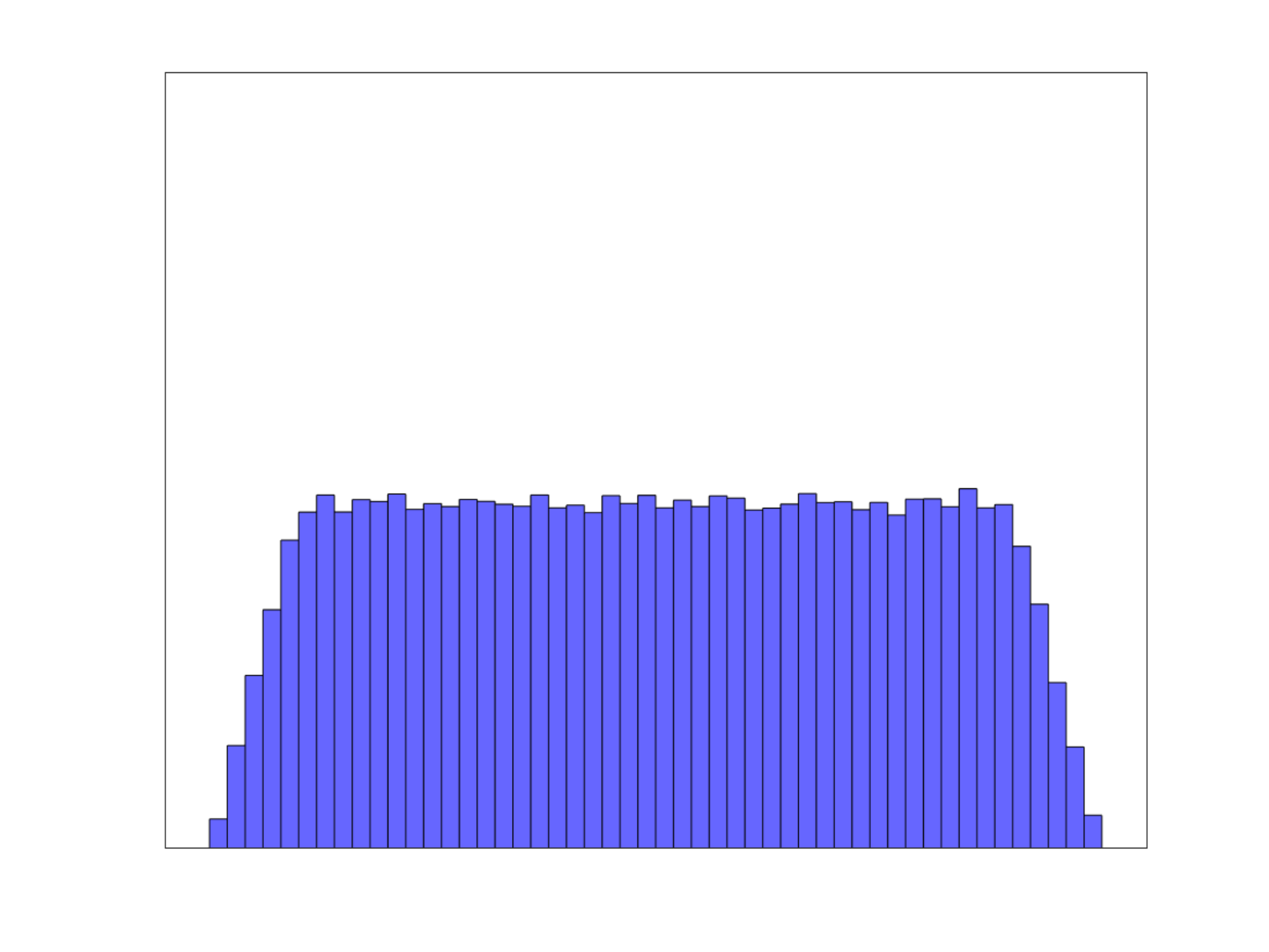} &  \includegraphics[width=.2 \textwidth]{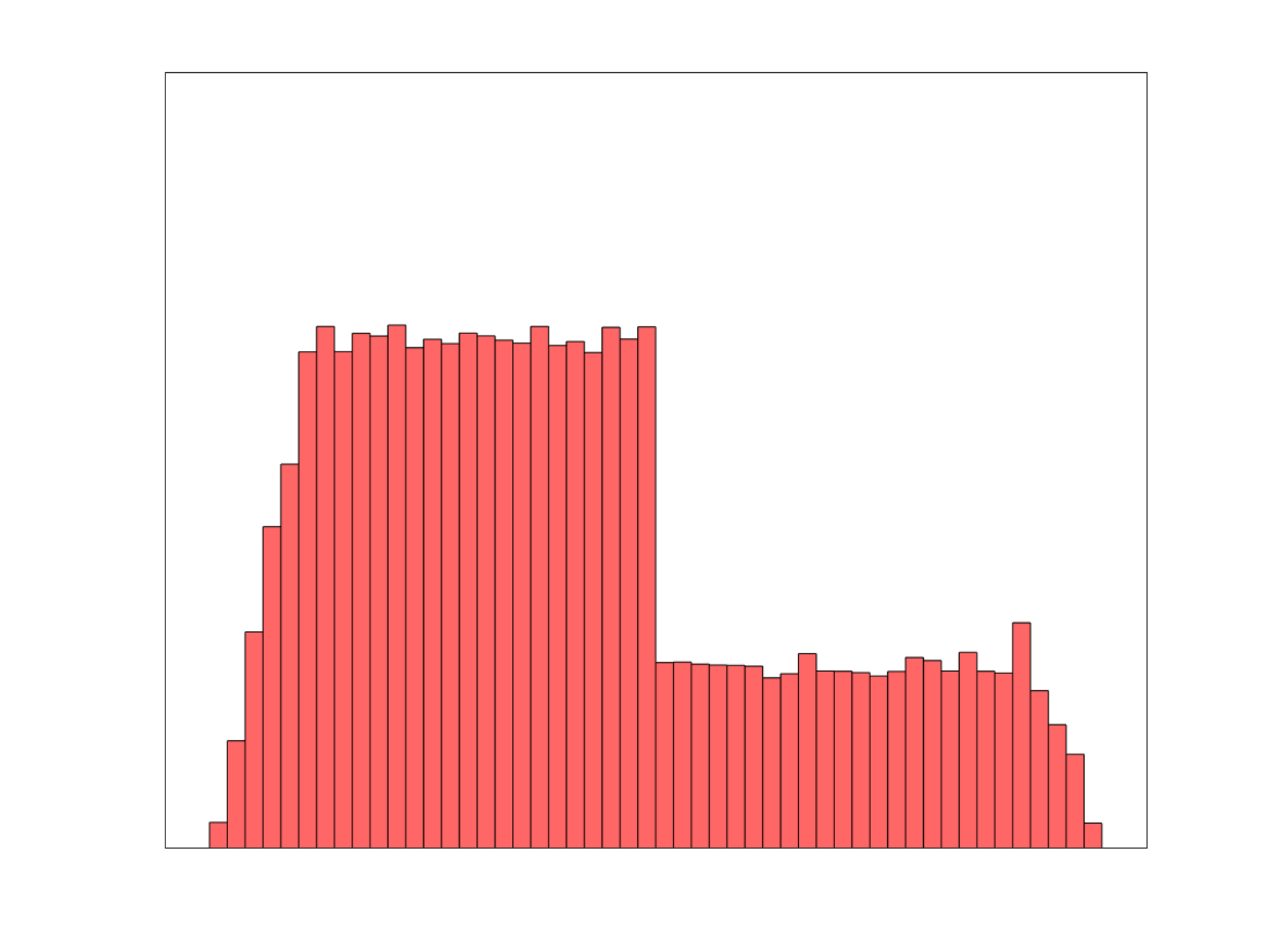} &  \includegraphics[width=.2 \textwidth]{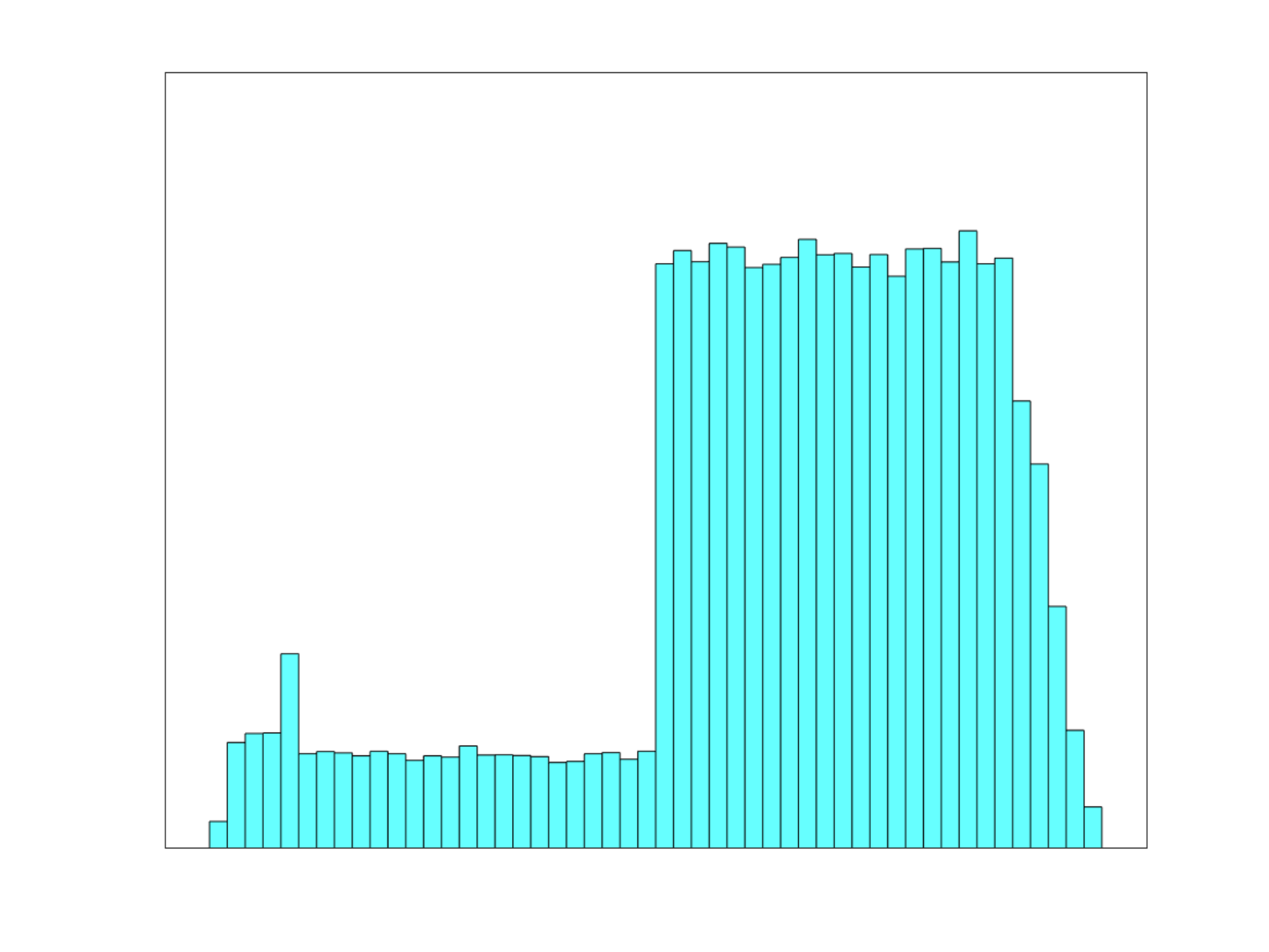} \\
$X_2$ & \includegraphics[width=.2 \textwidth]{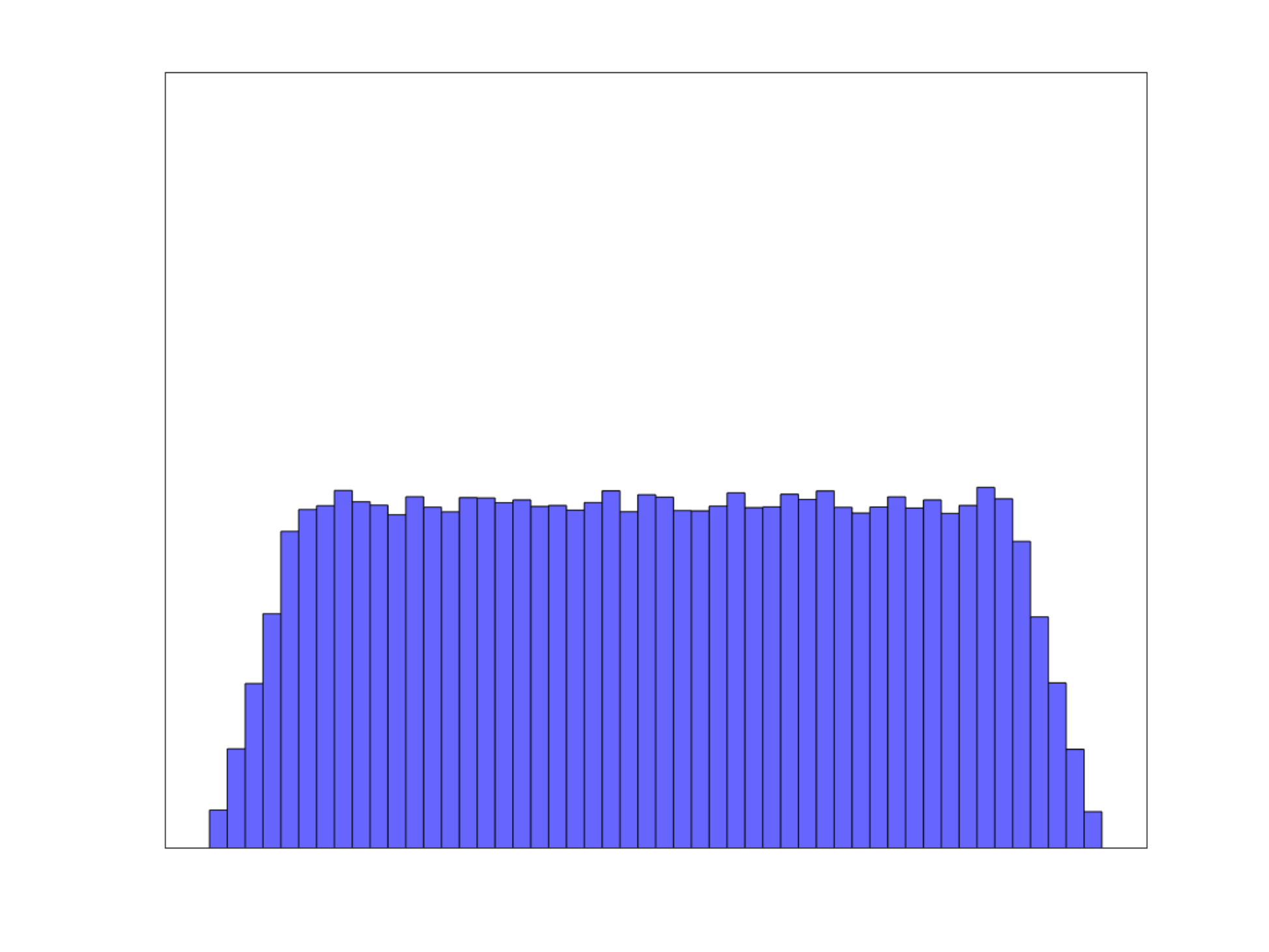} &  \includegraphics[width=.2 \textwidth]{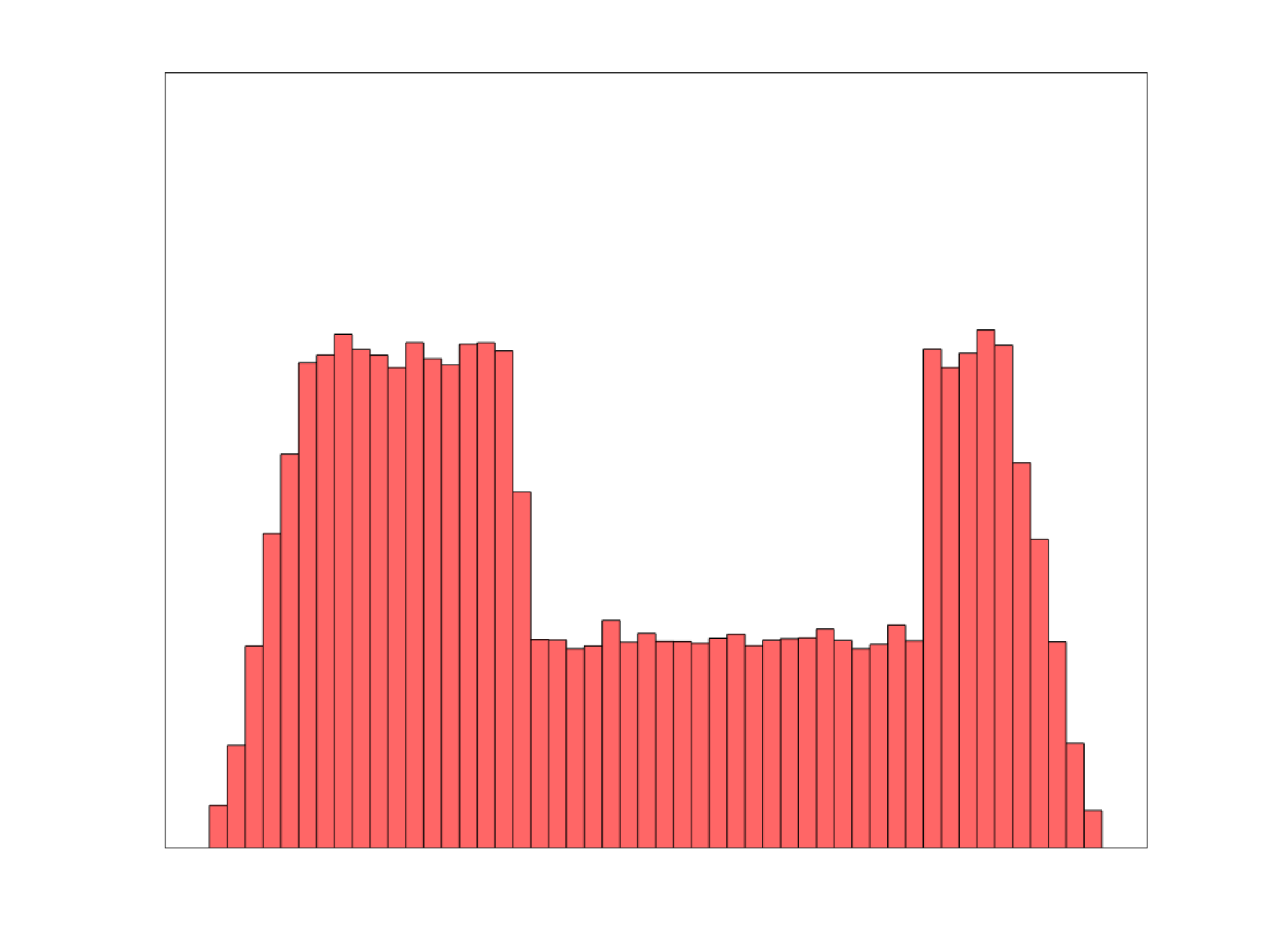} &  \includegraphics[width=.2 \textwidth]{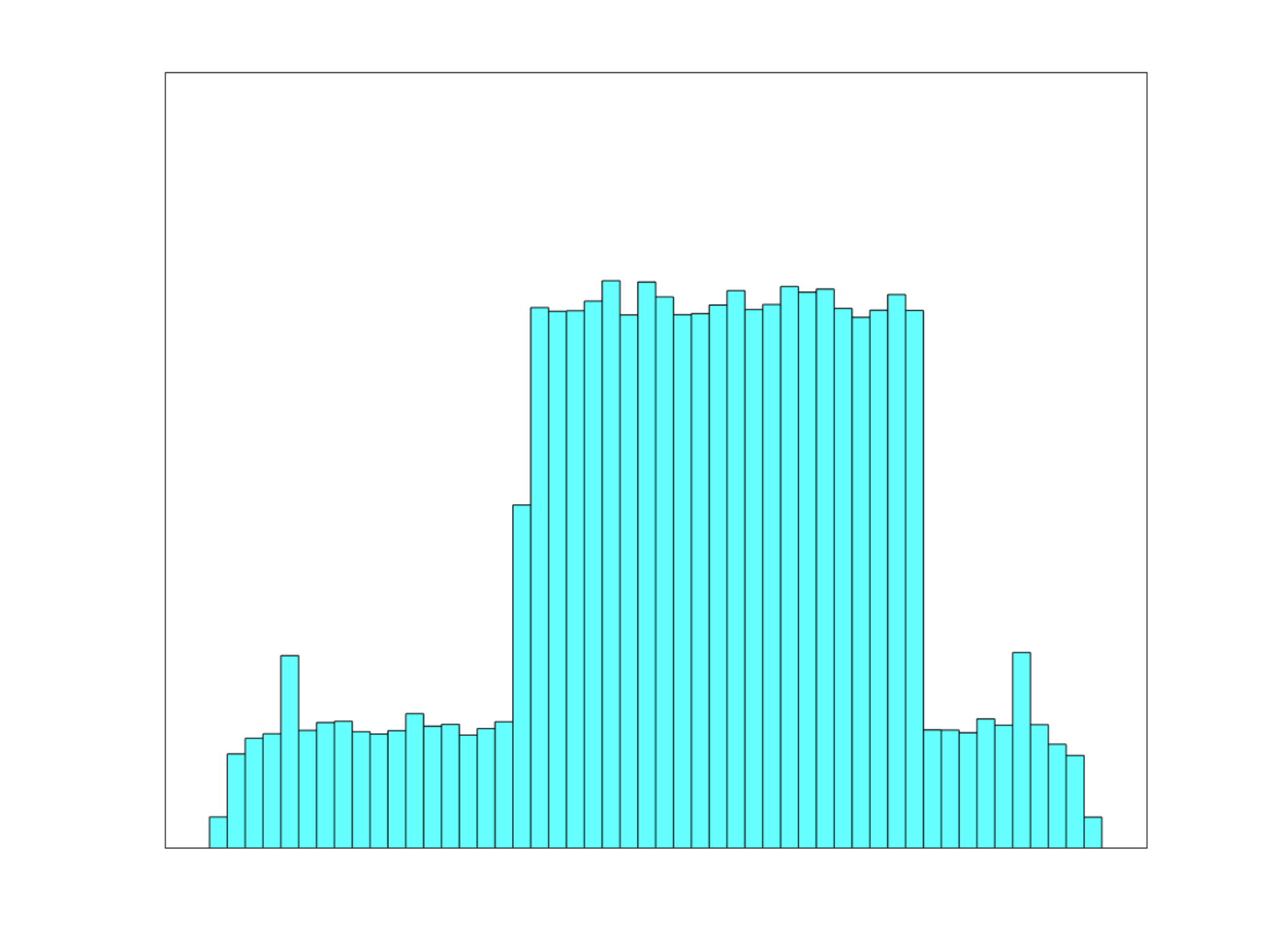} \\
$X_3$ & \includegraphics[width=.2 \textwidth]{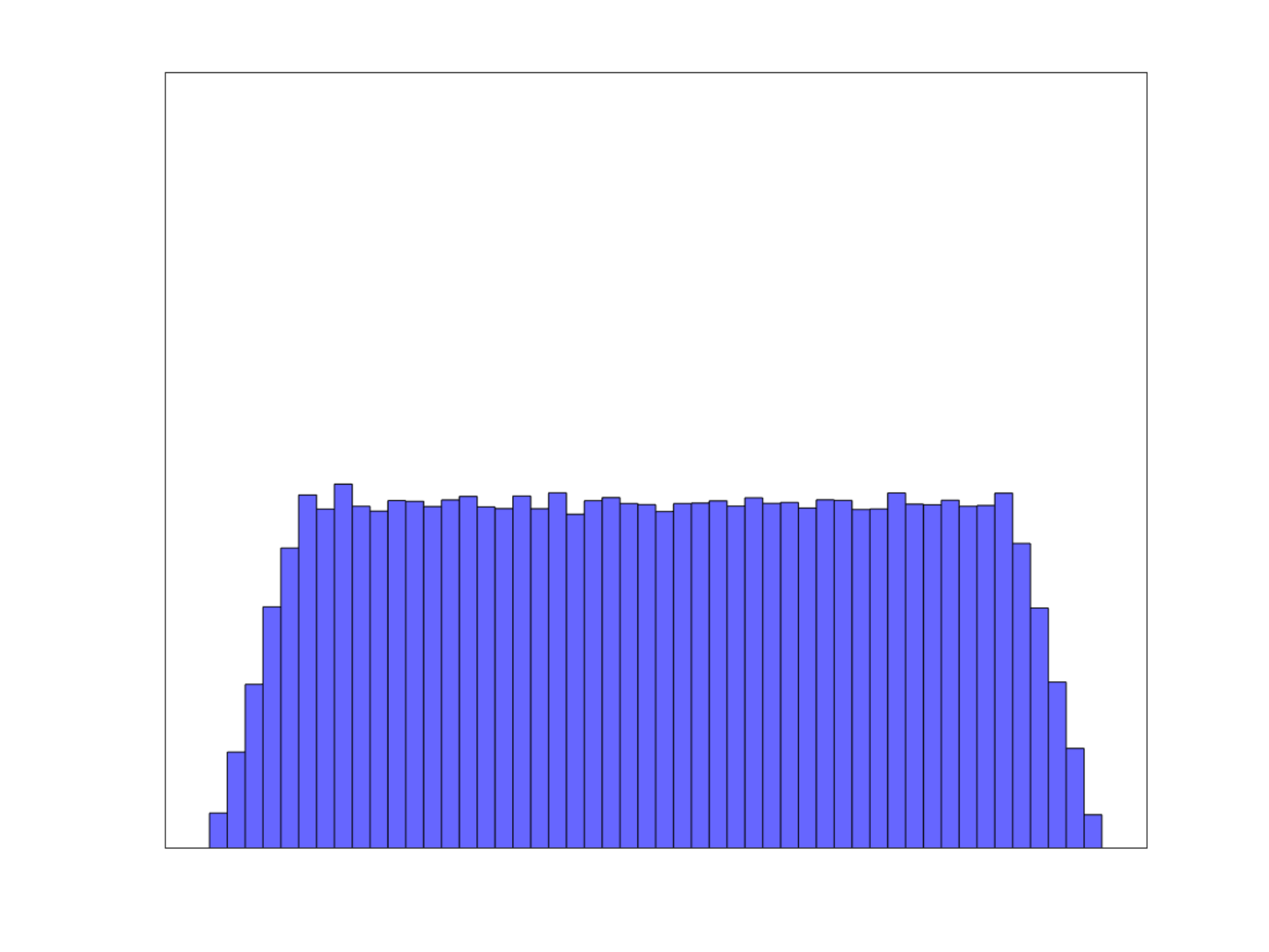} &  \includegraphics[width=.2 \textwidth]{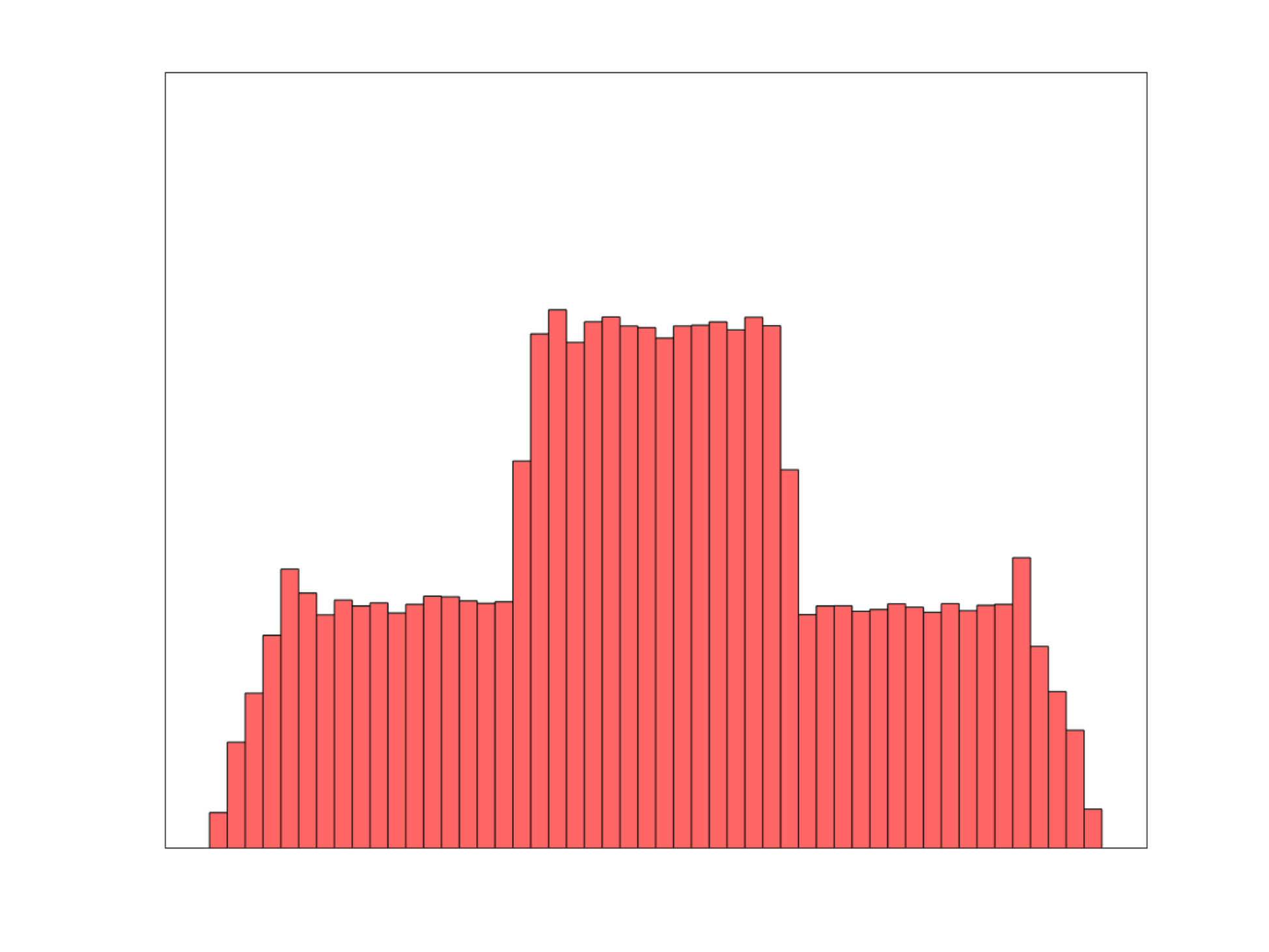} &  \includegraphics[width=.2 \textwidth]{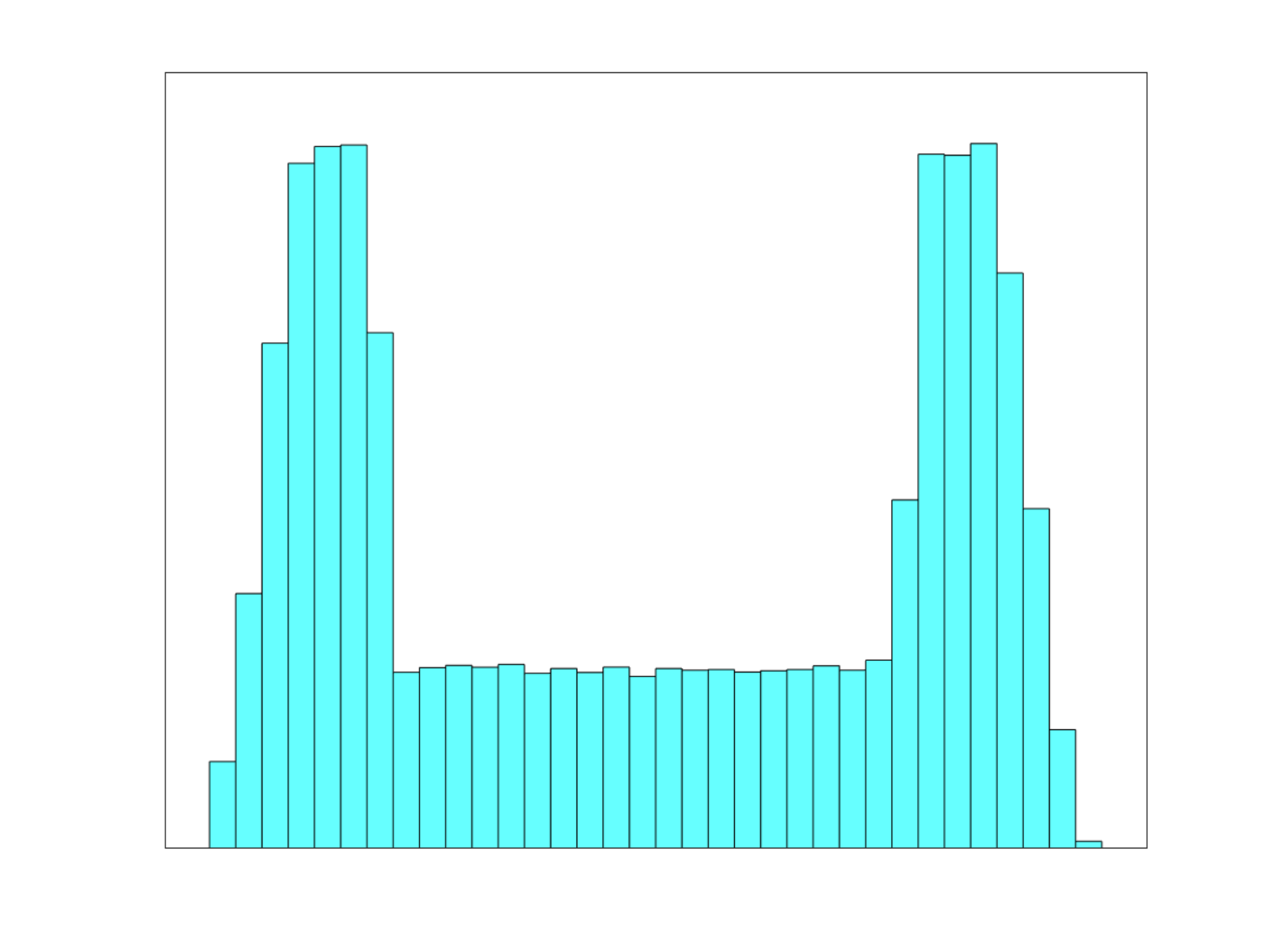} \\
\end{tabular}
\caption{Samples from the distribution of $\X$ and the perturbed distributions. The top, middle, and bottom rows correspond to variables $X_1,X_2$, and $X_3$ respectively. The left column correspond to samples from $\X$, the center column correspond to samples from the perturbed distribution which maximizes $T_2$, and the right column correspond to samples from the perturbed distribution which minimizes $T_2$.}
\label{fig:synthetic_ex_function_distributions}
\end{figure}

In practice, the support of $X_i$ is frequently unknown and in many cases chosen ad hoc. Using a mixture distribution to model $X_i$ and coupling it with marginal perturbation robustness analysis, is a simple way to account for this uncertainty. Further, visualizing the perturbed distributions indicates whether the support will significantly influence the Sobol' indices.

\subsection{Advection diffusion example}
In this example, we analyze the sensitivity of a contaminant concentration to uncertainties in its source and dynamics. Specifically, let $u(y_1,y_2)$ be the steady state contaminant concentration at a given point in the domain $[0,1]^2$ and our QoI be the average concentration in the subregion $[0.5,0.7] \times [0.5,0.7]$. To evaluate the QoI $f(\X)$, we solve 
\begin{align*}
-\epsilon \bigtriangleup u + v \cdot \nabla u =  s & \qquad \text{on } (0,1)^2\\
u = 0 & \qquad \text{on } \{0\} \times (0,1)\\
u= 0 & \qquad \text{on } (0,1) \times \{0\} \\
\mathbf n \cdot \nabla u = \nu_1 u & \qquad \text{on } \{1\} \times (0,1)\\
\mathbf n \cdot \nabla u = \nu_2 u & \qquad \text{on } (0,1) \times \{1\} \\
\end{align*}
where the velocity field and source are given by
\begin{eqnarray*}
v(y_1,y_2) = (\alpha_1(y_1+0.5),\alpha_2(y_2+0.5)) \qquad \text{and} \qquad s(y_1,y_2) = \beta e^{-\gamma ( (y_1-\xi_1)^2+(y_2-\xi_2)^2)},
\end{eqnarray*}
and $\mathbf n$ denotes the outer normal vector to the boundary. The partial differential equation is solved using Chebfun \cite{chebfun}. The QoI $f(\X)$ is then computed by integrating $u(y_1,y_2)$ over $[0.5,0.7] \times [0.5,0.7]$. In this example,
\begin{eqnarray*}
\X = (\epsilon,\alpha_1,\alpha_2,\xi_1,\xi_2,\gamma,\beta,\nu_1,\nu_2)
\end{eqnarray*}
is a vector of 9 uncertain inputs which are assumed to be independent uniform random variables with $\pm 30 \%$ uncertainty about their nominal values, given in Table~\ref{tab:nominal_values}.

\begin{table}[h]
\centering
\begin{tabular}{cccccccccc}
\toprule
Parameter & $\epsilon$ & $\alpha_1$ & $\alpha_2$ & $\xi_1$ & $\xi_2$ & $\gamma$ & $\beta$ & $\nu_1$ & $\nu_2$ \\
Nominal Value & 10 & 210 & 70 & 0.5 & 0.5 & 50 & 100 & 0.1 & 0.2\\
\bottomrule
\end{tabular}
\caption{Nominal values for parameters for the advection diffusion example. Each parameter is assumed to be independent and uniformly distributed with $\pm 30 \%$ uncertainty about these nominal values.}
\label{tab:nominal_values}
\end{table}

Figure~\ref{fig:state_and_source} illustrates several features of this advection diffusion example. The left panel of Figure~\ref{fig:state_and_source} shows contours of the average source (computed from 10,000 samples), with the average velocity field overlaid with red arrows. The right panel of Figure~\ref{fig:state_and_source} displays contours of the average solution $u(y_1,y_2)$. In both panels, cyan lines indicate the Robin boundaries and magenta lines indicate the Dirichlet boundaries and a black box encloses the region over which we integrate to evaluate $f$.

\begin{figure}[h]
\centering
\includegraphics[width=.49 \textwidth]{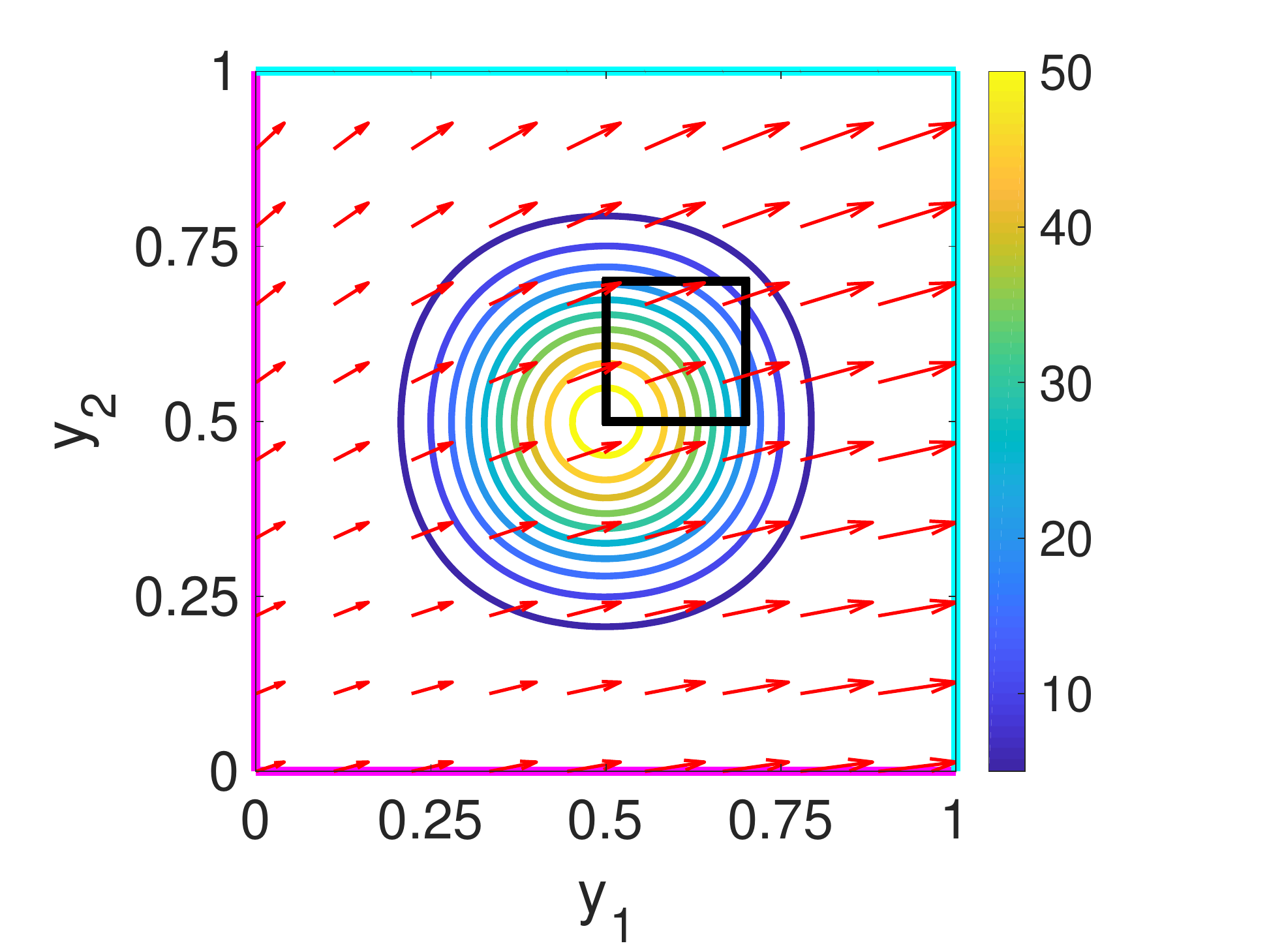}
\includegraphics[width=.49 \textwidth]{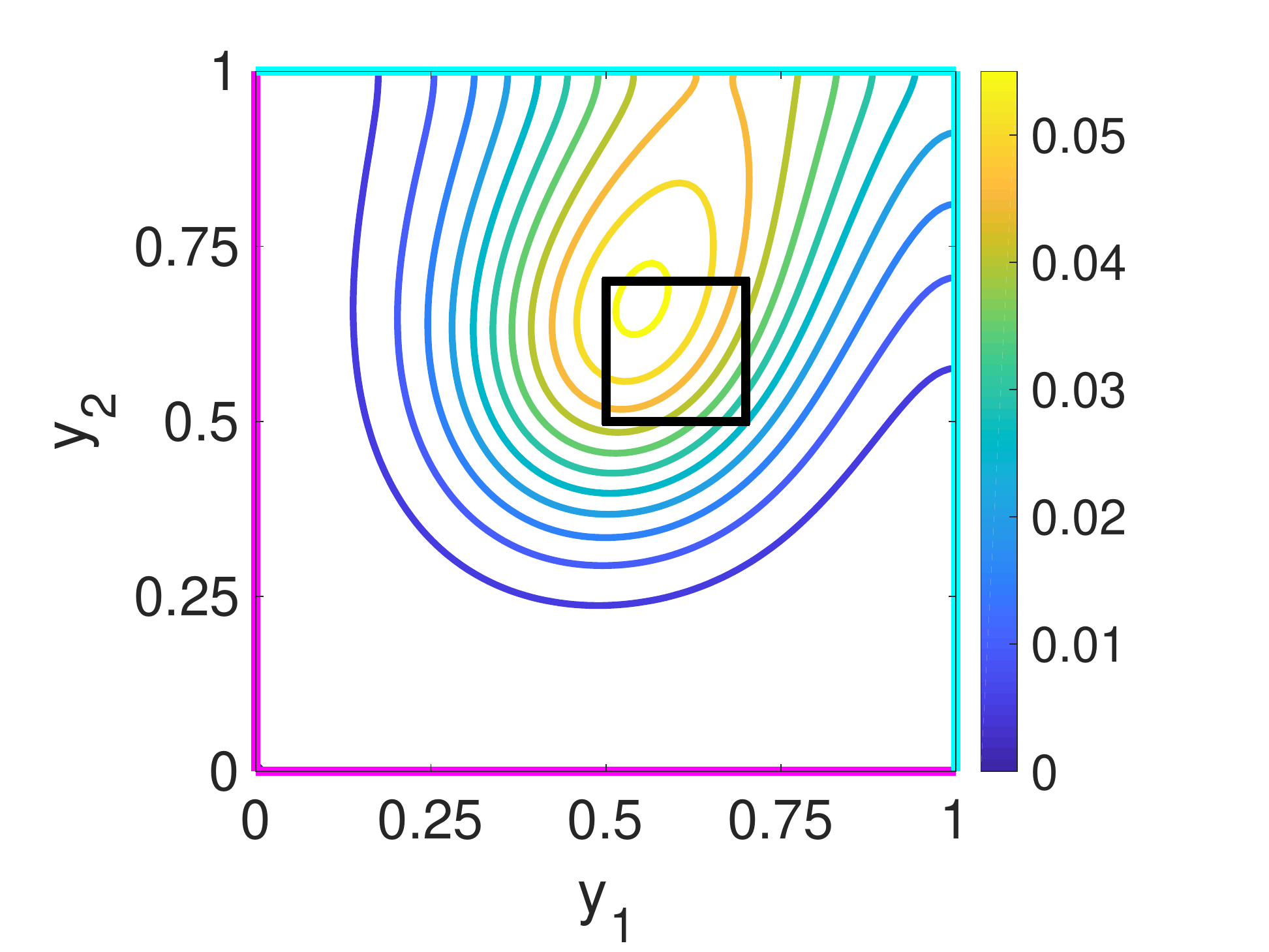}
\caption{Left: contours of the average source with the average velocity field overlaid with red arrows. Right: contours of the average contaminant concentration $u$. The cyan and magenta lines indicate the Robin and Dirichlet boundaries, respectively. The black box encloses the region over which we integrate to compute $f$.}
\label{fig:state_and_source}
\end{figure}

The total Sobol' indices of $f$ are computed via Monte Carlo integration with $n=10,000$ samples, followed by marginal distribution robustness analysis post processing. In a similar fashion to Figure~\ref{fig:linear_fun},  Figure~\ref{fig:adv_diff_indices_min_max} displays the nominal, maximum, and minimum total Sobol' indices observed over all perturbations with $\Delta \le \tau$. The top of the green bar and bottom of the magenta bar denote the maximum and minimum values, respectively; the blue region between them denotes the value of the nominal total Sobol' index. The green and magenta bars in Figure~\ref{fig:adv_diff_indices_min_max} do not correspond to a particular set of Sobol' indices for a fixed perturbation, but rather the extreme values observed over all perturbations. We observe that $\alpha_2$, $\nu_1$, and $\nu_2$ exhibit minimal influence on our QoI and do not change significantly when the distribution of $\X$ is perturbed. The source magnitude $\beta$ is found to be the most important parameter (under the nominal distribution), but its total Sobol' index varies significantly when the marginal distributions are perturbed.

\begin{figure}[h]
\centering
\includegraphics[width=.65 \textwidth]{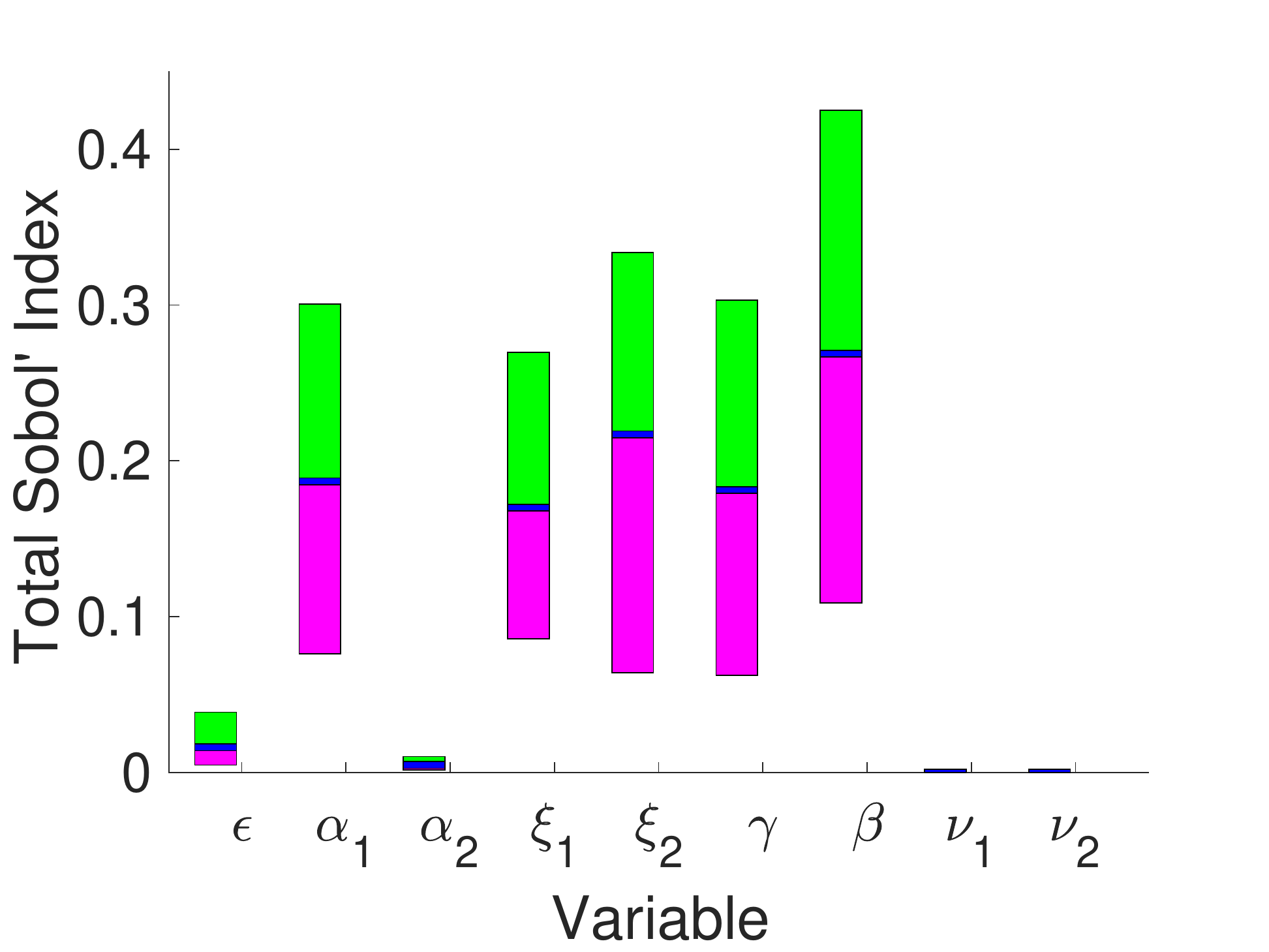}
\caption{Total Sobol' indices for the advection diffusion example. The top of the green bars and bottom of the magenta bars denote the maximum and minimum perturbed total Sobol' indices, respectively; the blue area between them indicates the nominal total Sobol' indices.}
\label{fig:adv_diff_indices_min_max}
\end{figure}

To further explore the influence of $\beta$, we analyze the particular perturbations which cause the greatest change in its total Sobol' index. Figure~\ref{fig:adv_diff_indices_pert_7} displays the nominal total Sobol' indices in blue, the perturbed total Sobol' indices which maximize the total Sobol' index for $\beta$ in red, and the perturbed total Sobol' indices which minimize the total Sobol' index for $\beta$ in cyan. This is complemented by Figure~\ref{fig:adv_diff_dist_pert_7} which shows the marginal distributions for each parameter using the nominal distribution in blue, the perturbed distribution which maximize the total Sobol' index for $\beta$ in red, and the perturbed distribution which minimize the total Sobol' index for $\beta$ in cyan. We observe that in the maximum case (red), the total Sobol' index for $\beta$ is much larger than all other total Sobol' indices. The perturbed distribution in this case exhibit several features:
	\begin{enumerate}
	\item[$\bullet$]  greater probability near the boundaries of $\beta$'s support, 
	\item[$\bullet$] greater probability near the interior of $\gamma$'s support,
	\item[$\bullet$] greater probability for smaller values of $\alpha_1$ and $\xi_2$,
	\item[$\bullet$] greater probability for larger values of $\alpha_2$ and $\xi_1$.
	\end{enumerate}
	The minimum case (cyan) gives a different conclusion where $\beta$ is the fifth most important parameter. The features of the perturbed distribution in this case are opposite of the maximum case (red). In particular, the source location and advection velocity become more important by giving greater probability to the source being in the top left region of the domain (smaller $y_1$, larger $y_2$), and having having greater probability in the tails of the horizontal advection speed $\alpha_1$.

\begin{figure}[h]
\centering
\includegraphics[width=.65 \textwidth]{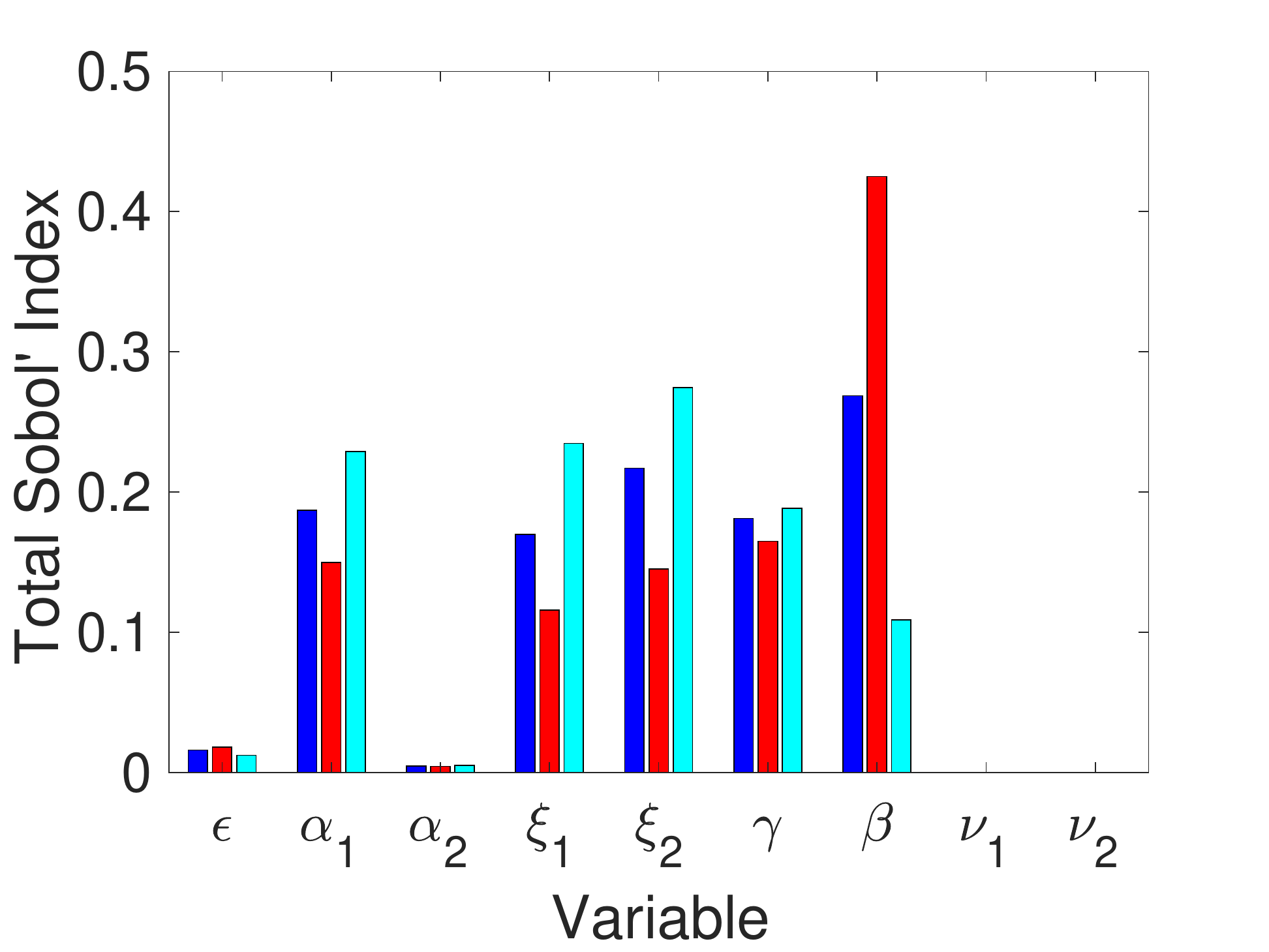}
\caption{Total Sobol' indices for the advection diffusion example. Blue bars denote the indices computed using the nominal distribution, red bars denote the indices when the distribution is perturbed to maximize the total Sobol' index for $\beta$, cyan bars denote the indices when the distribution is perturbed to minimize the total Sobol' index for $\beta$.}
\label{fig:adv_diff_indices_pert_7}
\end{figure}

\begin{figure}
\begin{tabular}{ m{1 cm} m{3.2 cm} m{3.2 cm} m{3.2 cm} }
$\epsilon$ & \includegraphics[width=.2 \textwidth]{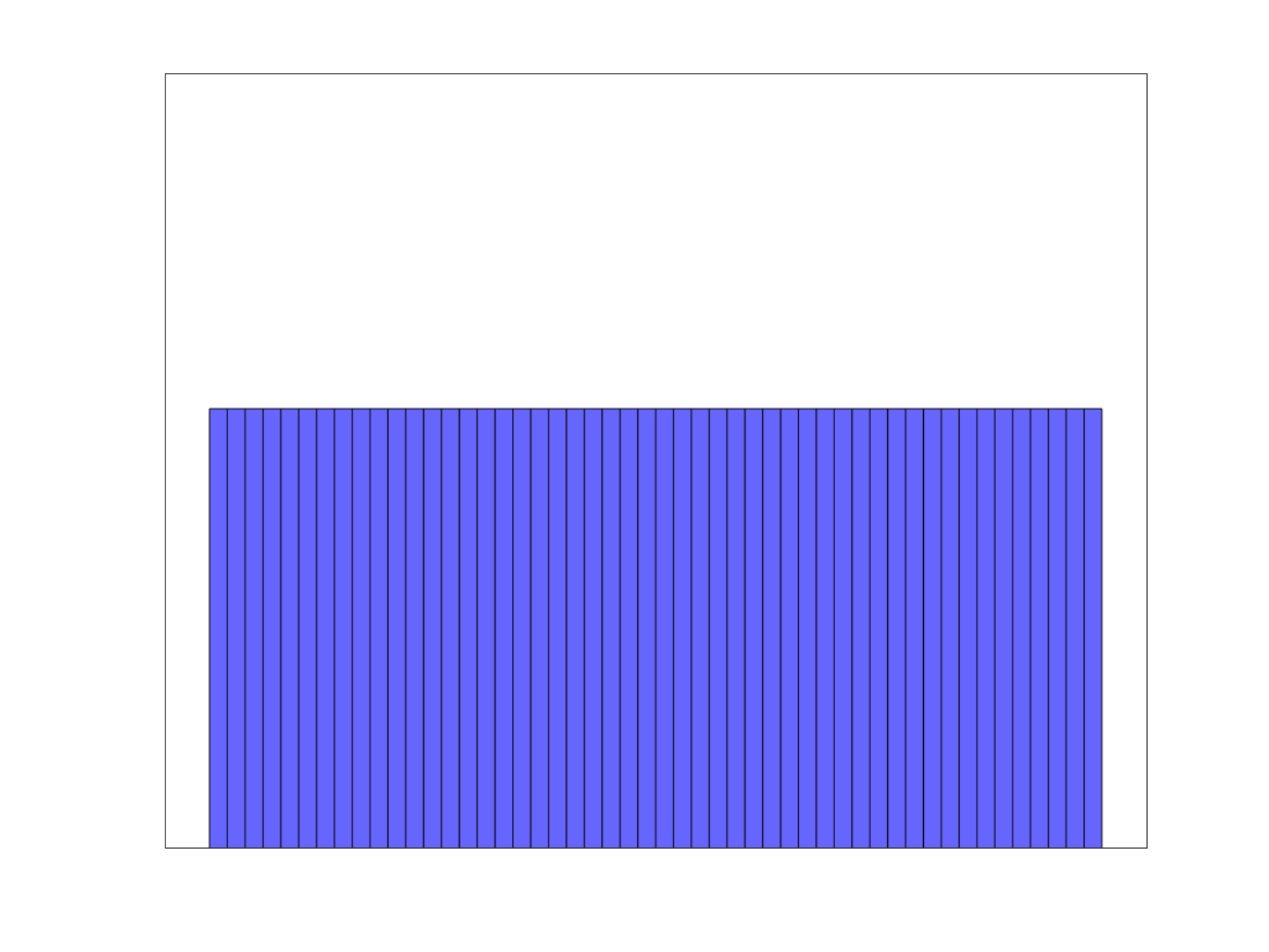} &  \includegraphics[width=.2 \textwidth]{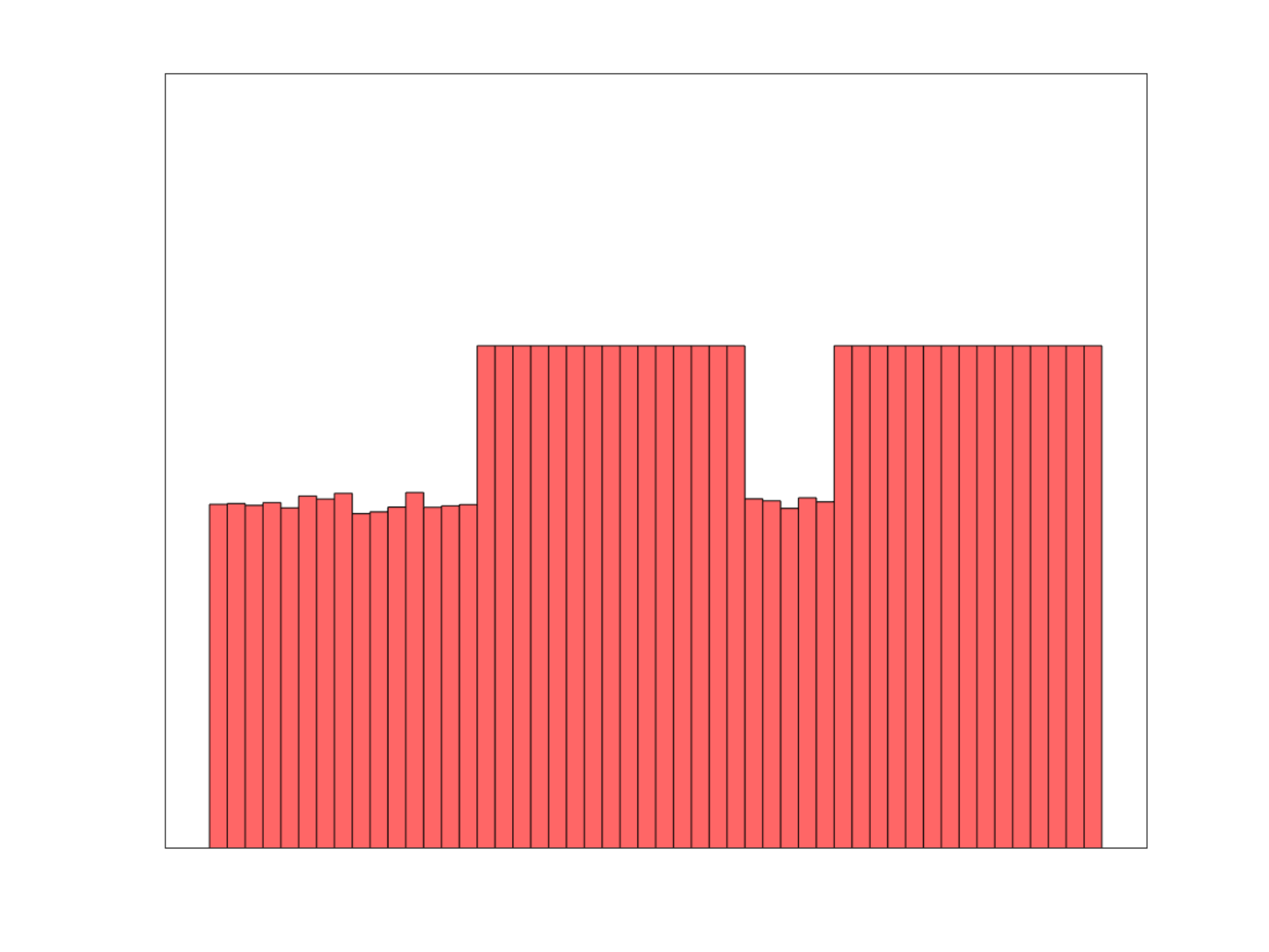} &  \includegraphics[width=.2 \textwidth]{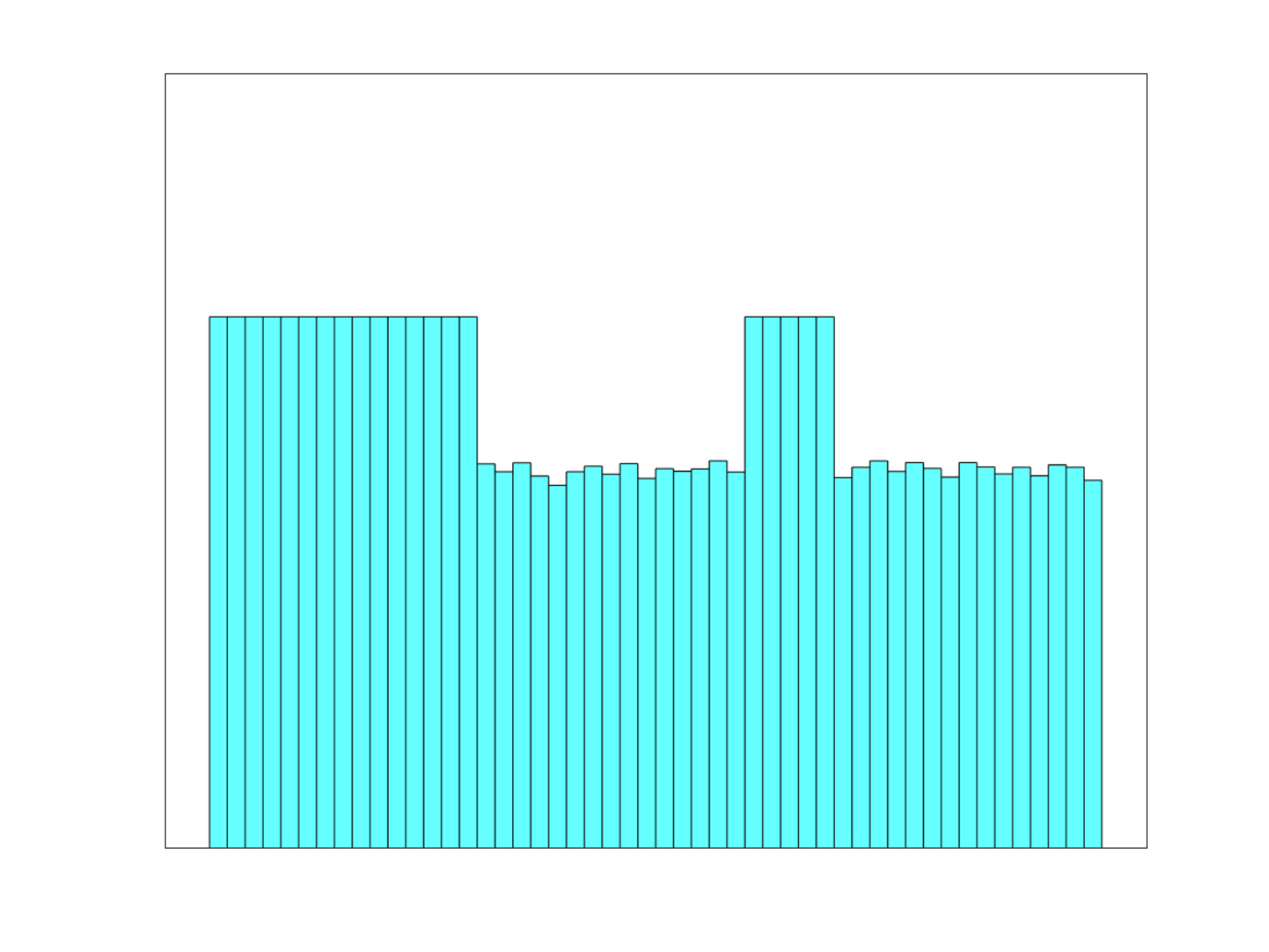} \\
$\alpha_1$ & \includegraphics[width=.2 \textwidth]{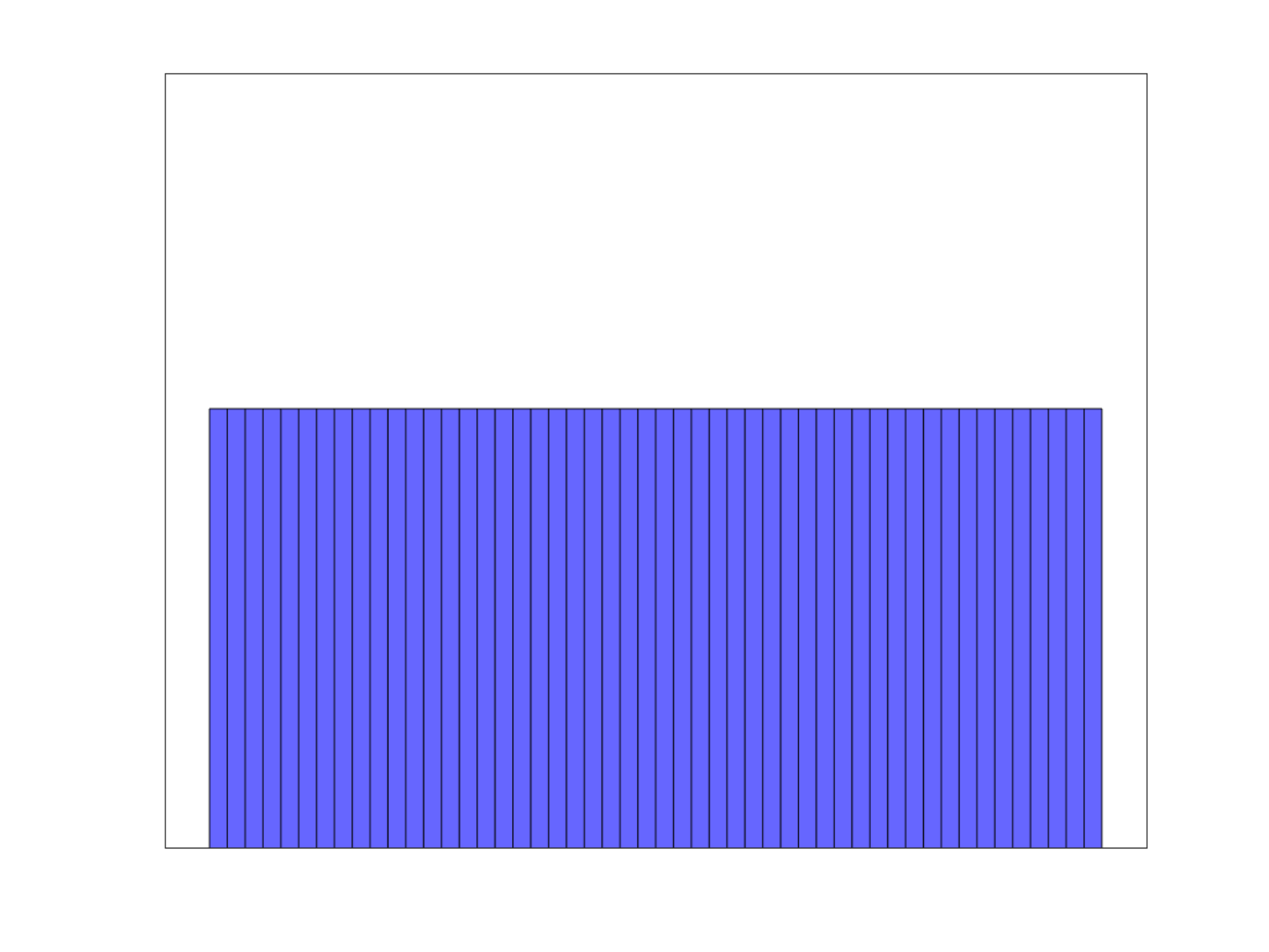} &  \includegraphics[width=.2 \textwidth]{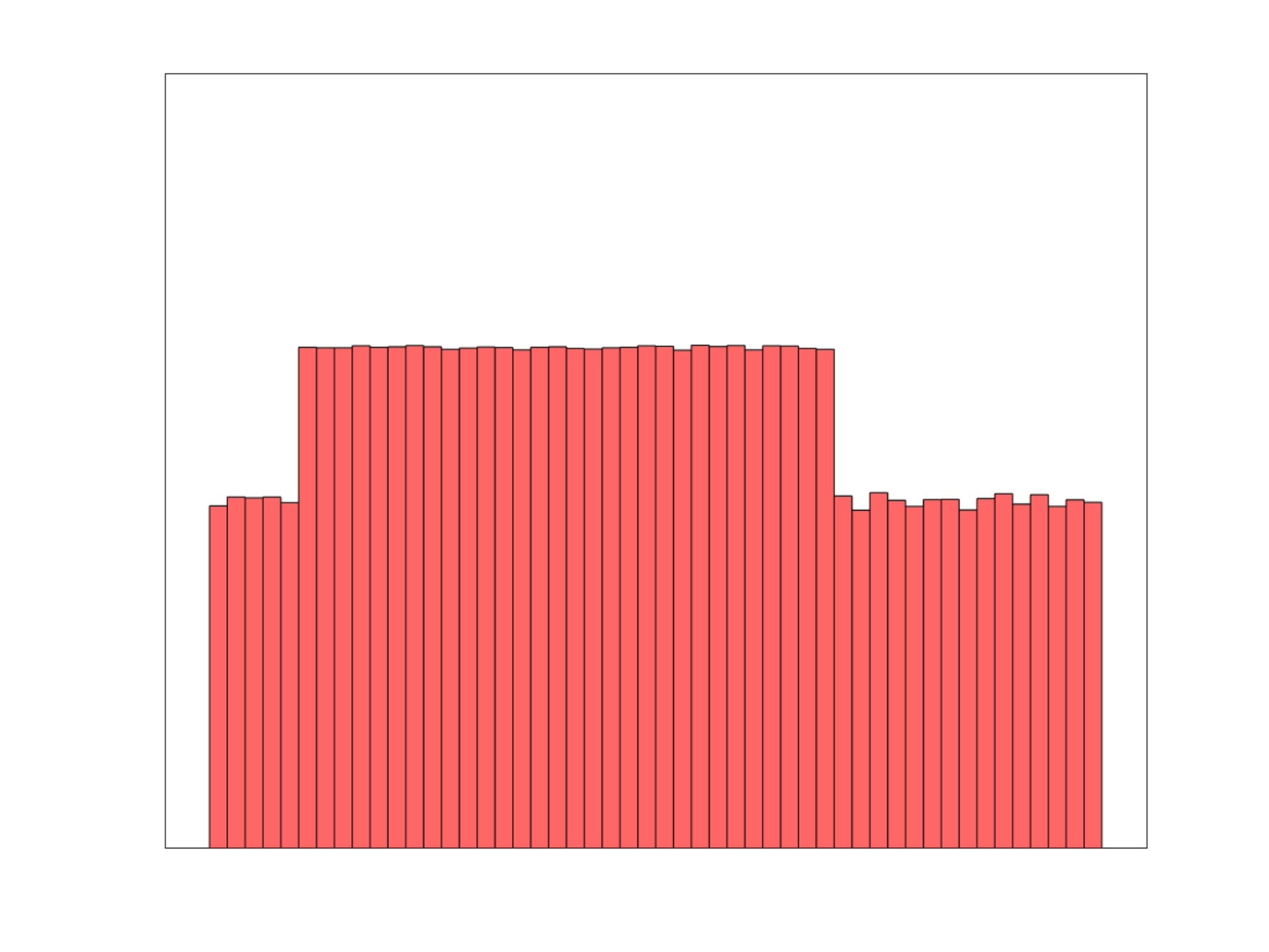} &  \includegraphics[width=.2 \textwidth]{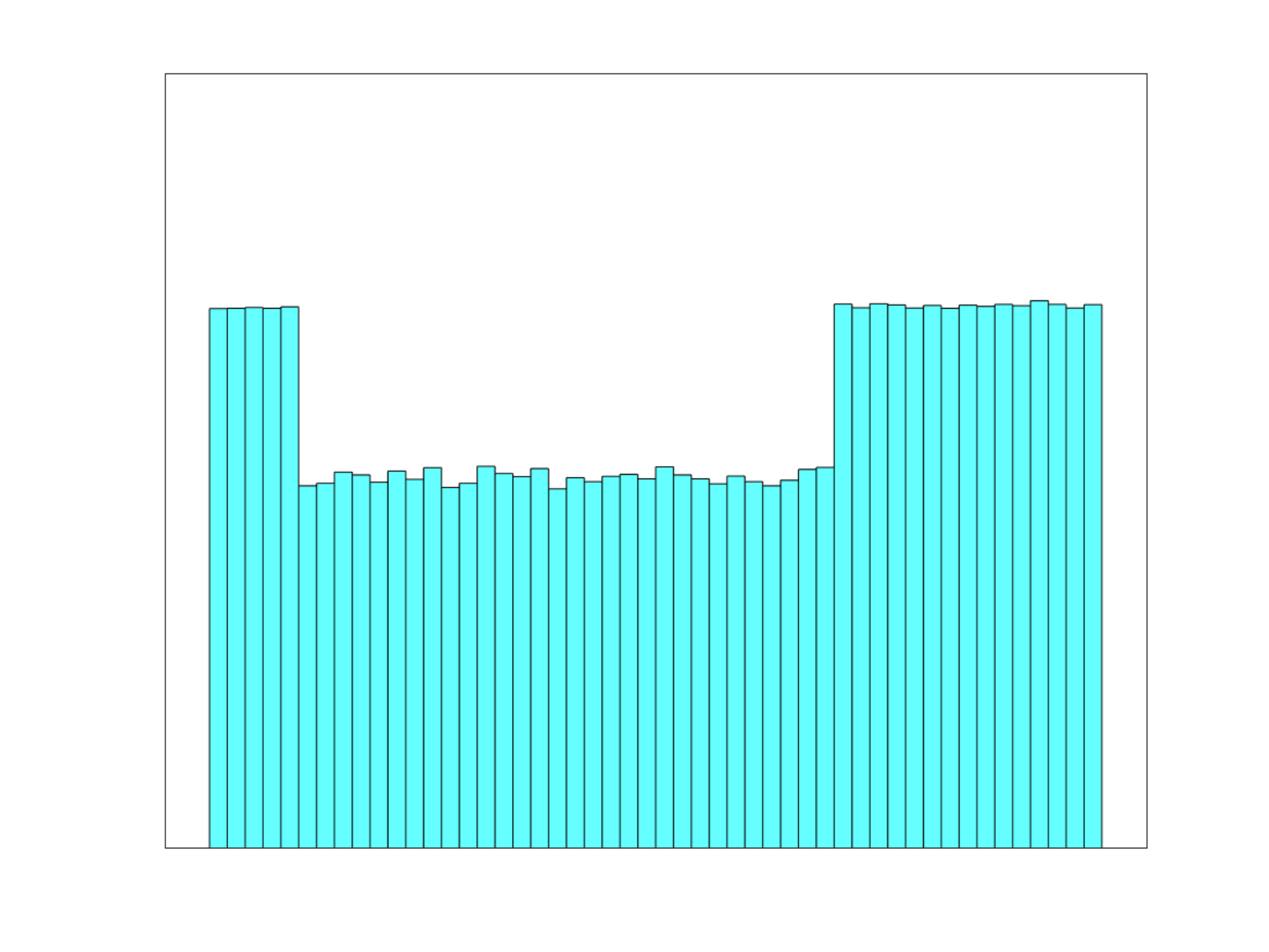} \\
$\alpha_2$ & \includegraphics[width=.2 \textwidth]{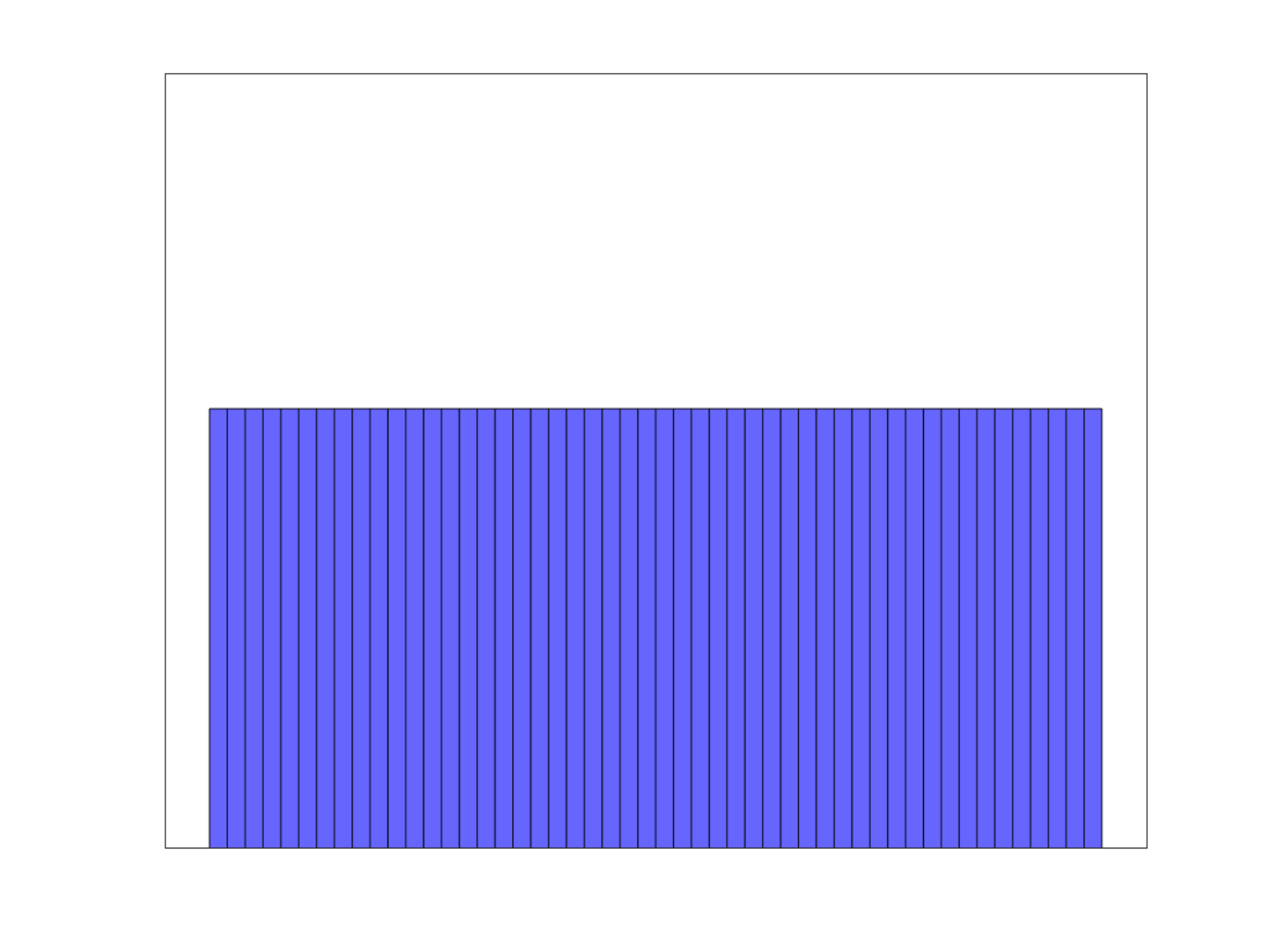} &  \includegraphics[width=.2 \textwidth]{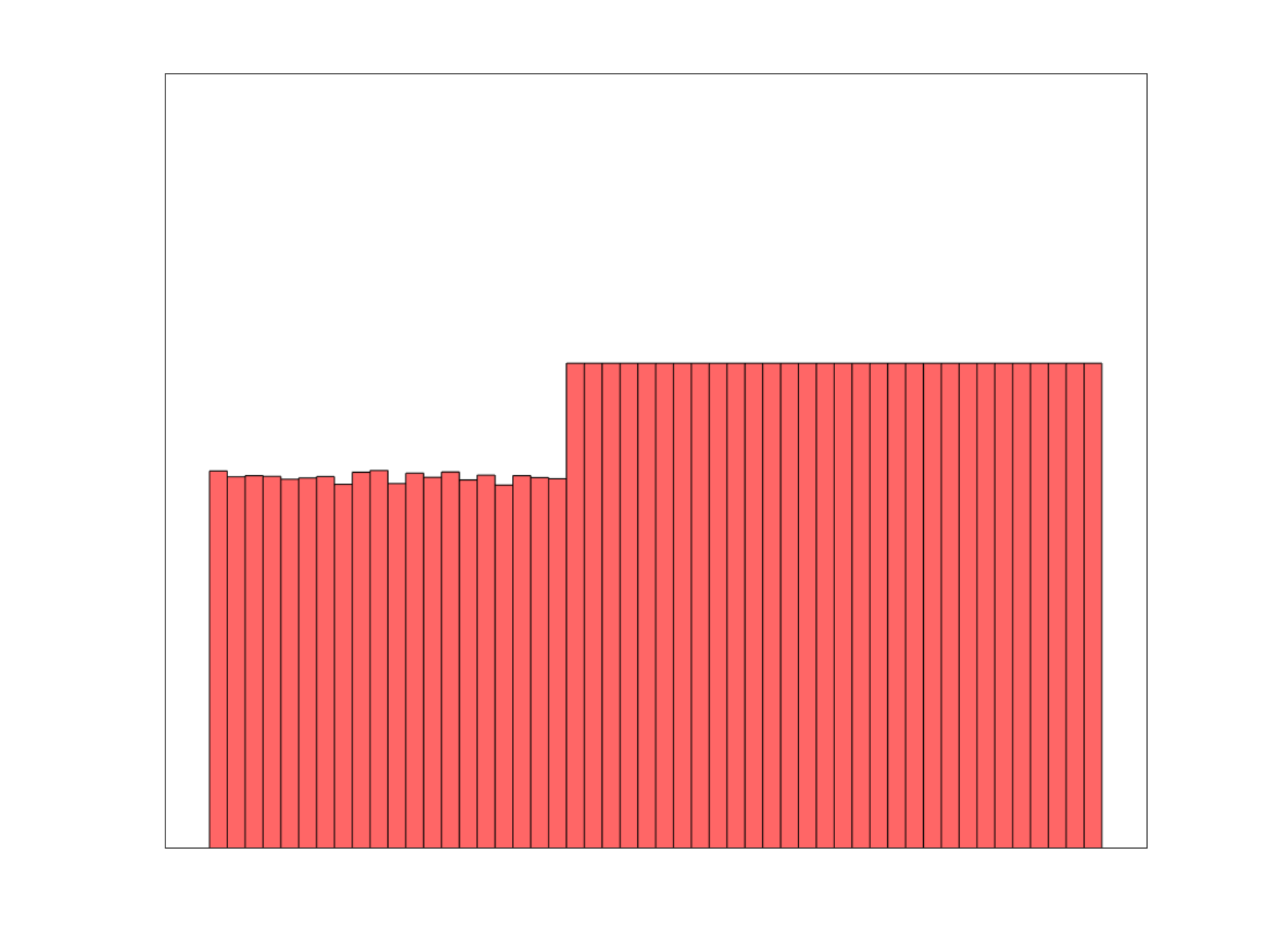} &  \includegraphics[width=.2 \textwidth]{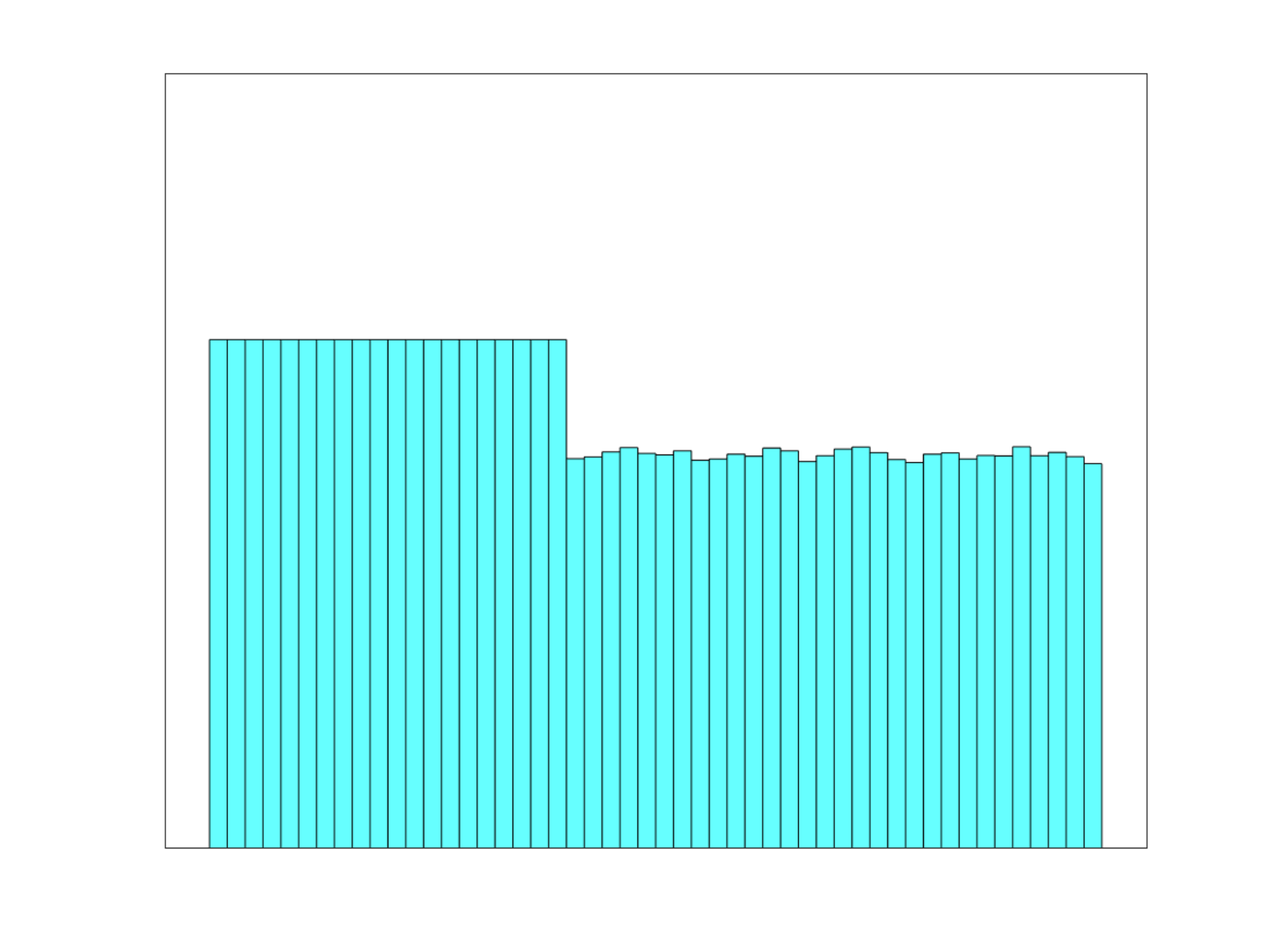} \\
$\xi_1$ & \includegraphics[width=.2 \textwidth]{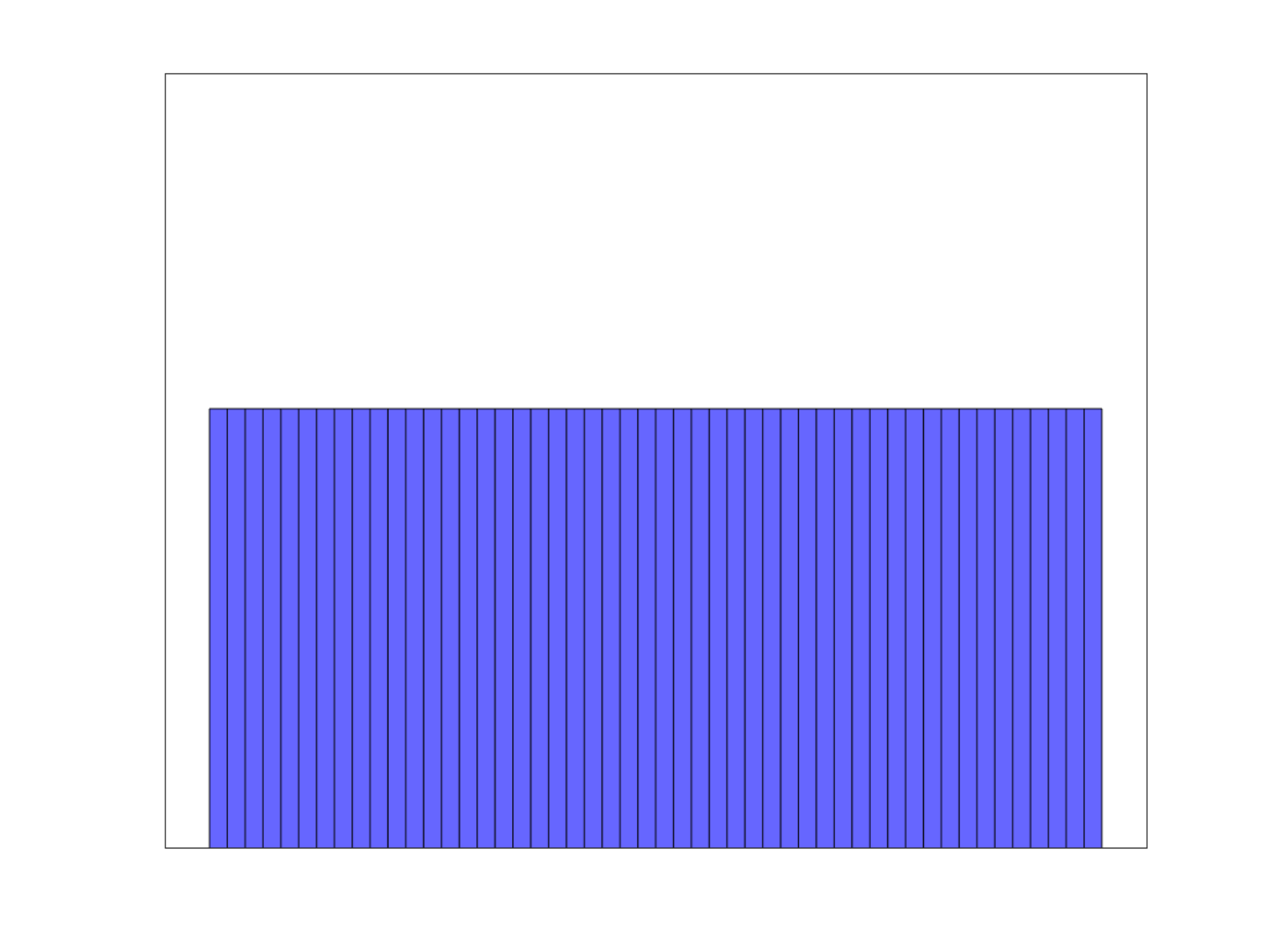} &  \includegraphics[width=.2 \textwidth]{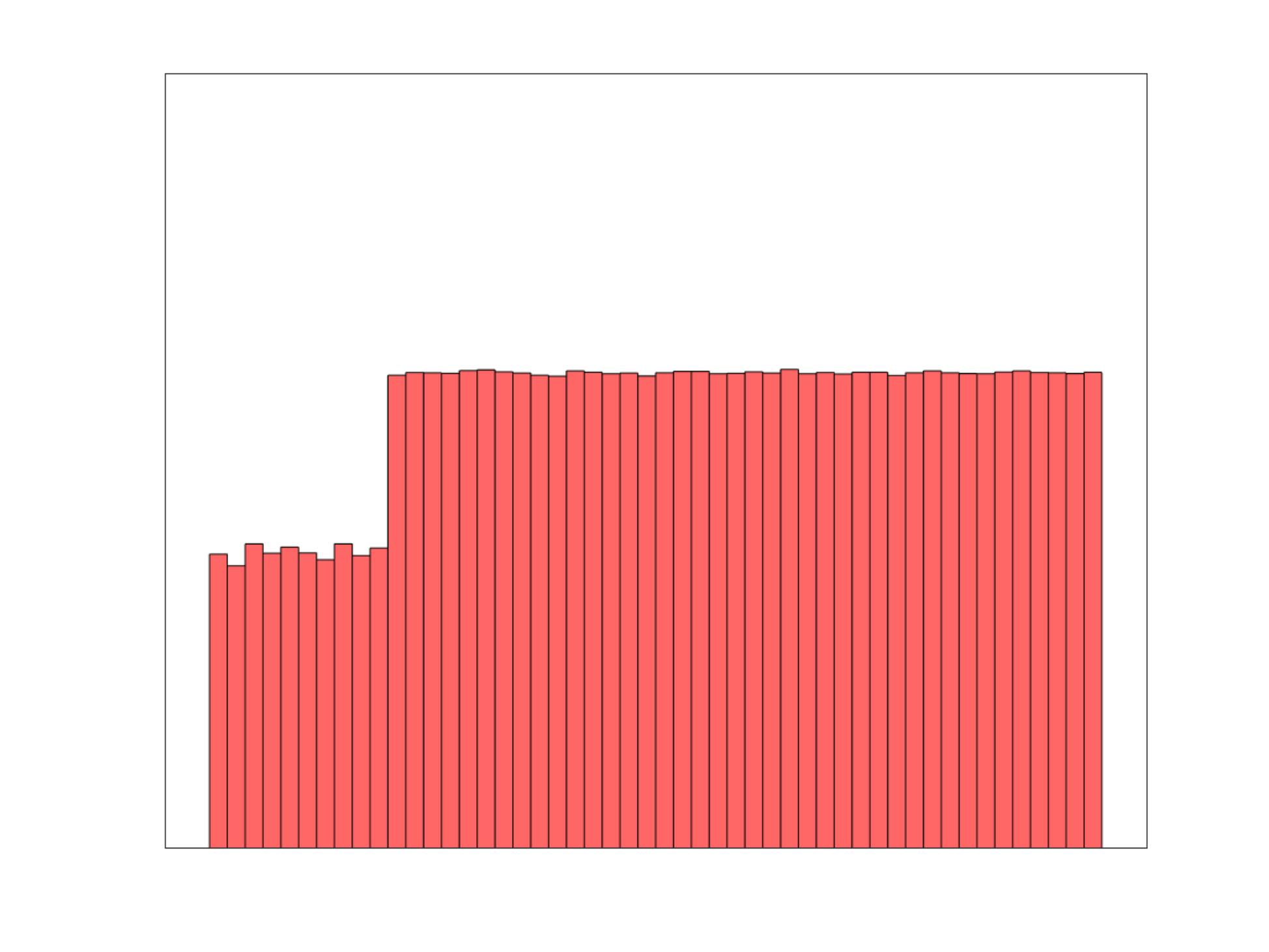} &  \includegraphics[width=.2 \textwidth]{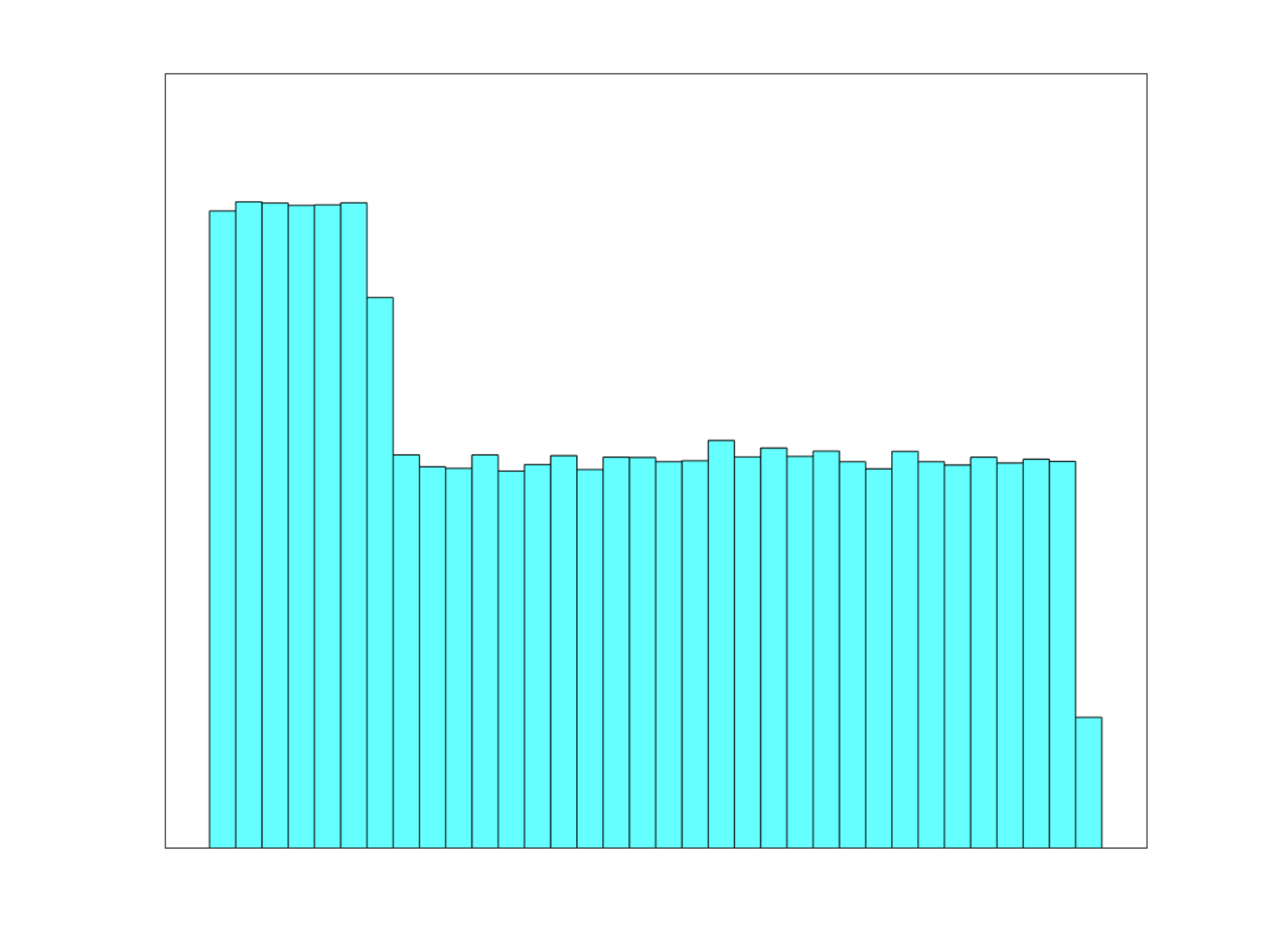} \\
$\xi_2$ & \includegraphics[width=.2 \textwidth]{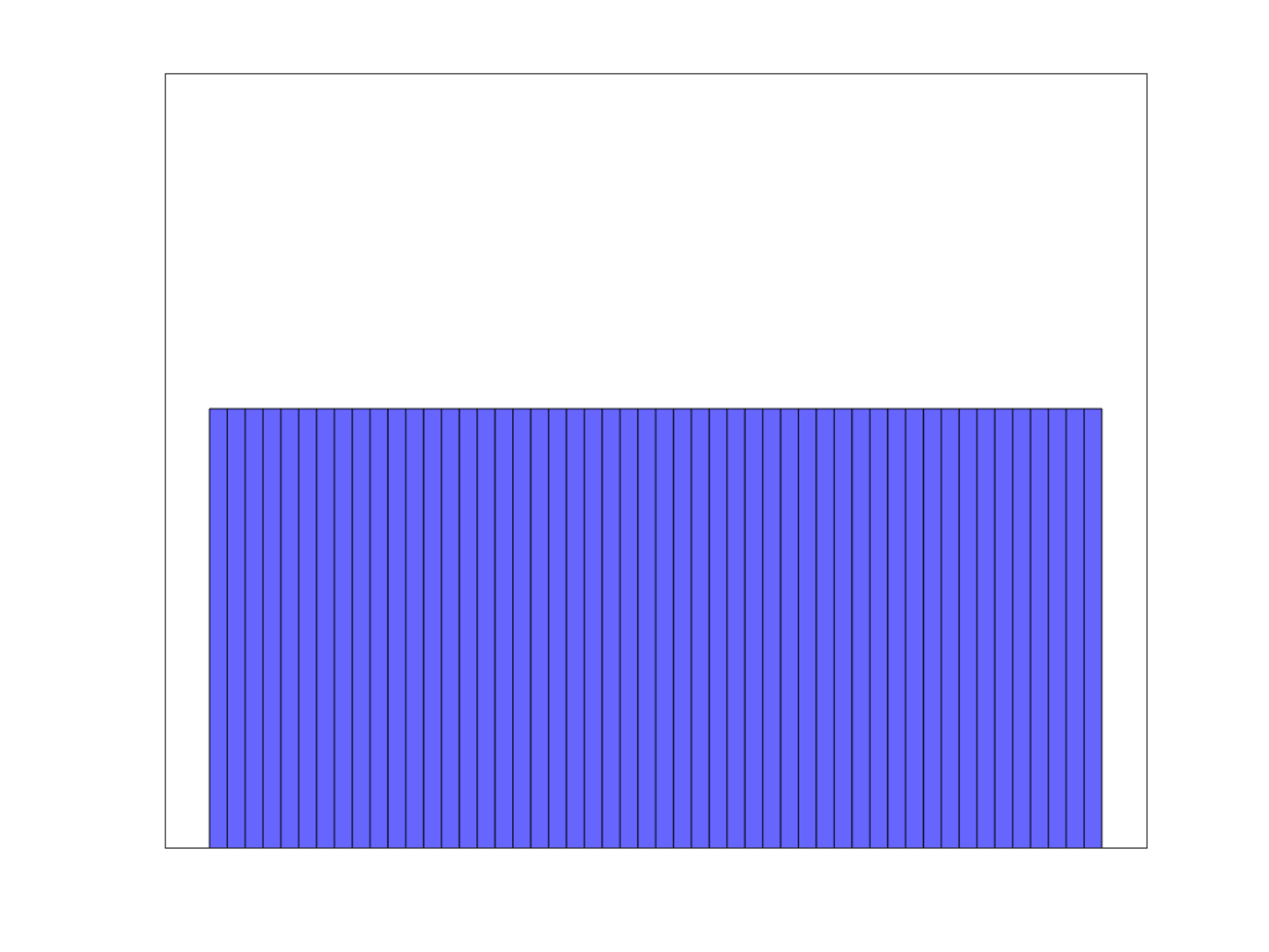} &  \includegraphics[width=.2 \textwidth]{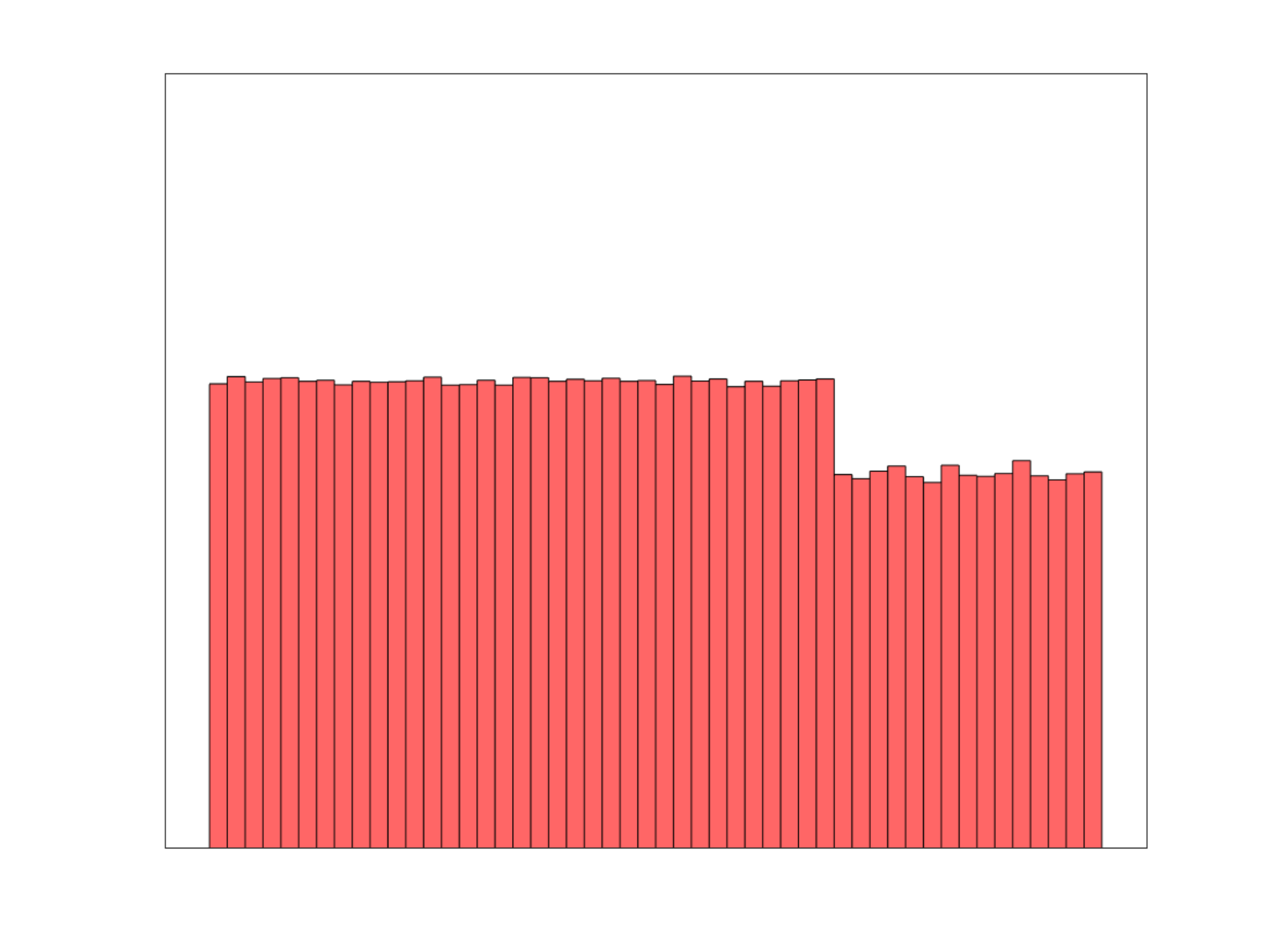} &  \includegraphics[width=.2 \textwidth]{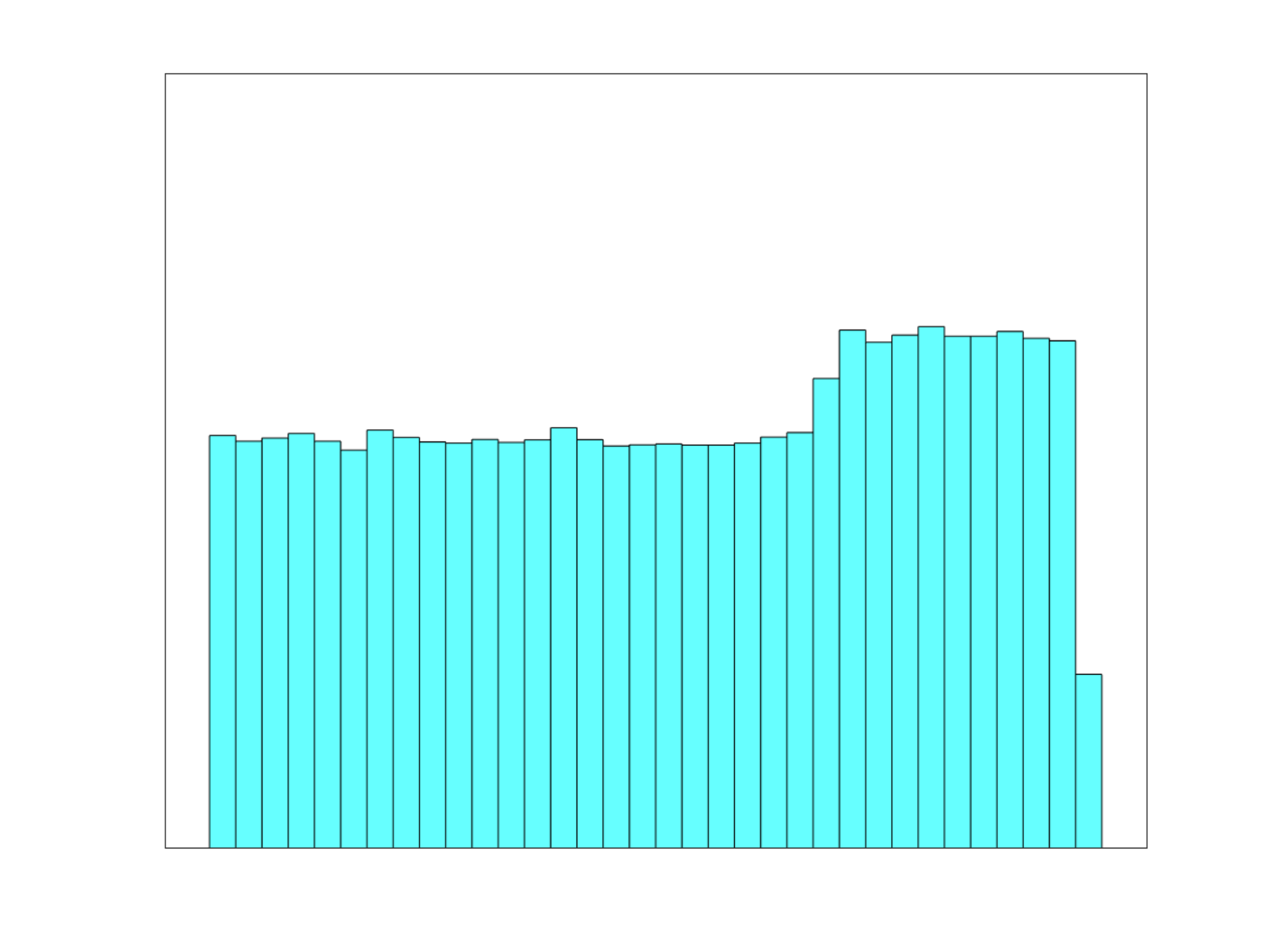} \\
$\gamma$ & \includegraphics[width=.2 \textwidth]{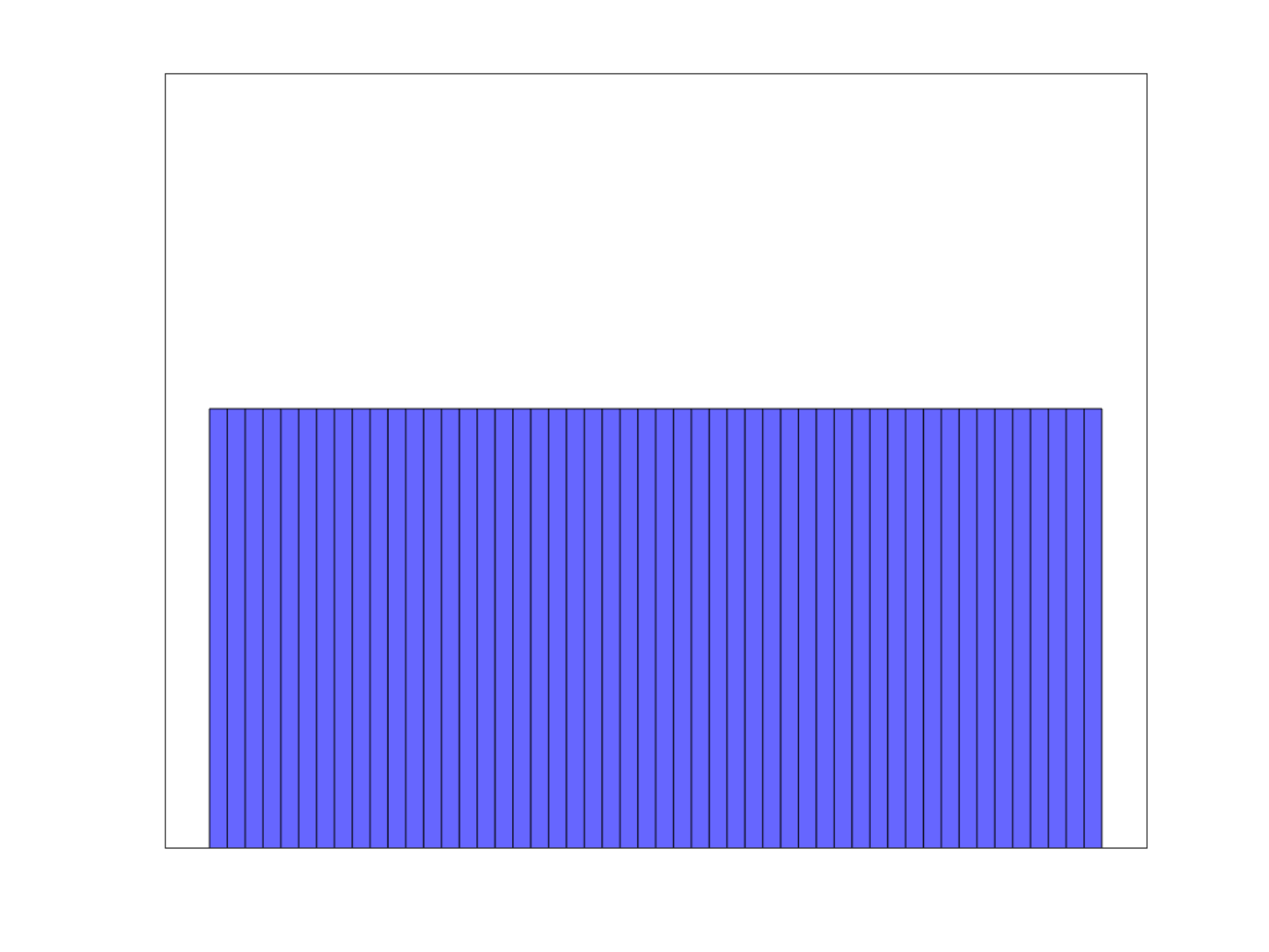} &  \includegraphics[width=.2 \textwidth]{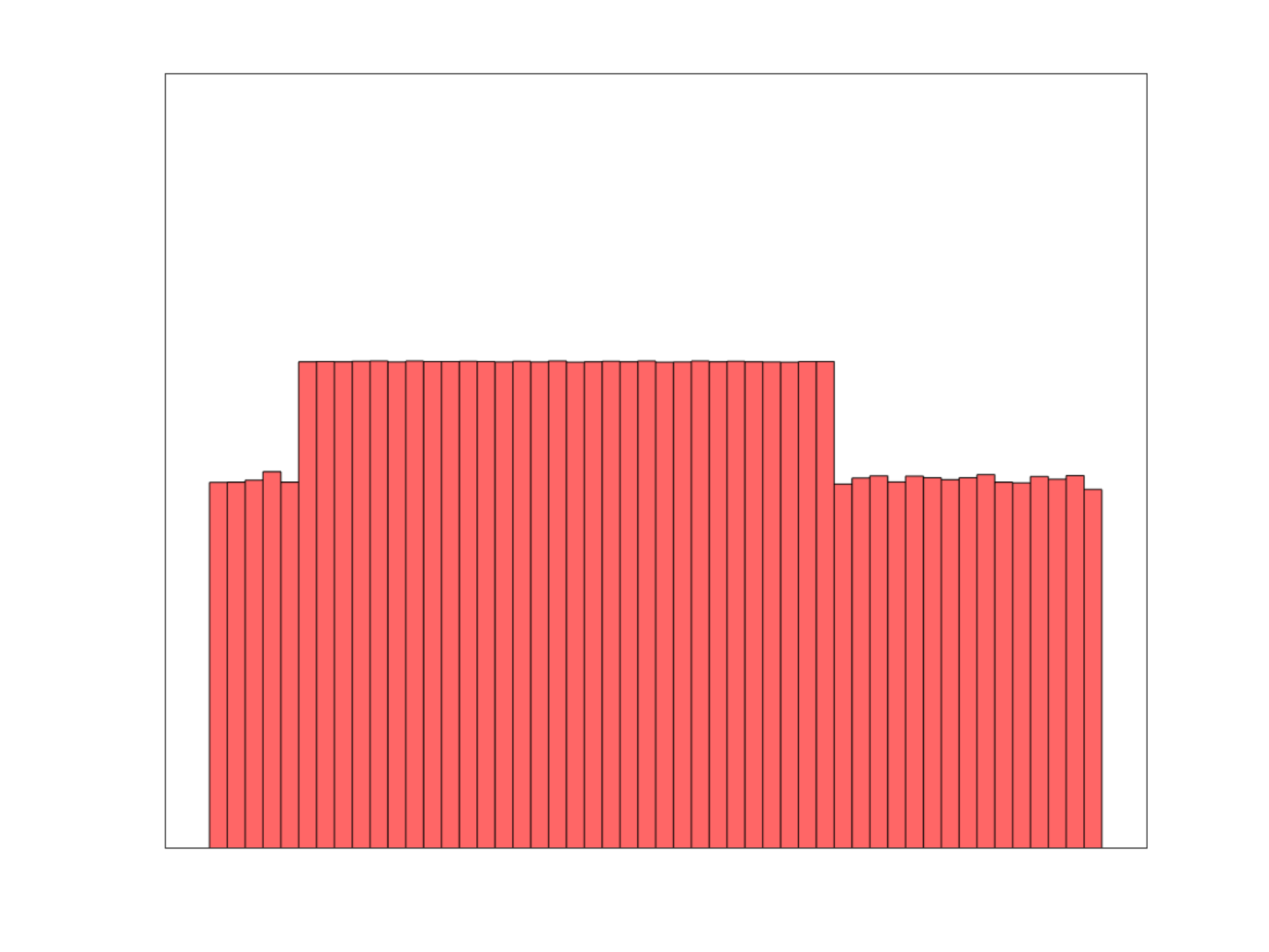} &  \includegraphics[width=.2 \textwidth]{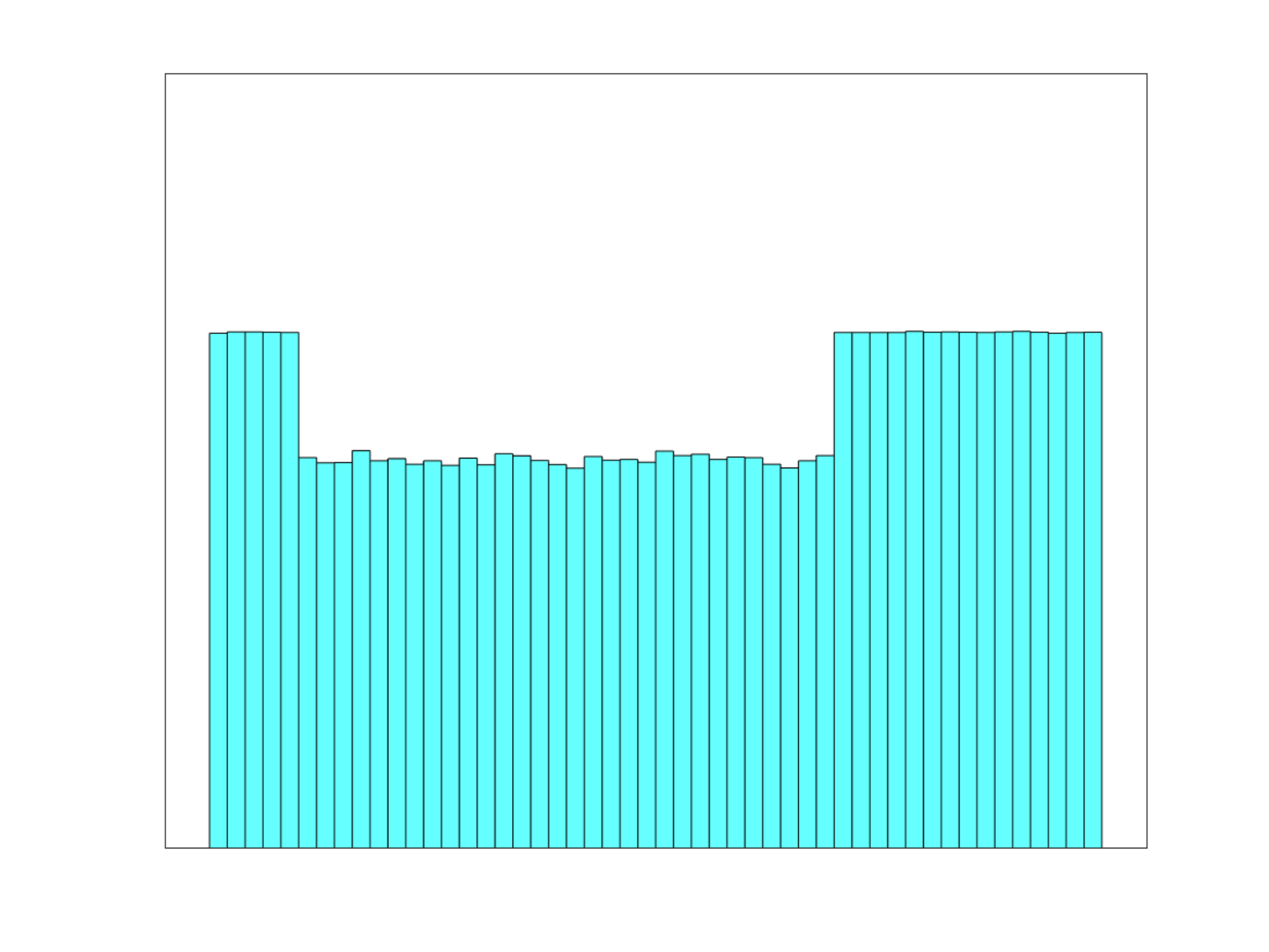} \\
$\beta$ & \includegraphics[width=.2 \textwidth]{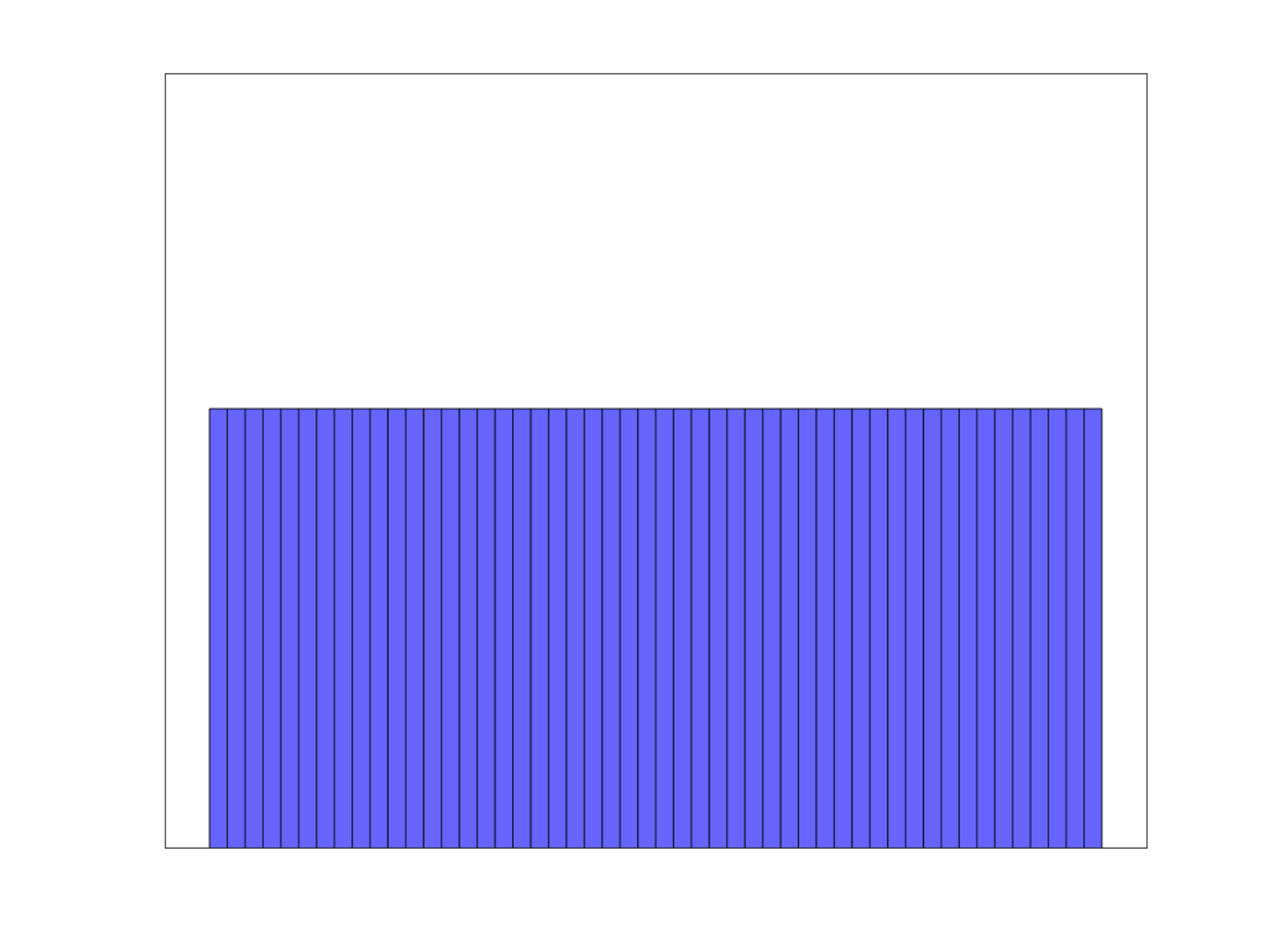} &  \includegraphics[width=.2 \textwidth]{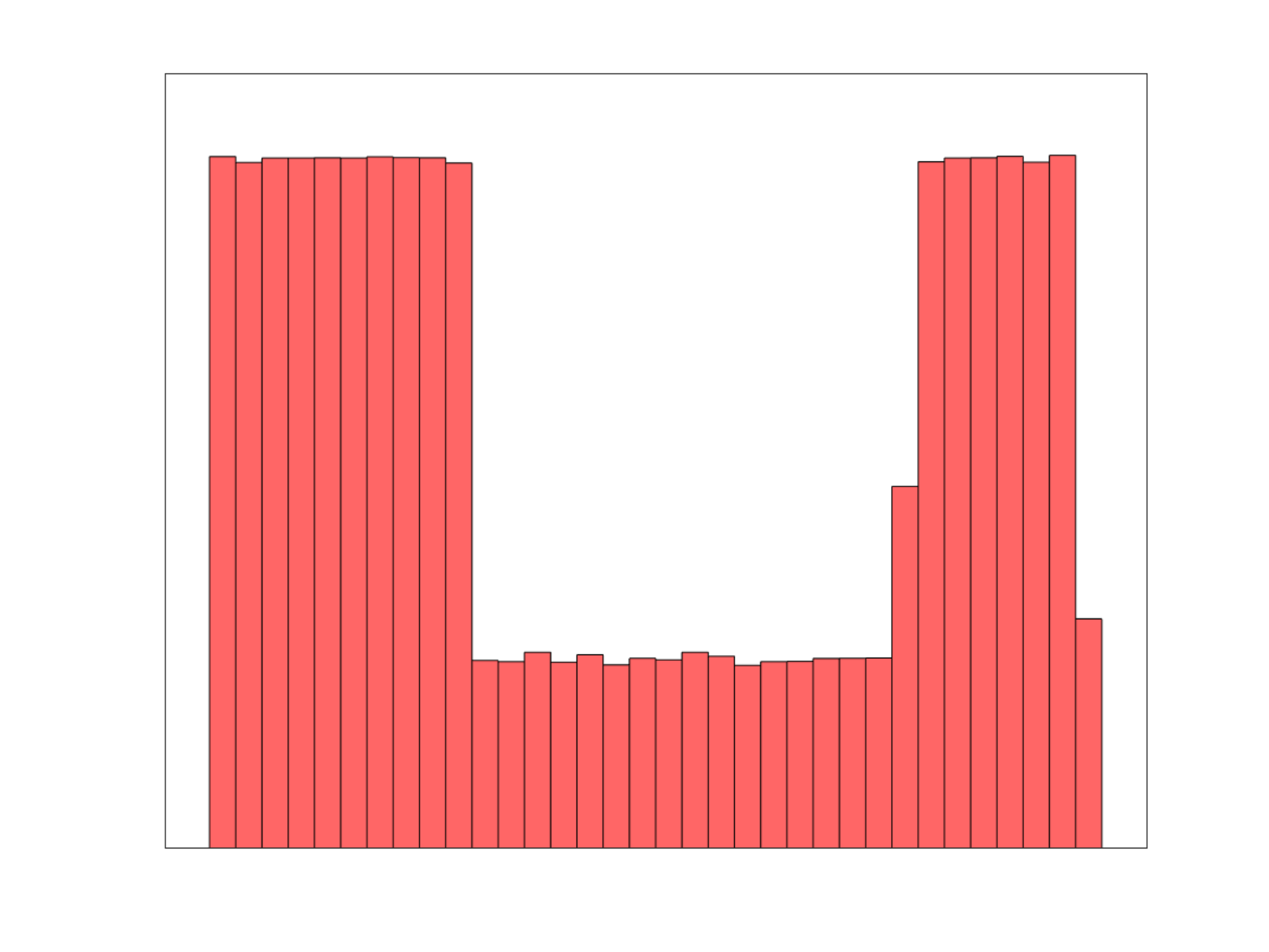} &  \includegraphics[width=.2 \textwidth]{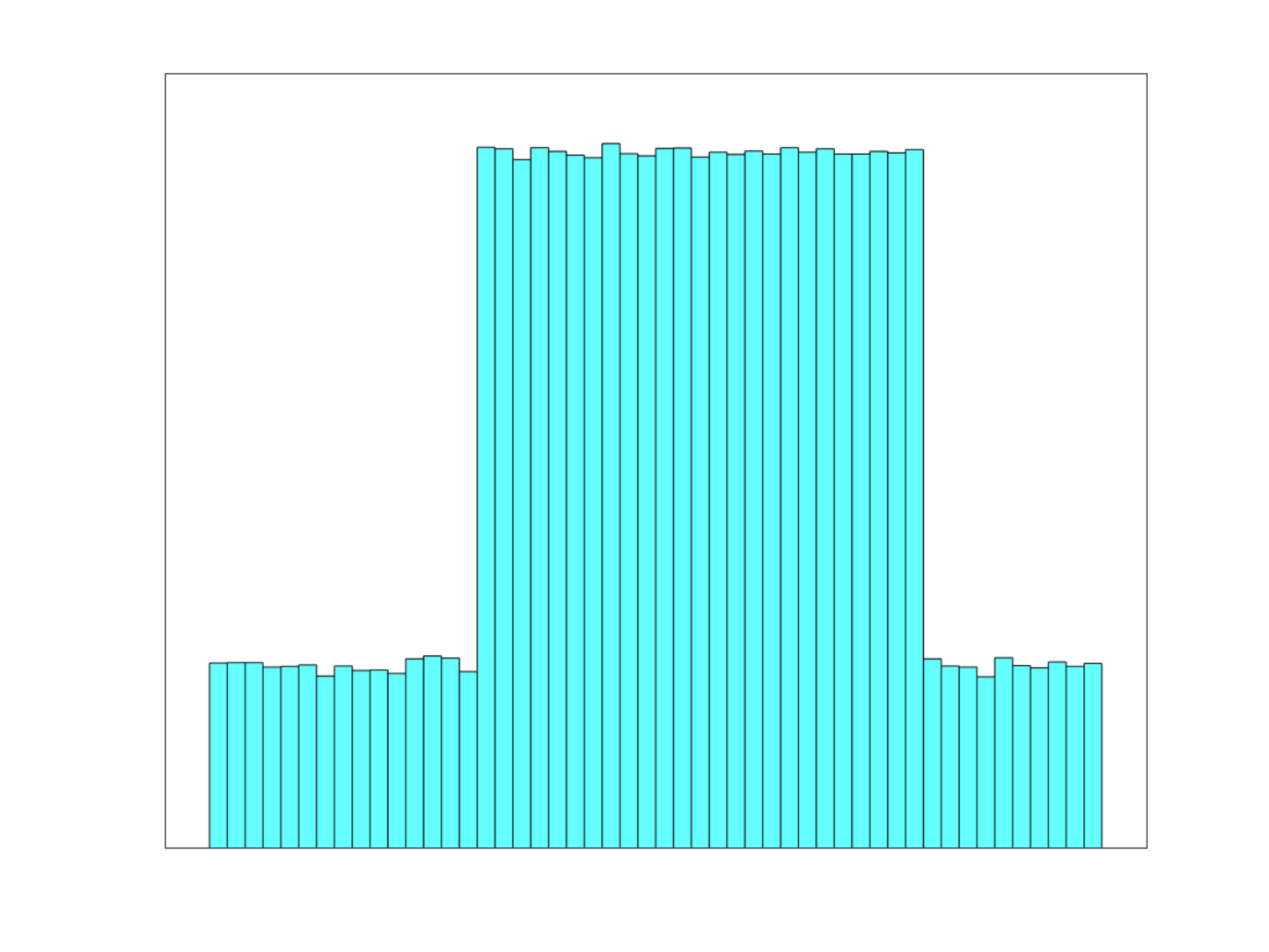} \\
$\nu_1$ & \includegraphics[width=.2 \textwidth]{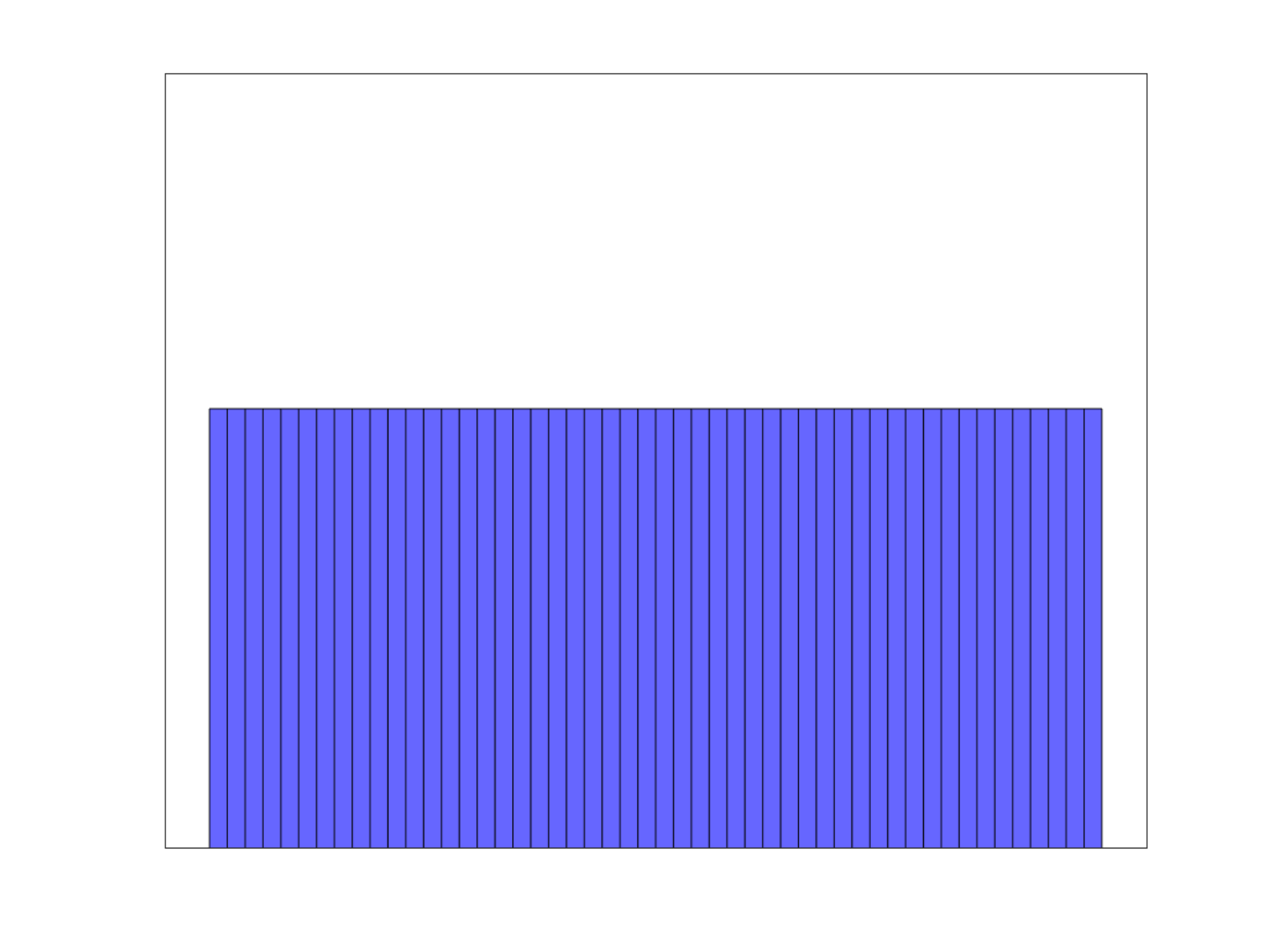} &  \includegraphics[width=.2 \textwidth]{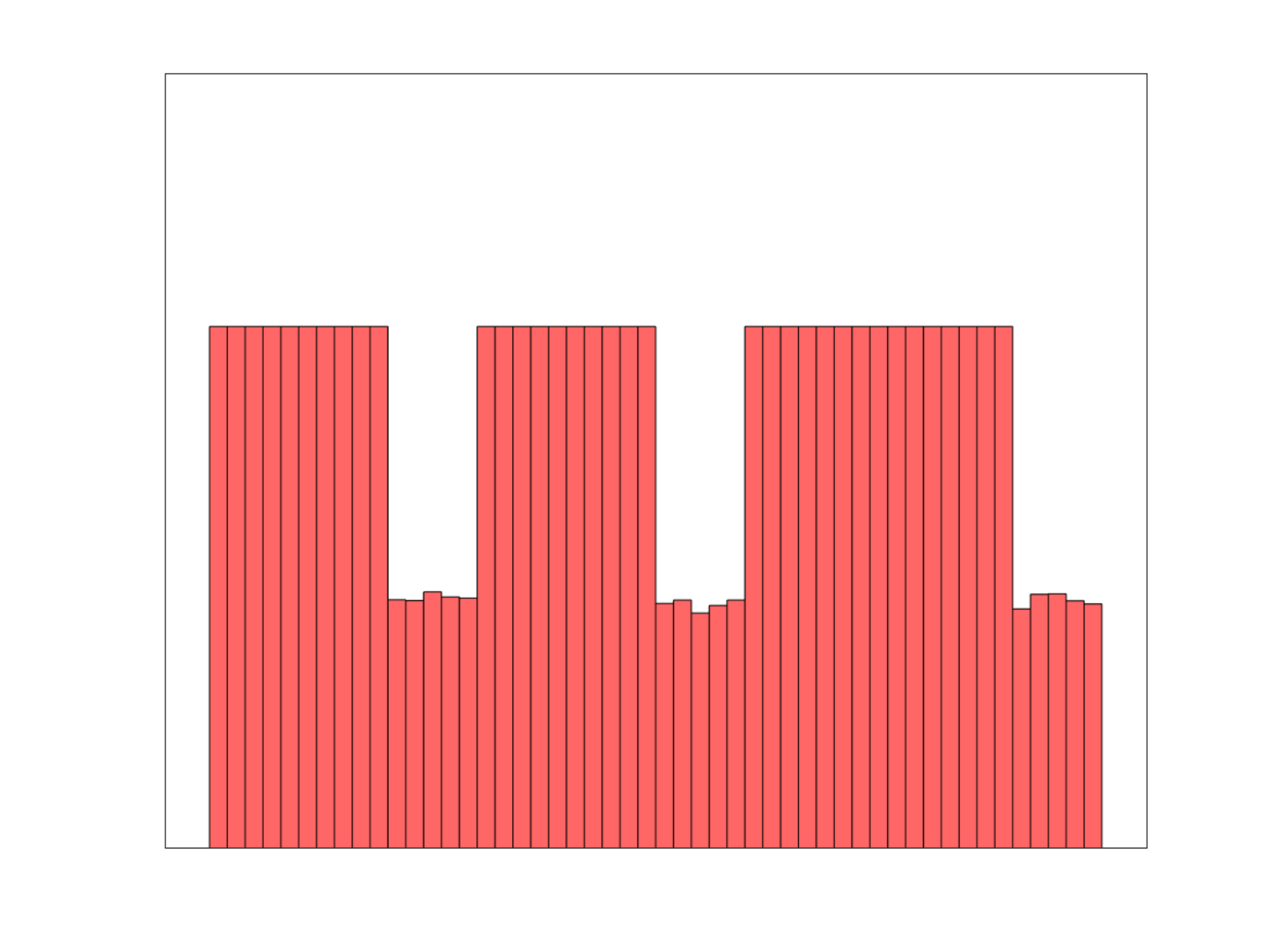} &  \includegraphics[width=.2 \textwidth]{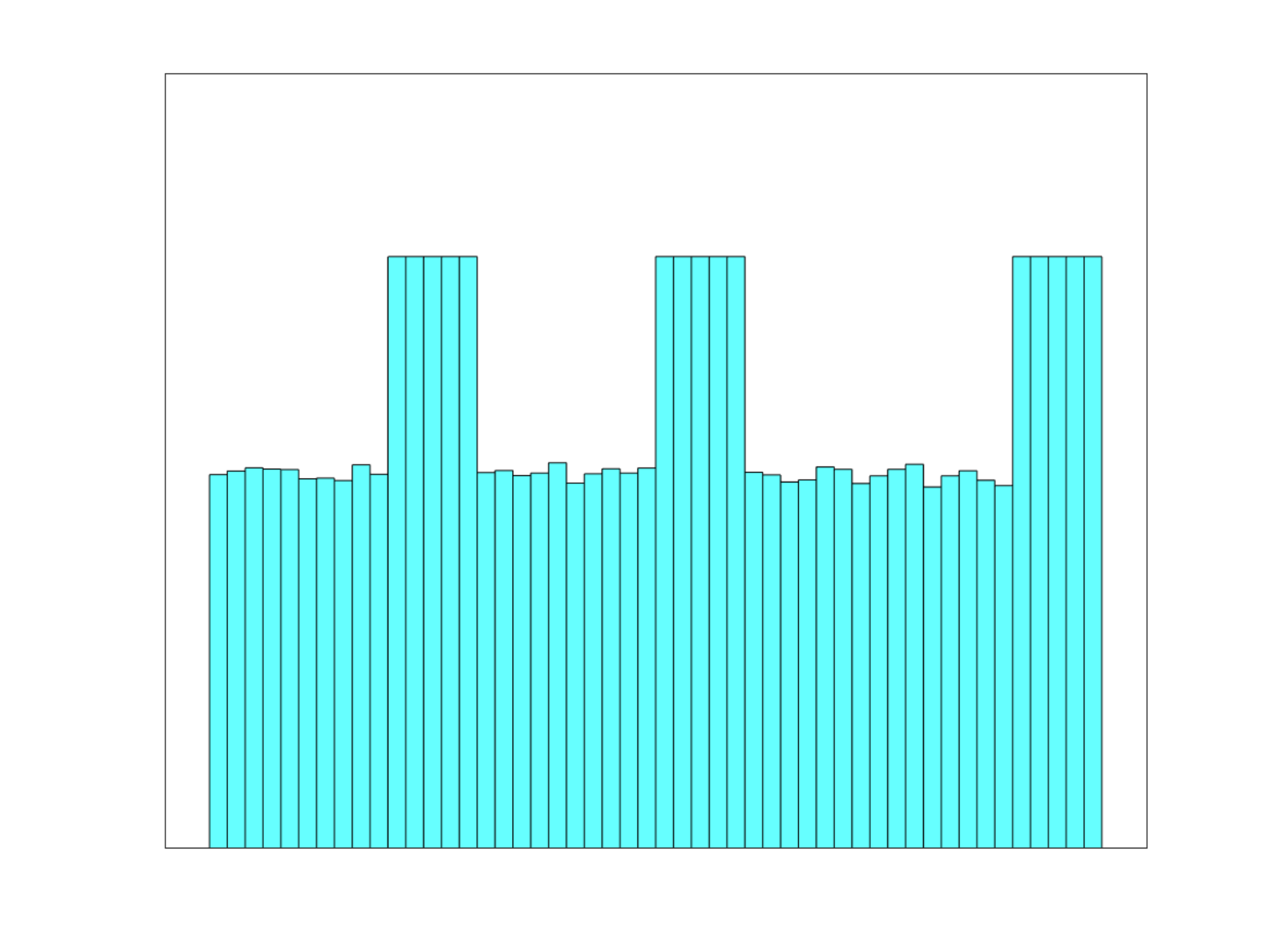} \\
$\nu_2$ & \includegraphics[width=.2 \textwidth]{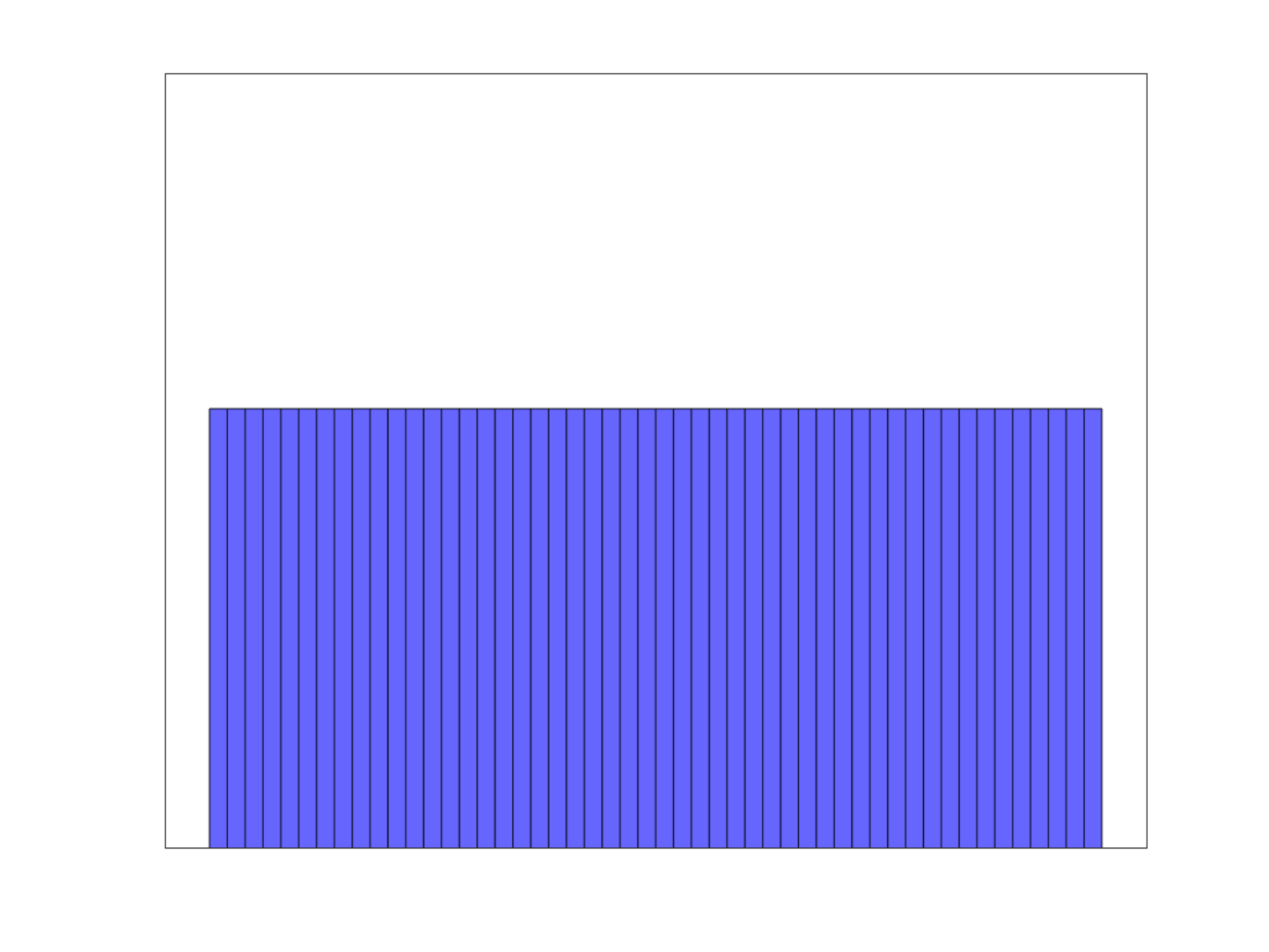} &  \includegraphics[width=.2 \textwidth]{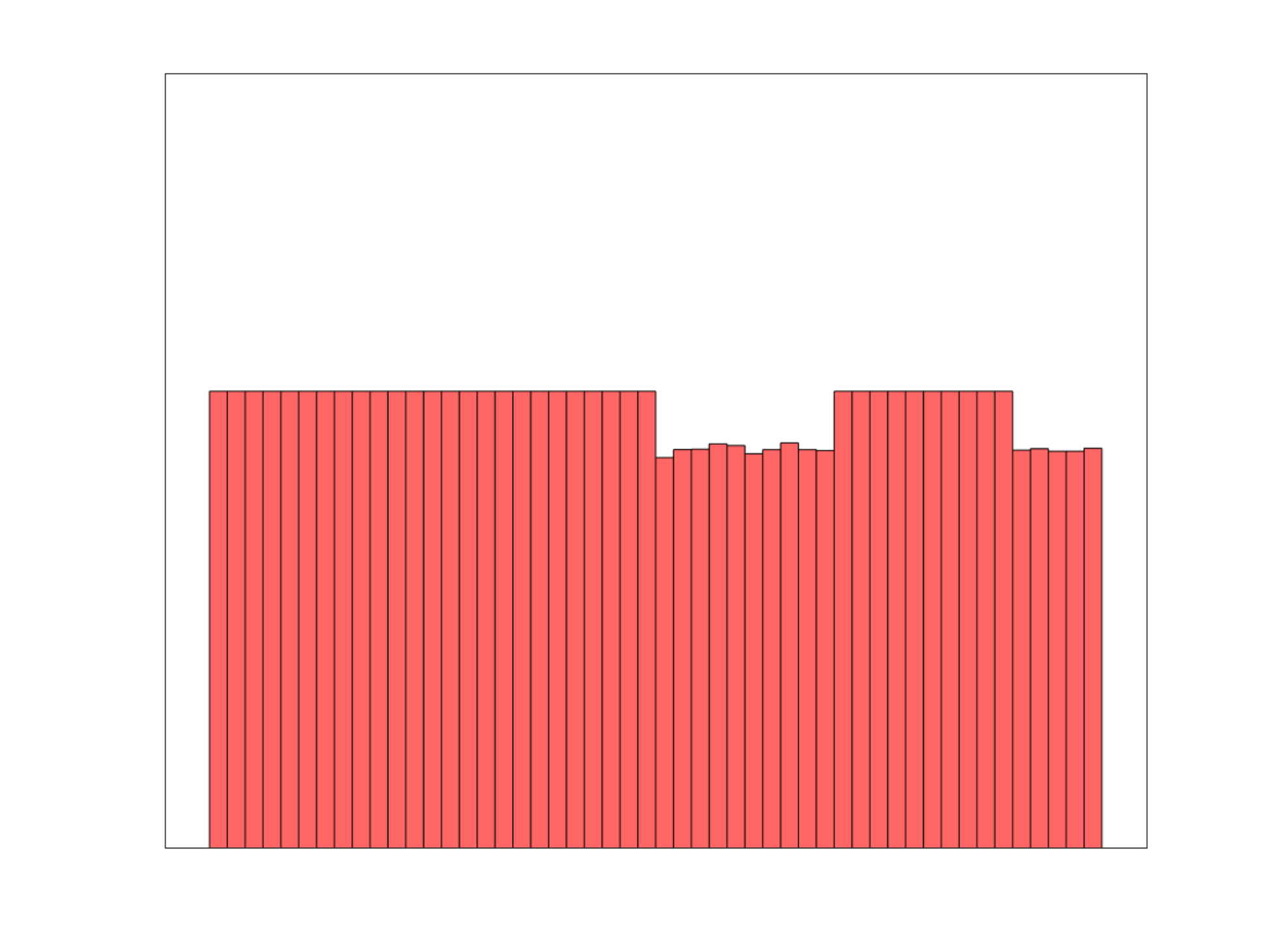} &  \includegraphics[width=.2 \textwidth]{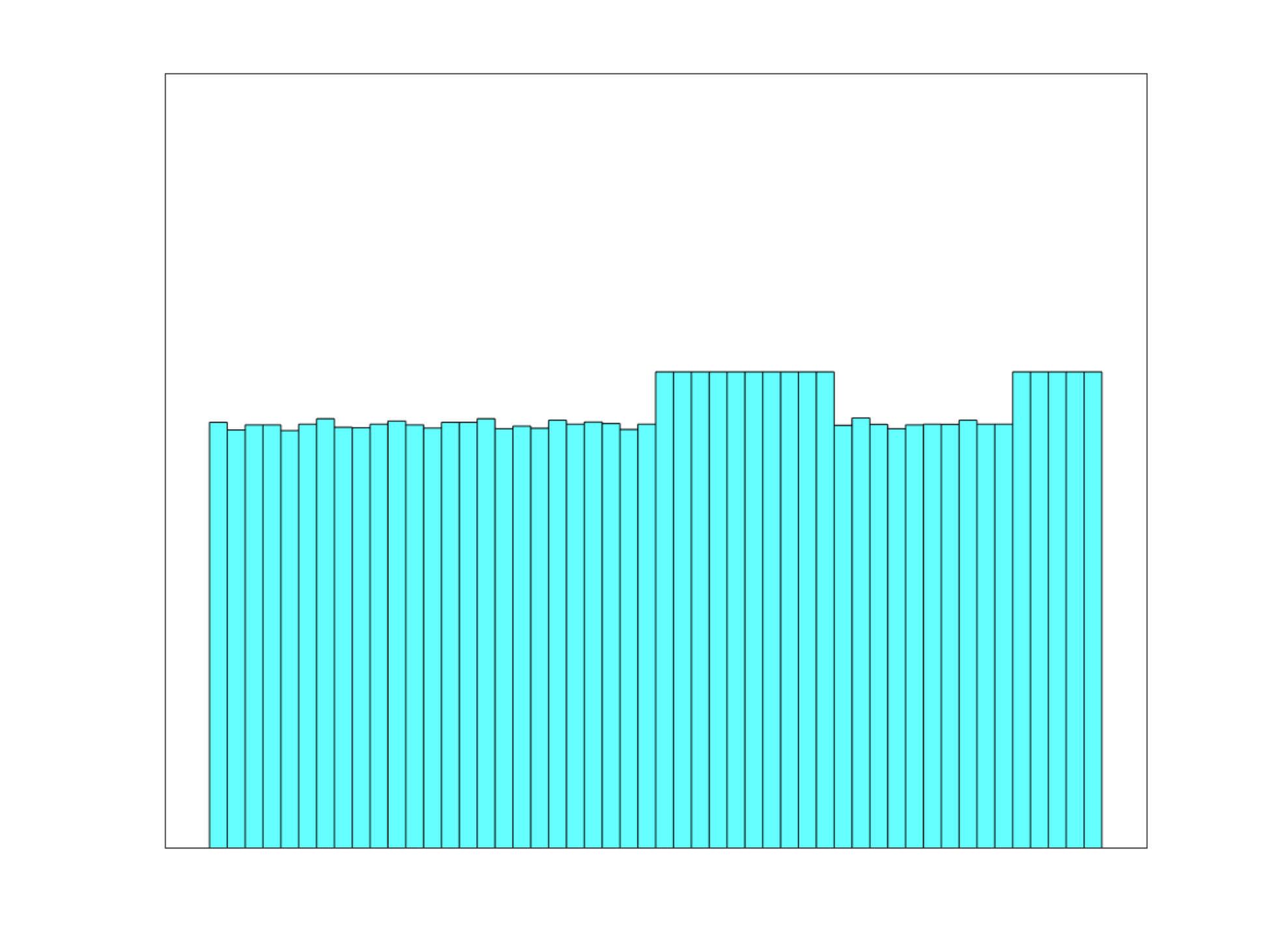} \\
\end{tabular}
\caption{Samples from the distribution of $\X$ and the perturbed distributions. The top, middle, and bottom rows correspond to variables $X_1,X_2$, and $X_3$ respectively. The left column correspond to samples from $\X$, the center column correspond to samples from the perturbed distribution which maximizes $T_2$, and the right column correspond to samples from the perturbed distribution which minimizes $T_2$.}
\label{fig:adv_diff_dist_pert_7}
\end{figure}

\section{Conclusion}
\label{sec:conclusion}
A novel method has been presented which considers perturbations of the one dimensional marginal PDFs to measure the robustness of Sobol' indices to distributional uncertainty. In its construction, this method assumes that the inputs are independent and the uncertainty is strictly in their one dimensional marginal distributions. This is a common occurrence in practice when, for instance, application specific information indicates that the inputs do not posses statistical dependences, but their statistical characterization is poorly known because of sparse and/or noisy data. In contrast, the authors previous work \cite{hart_robustness} provides a method to study robustness with respect to changes in the joint PDF, thus accounting for possible statistical dependencies. The source of distributional uncertainty is problem dependent and must be determined by the user. Since both methods are post processing steps after computing Sobol' indices, the user may apply both methods and compare their results. Future work may include exportation of other PDF perturbation forms such as block dependency structures, and extensions of the PDF perturbation approach to other global sensitivity analysis methods.

\section*{Software Availability}
Matlab codes are available at\\

\href{https://github.com/jlhart352/Sobol_Index_Robustness}{\texttt{https://github.com/jlhart352/Sobol\char`_Index\char`_Robustness}}\\

\noindent
which implement the joint perturbation approach of \cite{hart_robustness} and the marginal perturbation approach presented in this article.
\bibliographystyle{plain}
\bibliography{Uncertain_Marginals_Sobol_3}

\begin{thebibliography}{10}

\bibitem{beckman_mckay}
Richard~J. Beckman and Michael~D. McKay.
\newblock Monte {C}arlo estimation under different distributions using the same
  simulation.
\newblock {\em Technometrics}, 29(2):153--160, 1987.

\bibitem{anova_mult_dist}
Emanuele Borgonovo, Max~D. Morris, and Elmar Plischke.
\newblock Functional {ANOVA} with multiple distributions: Implications for the
  sensitivity analysis of computer experiments.
\newblock {\em SIAM/ASA J. Uncertain. Quantif.}, 6(1):397--427, 2018.

\bibitem{chick}
Stephen~E. Chick.
\newblock Input distribution selection for simulation experiments: Accounting
  for input uncertainty.
\newblock {\em Operations Research}, 49(5):744--758, 2001.

\bibitem{cousins}
Areski Cousin, Alexandre Janon, Veronique~Maume Deschamps, and Ibrahima Niang.
\newblock On the consistency of {S}obol' indices with respect to stochastic
  ordering of model parameters.
\newblock {\em https://hal.archives-ouvertes.fr/hal-01026373}, 2014.

\bibitem{chebfun}
Tobin~A. Driscoll, Nicholas Hale, and Lloyd~N. Trefethen, editors.
\newblock {\em Chebfun Guide}.
\newblock Pafnuty Publications, Oxford, 2014.

\bibitem{deep_uncertainty}
Lei Gao, Brett~A. Bryan, Martin Nolan, Jeffery~D. Connor, Xiaodong Song, and
  Gang Zhao.
\newblock Robust global sensitivity analysis under deep uncertainty via
  scenario analysis.
\newblock {\em Environmental Modelling \& Software Software}, 76:154--166,
  2016.

\bibitem{lca_correlations}
Evelyne~A Groena and Reinout Heijungs.
\newblock Ignoring correlation in uncertainty and sensitivity analysis in life
  cycle assessment: what is the risk?
\newblock {\em Environmental Impact Assessment Review}, 62:98--109, 2017.

\bibitem{hall}
Jim Hall.
\newblock Uncertainty-based sensitivity indices for imprecise probability
  distributions.
\newblock {\em Reliability Engineering and System Safety}, 91:1443--1451, 2006.

\bibitem{hart_corr_var}
Joseph Hart and Pierre Gremaud.
\newblock An approximation theoretic perspective of {S}obol' indices with
  dependent variables.
\newblock {\em {I}nternation {J}ournal for {U}ncertainty {Q}uantification},
  8(6):483--493, 2018.

\bibitem{hart_robustness}
Joseph Hart and Pierre Gremaud.
\newblock Robustness of the sobol' indices to distributional uncertainty.
\newblock {\em https://arxiv.org/abs/1803.11249. Under review}, 2018.

\bibitem{dice}
Zhaolin Hu, Jing Cao, and L.~Jeff Hong.
\newblock Robust simulation of global warming policies using the dice model.
\newblock {\em Management Science}, 58(12):2190--2206, 2012.

\bibitem{iooss}
Bertrand Iooss and Paul {Lema\^itre}.
\newblock A review on global analysis methods.
\newblock In G.~{Dellino} and C.~{Meloni}, editors, {\em Uncertainty management
  in simulation-optimization of complex systems}, chapter~5, pages 543--501.
  Springer, 2015.

\bibitem{kucherenko}
Sergei Kucherenko, Stefano Tarantola, and Paola Annoni.
\newblock Estimation of global sensitivity indices for models with dependent
  variables.
\newblock {\em Computer Physics Communications}, 183:937--946, 2012.

\bibitem{lca_app}
Martino Lacirignola, Philippe Blanc, Robin Girard, Paula P{\'e}rez-L{\'o}pez,
  and Isabelle Blanc.
\newblock {LCA} of emerging technologies: addressing high uncertainty on
  inputs' variability when performing global sensitivity analysis.
\newblock {\em Science of the Total Environment}, 578:268--280, 2017.

\bibitem{climate_app}
Antony Millner, Simon Dietz, and Geoffrey Heal.
\newblock Scientific ambiguity and climate policy.
\newblock {\em Environ Resource Econ}, 55:21--46, 2013.

\bibitem{WARM}
Livia Paleari and Roberto Confalonieri.
\newblock Sensitivity analysis of a sensitivity analysis: We are likely
  overlooking the impact of distributional assumptions.
\newblock {\em Ecological Modelling}, 340:57--63, 2016.

\bibitem{Sobol_UQ_handbook}
Cl\'ementine Prieur and Stefano Tarantola.
\newblock Variance-based sensitivity analysis: Theory and estimation
  algorithms.
\newblock In Roger Ghanem, David Higdon, and Houman Owhadi, editors, {\em
  Handbook of Uncertainty Quantification}. Springer, 2017.

\bibitem{saltelli2010}
Andrea Saltelli, Paola Annoni, Ivano Azzini, Francesca Campolongo, Marco Ratto,
  and Stefano Tarantola.
\newblock Variance based sensitivity analysis of model output. design and
  estimator for the total sensitivity index.
\newblock {\em Computer Physics Communications}, 181:259--270, 2010.

\bibitem{sobol93}
Ilya~M. Sobol'.
\newblock Sensitivity estimates for non linear mathematical models.
\newblock {\em Math. Mod. Comp. Exp.}, 1:407--414, 1993.

\bibitem{sobol}
Ilya~M. Sobol'.
\newblock Global sensitivity indices for nonlinear mathematical models and
  their {Monte Carlo} estimates.
\newblock {\em Mathematics and Computers in Simulation}, 55:271--280, 2001.

\bibitem{sobol2003}
Ilya~M. Sobol'.
\newblock Theorems and examples on high dimensional model representation.
\newblock {\em Reliability Eng. Sys. Safety}, 79:187--193, 2003.

\bibitem{fast_dependent_variables}
Stefano Tarantola and Thierry~A. Mara.
\newblock Variance-based sensitivity indices of computer models with dependent
  inputs: {T}he {F}ourier {A}mplitude {S}ensitivity {T}est.
\newblock {\em International Journal for Uncertainty Quantication},
  7(6):511--523, 2017.

\end{thebibliography}

\end{document}